\documentclass{lmcs}
\pdfoutput=1

\usepackage{lastpage}
\lmcsdoi{20}{1}{5}
\lmcsheading{}{\pageref{LastPage}}{}{}%
{Jul.~18,~2022}{Jan.~23,~2024}{}


\usepackage{mathpartir}
\usepackage{listings}
\usepackage[normalem]{ulem}
\usepackage{amssymb}
\usepackage{tikz}
\usetikzlibrary{calc,cd,decorations.pathmorphing}

\newcommand \termsName[1] {\mathtt{T}_{#1}}
\newcommand \redName[1] {\mathtt{#1}}
\newcommand \ctxName[1] {\mathtt{#1}}

\newcommand \rterms {\termsName R}
\newcommand \pterms {\termsName P}

\newcommand \llset {\texttt T} 
\newcommand \purelterms {\texttt U} 

\newcommand \rewrel {\mathcal R}
\newcommand \redrule {\redName r}

\newcommand \lcalc {\lambda}
\newcommand \whl {\redName {whr}}

\newcommand \db {\redName{dB}}

\newcommand \acalc {\lambda a}

\newcommand \rcalc {\lambda R}
\newcommand \apprr {\redName{app}}
\newcommand \absrr {\redName{abs}}
\newcommand \varrr {\redName{var}}
\newcommand \distrr {\redName{dist}}
\newcommand \permlr {\pi}
\newcommand \ndlrdb {\redName {dB}}
\newcommand \ndlrdist {\redName {spl}} 
\newcommand \ndlrabs {\redName {ls}} 

\newcommand \nrsub {\redName{sub}} 
\newcommand \nrep {\redName R} 
\newcommand \whlr {\redName {name}} 
\newcommand \whdblr {\redName {n}\db} 
\newcommand \whslr {\redName {n}\nrsub} 
\newcommand \skelss {\redName {st}} 
\newcommand \ndlr {\redName {flneed}} 
\newcommand \purelr[1] {#1'}
\newcommand \factorise {\redName F}
\newcommand \whlrdbroot {(\textsc{db})}
\newcommand \whlrdbapp {(\textsc{appdb})}
\newcommand \whlrdbsub {(\textsc{subdb})}
\newcommand \whlrsubroot {(\textsc{s})}
\newcommand \whlrsubapp {(\textsc{apps})}
\newcommand \whlrsubsub {(\textsc{subs})}
\newcommand \skelsstop {(\textsc{ctx1})}
\newcommand \skelssdeep {(\textsc{ctx2})}

\newcommand \ec {\Diamond} 
\newcommand \fc {\ctxName C} 
\newcommand \lc {\ctxName L} 
\newcommand \hlc {\verdear{\ctxName L}} 
\newcommand \llc {\ctxName{LL}} 
\newcommand \ndc {\ctxName N} 

\newcommand \id {\mathtt I}

\newcommand \MM {\mathcal M}

\newcommand \MN {\mathcal N}
\newcommand \ans {\mathtt a}

\newcommand \sysWR {\cap R}

\newcommand \many {(\textsc{many})}

\newcommand \ruleAxR {(\textsc{ax})}
\newcommand \ruleAbsR {(\textsc{abs})}
\newcommand \ruleAnsR {(\textsc{ans})}
\newcommand \ruleAppR {(\textsc{app})}
\newcommand \ruleCutR {(\textsc{cut})}

\newcommand \lx {\lambda x}
\newcommand \ly {\lambda y}
\newcommand \lz {\lambda z}
\newcommand \app[2] {#1\!\cdot\!#2}

\newcommand{\letb}[3]{\text{\lstinline!let!}\ #1\ \text{\lstinline!=!}\ #2\ \text{\lstinline!in!}\ #3}

\newcommand \msub[2] {\{#1/#2\}}
\newcommand \esub[2] {[#1/#2]}
\newcommand \dist[2] {[#1/\!\!/#2]}
\newcommand \cut[2] {[#1 \triangleleft #2]}

\newcommand \inter {\uplus}
\newcommand \emult {\mult{\, }}
\newcommand \mult[1] {[#1]}
\newcommand \multunion {\sqcup}

\newcommand{\seq}[3]{\ensuremath{#1 \vdash #2:#3}}
\newcommand{\der}[3]{\ensuremath{#1 \vdash #2:#3}}

\newcommand{\dem}[4]{#1 \rhd \der{#2}{#3}{#4}}

\newcommand \rrule[1] {\mapsto_{#1}}
\newcommand \rew[1] {\rightarrow_{#1}}
\newcommand \rewp[1] {\twoheadrightarrow_{#1}^+}
\newcommand \rewn[1] {\twoheadrightarrow_{#1}}

\newcommand \skelbs[1] {\Downarrow^{#1}}
\newcommand \funcrel \Downarrow

\newcommand \obseq \equiv

\newcommand \fv[1] {\operatorname{fv}(#1)} 
\newcommand \ndv[1] {\operatorname{ndv}(#1)} 
\newcommand \ctx[2] {#1 \langle  #2 \rangle} 
\newcommand \ctxnc[2] {#1 \langle\!\langle  #2 \rangle\!\rangle} 
\newcommand \lv[2] {\operatorname{lv}_{#1}(#2)} 
\newcommand \isES[1] {\operatorname{ES}(#1)}
\newcommand \tmsz[1] {|#1|} 
\newcommand \nbocc[2] {\tmsz{#2}_{#1}} 
\newcommand \domlc[1] {\operatorname{dom}(#1)} 

\newcommand \dom[1] {\operatorname{dom}(#1)}
\newcommand \cmin[2] {#1 {\setminus\!\!\setminus}\, #2}
\newcommand \dersz[1] {\operatorname{sz}({#1})} 
\newcommand \dermeas[1] {\mathsf{D}\left(#1\right)} 
\newcommand \dermeasaux[2] {\mathsf{M}\left(#1, #2\right)} 
\newcommand \msetsz[1] {|#1|} 

\newcommand \maprl[1] {#1^\downarrow}

\newcommand \skel[2] {\{ \!\{ {#2}\}\!\}^{#1}}
\newcommand \askel[2] {\{\!\{\!\{{#2}\}\!\}\!\}^{#1}}

\newcommand \eqdef {\coloneqq}
\newcommand \tuple[1] {\langle #1 \rangle}
\newcommand \pair[2] {\tuple{#1, #2}}
\newcommand \omult[1] {[#1]}

\newcommand \deft[1] {\textbf{#1}}
\newcommand \gramTitle[1] {\mbox{(\textbf{#1})}}

\newcommand \subobj {\mathcal O}
\DeclareMathOperator \subesname {a}
\DeclareMathOperator \subdistname {b}
\newcommand \subes[2] {\subesname(#1,#2)}
\newcommand \subdist[1] {\subdistname(#1)}
\newcommand \subge {>^\subobj}
\newcommand \subgeq {\geq^\subobj}
\newcommand \subgemult {>^\subobj_{\texttt{MUL}}} 
\newcommand \subgeqmult {\geq^\subobj_{\texttt{MUL}}}
\newcommand \submeas[1] {\mathsf{CL}\left(#1\right)}
\newcommand \subgelex {>_{\texttt{LEX}}}
\DeclareMathOperator \subproj P
\newcommand \subka {\mathbb L}
\newcommand \suben {\mathbb S}
\DeclareMathOperator \level {fst}
\DeclareMathOperator \snd {snd}

\newcommand{\ih}{\textit{i.h.}}
\newcommand{\ie}{\textit{i.e.}}

\newcommand{\eg}{\textit{e.g.}}

\newcommand{\seeappendix}[1]{(see appendix on \autopageref{#1})}

\newcommand \iI {{i \in I}}
\newcommand \jJ {{j \in J}}
\newcommand \jJi {{j \in J_i}}

\newcommand \multii[1] {\mult{#1}_{\iI}}

\newcommand \interii {\inter_\iI}

\newcommand \Del \Delta
\newcommand \Gam \Gamma
\newcommand \sig \sigma
\newcommand \lam[2] {\lambda #1.#2}
\newcommand \es[2] {[#1/#2]}

\newcommand \ft \to
\newcommand{\neutg}{\overline{\texttt{Ne}}}
\newcommand{\normg}{\texttt{Ne}}
\newcommand{\nmneutg}{\overline{\texttt{Na}}}
\newcommand{\nmnormg}{\texttt{Na}}

\definecolor{LightGray}{gray}{.10}

\definecolor{verdelight}{RGB}{171,232,194}
\newcommand{\verdear}[1]{\colorbox{verdelight}{\ensuremath{#1}}}
\definecolor{naranjalight}{RGB}{255,204,153}
\newcommand{\anaranjear}[1]{\colorbox{naranjalight}{\ensuremath{#1}}}
\definecolor{celeste}{RGB}{178,255,255}
\newcommand{\celeastar}[1]{\colorbox{celeste}{\ensuremath{#1}}}
\newcommand{\sep}{\hspace{.5cm}}

\begin{document}

\title{Node Replication: Theory and Practice}
\author[D.~Kesner]{Delia Kesner}[a,c]
\author[L.~Peyrot]{Loïc Peyrot}[a]
\author[D.~Ventura]{Daniel Ventura}[b]

\address{Université Paris Cité, CNRS, IRIF, Paris, France}
\email{kesner@irif.fr, lpeyrot@irif.fr}

\address{Univ. Federal de Goiás, INF, Goiânia, Brazil}
\email{ventura@ufg.br}

\address{Institut Universitaire de France}

\begin{abstract}
  We define and study a term calculus implementing higher-order node replication.
It is used to specify two different (weak) evaluation strategies: call-by-name and fully
lazy call-by-need, that are shown to be observationally equivalent
by using type theoretical technical tools.
\end{abstract}

\maketitle

\section{Introduction}

Computation in the $\lambda$-calculus is based on higher-order
substitution, a complex operation being able to erase and copy terms
during evaluation.  Several formalisms have been proposed to model
higher-order substitution, going from explicit substitutions
(ES)~\cite{abadi90} (see a survey in~\cite{kesner09}) and labeled
systems~\cite{LevyPhD,blanc05} to pointer graphs~\cite{wadsworth71} or
optimal sharing graphs~\cite{lamping90}. The operational
  semantics used in each formalism to implement copying
  (\ie\ duplication) of terms  is not the same.

Meanwhile, the Curry-Howard isomorphism~\cite{soerensen06} uncovers a
deep connection between logical systems and term calculi. Also in this
framework, different ways to normalize proofs in different logical
systems correspond to different implementations of substitution.
Indeed, the process of \emph{full substitution} is somehow
  analogous to the normalization process in natural deduction, while
\emph{linear substitution} corresponds to cut elimination in
Proof-Nets~\cite{PNESAccattoli18}.  \emph{Node-by-node replication} is
based on a Curry-Howard interpretation of deep
inference~\cite{guglielmi10,gundersen13a}.  Let us illustrate these
different models using an example.

Indeed, suppose one wants to substitute all the free occurrences of some
variable $x$ in a term $t= \app x x$ by the term $u =\app y {(\app z
  w)}$ ($\app \_ \_$ denotes the
application constructor).
The variable substituted at each reduction step  is going to be highlighted in \celeastar{\mbox{light blue}}:
\begin{align}
  (\app {\celeastar x}{\celeastar x}) \esub x u
  &\rew{} \app u u \label{ex:fsub}\\
  (\app{\celeastar x} x) \esub x u
  &\rew{} (\app u {\celeastar x}) \esub x u
  \rew{} \app u u \label{ex:psub}\\
  (\app{\celeastar x}{\celeastar x}) \esub x {\app y {(\app w z)}}
  &\rew{} (\app{(\app y {\celeastar{x'}})}{(\app y {\celeastar{x'}})})
  \esub {x'} {\app w z}
  \rew{} \app u u \label{ex:fnrsub}
\end{align}

Full (or \emph{non-linear}) substitution \eqref{ex:fsub} is a
one-shot substitution, replacing simultaneously \emph{all} the free occurrences of $x$ in $t$ by the whole term $u$.
This notion is generally defined by induction on the structure of the term $t$.
Linear (or \emph{partial})
substitution \eqref{ex:psub} replaces \emph{one} free occurrence of $x$ at a time.  This
notion is generally defined by induction on the number of free occurrences of
$x$ in the term $t$.
Node replication \eqref{ex:fnrsub} is yet another approach
which replicates \emph{one} term-constructor of $u$
\emph{at a time}, instead of replicating $u$ as a whole.
This notion can be defined by
induction on the structure of the term $u$. 
In the Curry-Howard interpretation of deep inference, node replication is full:
all occurences of $x$ are replaced by a node of $u$. A
linear version of the node replication approach can be formally defined by
combining the last two models, although this does not
correspond to any logical system we are aware of:
\begin{equation}
  \label{ex:pnrsb}%
  \begin{alignedat}{1}
  (\app{\celeastar x} x) \esub x u
  &\rew{} (\app{(\app y {x'})}{\celeastar x}) \esub x {\app y {x'}} \esub {x'} {\app w z}
  \rew{} (\app{(\app y {\celeastar{x'}})}{(\app y {x'})}) \esub {x'}{\app w z}\\
  &\rew{} (\app{(\app y {(\app w z)})}{(\app y {\celeastar{x'}})}) \esub {x'}{\app w z}
  \rew{} \app{(\app y {(\app w z)})}{(\app y {(\app w z)})}
  = \app u u
  \end{alignedat}
\end{equation}
In this sense, node
  replication is orthogonal to the full/linear aspect of
  \emph{classical} substitution, since it can be 
  implemented by either full or linear operations, as explained above.
Unsurprisingly, different notions of
substitution give rise to different evaluation strategies.
Indeed, linear substitution is the common model in
well-known abstract machines for call-by-name and
call-by-value~(see \eg\ \cite{accattoli14b}), while (linear) node replication
can be used to implement fully
lazy sharing~\cite{wadsworth71}.  However, node replication, originally
introduced to implement optimal graph
reduction in a graphical formalism, has only
been studied from a Curry-Howard perspective by means of a term
language known as the atomic
$\lambda$-calculus~\cite{gundersen13a}.

\paragraph{The Atomic Lambda-Calculus.}
Logical aspects of intuistionistic deep inference are captured by the
atomic $\lambda$-calculus $\acalc$~\cite{gundersen13a}, where copying
of terms proceeds \emph{atomically}, \ie\ node by node, similar to the
optimal graph reduction of Lamping~\cite{lamping90}. The atomic
$\lambda$-calculus is based on \emph{explicit control of resources},
such as erasure and duplication. Its operational semantics explicitly
handles the structural constructors of weakening and contraction, as
in the calculus of resources $\lambda {\tt
  lxr}$~\cite{kesner07,kesner11}.  As a result,
understanding the meta-properties of the
term-calculus, in a higher-level, and its application to concrete
implementations of reduction strategies in programming languages, turn
out to be quite difficult.  In this paper, we take one step back, by
studying the paradigm of \emph{node replication} based on
\emph{implicit}, rather than \emph{explicit}, weakening and
contraction.  This gives a new concise formulation of node replication
which is simple enough to model different programming languages based
on reduction strategies.

\paragraph{Call-by-Name, Call-by-Value, Call-by-Need.}

The theory of programming is usually based on some notion of
\emph{calculus}, which can be often seen as a higher-order rewriting
system. A calculus
  definition results in general in a
non-deterministic and unrestricted relation: a term can be reduced in
different ways, some of them being more efficent than others.  When
implementing  programming languages based on these theories, a
specific (deterministic) \emph{reduction strategy} is necessary to provide a
concrete mechanism to evaluate a program. 

In the specific case of functional programming, the theory is given by the $\lambda$-calculus, giving rise to different reduction strategies used by different functional programming languages.

\emph{Call-by-name} is used to implement programming languages in which
arguments of functions are first copied, then evaluated.  This is frequently
expensive, and may be improved by \emph{call-by-value}, in which arguments are
evaluated first, then consumed.
The difference can be illustrated by the term $t=\Delta(\id \id)$, where
$\Delta=\lx.xx$ and $\id=\lz.z$: call-by-name first duplicates the
argument $\id \id$, so that its evaluation is also duplicated, while
call-by-value first reduces $\id \id$ to (the value) $\id$, so that duplications
of the argument do not cause any duplicated evaluation.  It is not always the
best solution, though, because evaluating erasable arguments is useless.
Compare for instance $(\lx.z)(\id\id) \rew\beta (\lx.z) \id
\rew\beta z$ to $(\lx.z)(\id\id) \rew\beta z$.

\emph{Call-by-need}, instead, combines the best of call-by-name and
call-by-value: as in call-by-name, erasable arguments are not evaluated at all,
and as in call-by-value, reduction of arguments occurs at most once.
Furthermore, call-by-need implements a
\emph{demand-driven} evaluation, in which erasable arguments are never
needed (so they are not evaluated), and non-erasable arguments are
evaluated only if needed.  Technically, some sharing mechanism is
necessary, for example by extending the $\lambda$-calculus with explicit
substitutions/let constructs~\cite{ariola97}. Then, any $\beta$-reduction
step can be
decomposed in at least two steps: one creating an explicit (pending)
substitution, and the other ones (linearly)
substituting \emph{values}.  Thus for example,
$(\lx.xx)(\id\id)$ reduces to  $ (xx) \es x {\id\id}$,
and the substitution argument is thus evaluated in order to find  a value before
performing the linear substitution.

Even when adopting this wise evaluation scheme, there are still some
unnecessary copies of redexes: while only \emph{values}
(\ie\ abstractions) are duplicated, they may contain redexes as subterms,
\eg\ $\lz.z(\id\id)$ whose
subterm $\id \id$ is a redex. Duplication of such values
might cause redex duplications in \emph{weak} (\ie\ when evaluation is forbidden inside abstractions) call-by-need.
This happens in particular in the \emph{confluent} variant of weak reduction in~\cite{levy99}:
\[ \begin{array}{llllll}
  (\lx.xx)(\lz.z(\id\id))
  &\rew{} & (xx) \es x {\lz.z(\id\id)} 
  & \rew{} & ((\lz.z(\id\id)) x )\es x {\lz.z(\id\id)}
                     & \\
  & \rew{} & (z(\id\id)) \es z x \es x {\lz.z(\id\id)} & \rew{}
                                 & (x(\id\id)) \es x {\lz.z(\id\id)} & \\
  &  \rewn{} & {(\id\id)(\id\id) } & \rewn{} & \id &  \\
\end{array} \]

\paragraph{Full laziness.}
Alas, it is not possible to keep all values shared forever,
typically when they potentially contribute to the creation of a future
$\beta$-reduction step. The key idea to gain efficiency is then
to keep the subterm $\id\id$ as a \emph{shared} redex.
Therefore, the  value $\lz.z(\id\id)$
to be copied is split into two separate parts. The  first one, called
\emph{skeleton}, contains the minimal information
preserving the
bound structure of the value, \ie\ the linked structure between the binder and
each of its bound variables.  In our example, this is the term $\lz.zy$, where $y$ is a fresh variable. The second one is a multiset
of \emph{maximal free expressions} (MFE), representing all the
shareable expressions (here only the term $\id\id$).
Only the skeleton is then copied, while
the problematic redex $\id \id$ remains shared:
\[ (\lx.xx)(\lz.z(\id\id))
\rew{} (xx) \es x {\lz.z(\id\id)} \rew{} ((\lz.zy) x) \es
x {\lz.z y} \es y {\id \id} \]
When the subterm $\id\id$ is
needed ahead, it is first reduced inside the explicit substitution,
as it is usual in  call-by-need, thus
avoiding to compute the redex twice.
This optimization is called \emph{fully lazy sharing}
and is due to Wadsworth~\cite{wadsworth71}.

As shown by Balabonski~\cite{balabonski13}, any weak strategy
  of the $\lambda$-calculus which is not using sharing is undecidable,
  as for example the \textit{innermost needed reduction}. However, in
the confluent weak setting evoked earlier~\cite{levy99}, the fully
lazy optimization is even optimal in the sense of
Lévy~\cite{levy80}. This means that the strategy reaches the weak
normal form in the same number of $\beta$-steps as the shortest
possible weak reduction sequence in the usual $\lambda$-calculus
without sharing.  Thus, fully lazy sharing turns out to be a
\emph{decidable} optimal strategy

\paragraph{Quantitative Types}
Intersection types were introduced as \emph{(denotational)}
    models capturing computational properties of functional
  programming in a broader
  sense~\cite{coppo78,CoppoDezaniVenneri81,BarendregtCoppoDezani83}.
  They extend simple types with a new
  constructor $\cap$, thus allowing to support
  polymorphism in a finitary way: a
  program $t$ is typable with $\sigma \cap \tau$ if $t$ is typable
  with both types $\sigma$ and $\tau$ independently. Originally,
  intersection enjoys associativity, commutativity, and in particular
  idempotency (\ie\ $\sigma \cap \sigma = \sigma$).  By changing to a
  \emph{non-idempotent} intersection constructor, one naturally comes
  to represent such a type by a
  multiset, also known as
  \emph{multi-type}. Idempotent as well as
    non-idempotent types allow for a characterization of several
  operational properties of programs, \eg\ termination
    of different evaluations strategies can be characterized by
    typability in some appropriate intersection type system. There
      is however a major difference between the two: while idempotent
      types provide \emph{qualitative} information, non-idempotent
      ones also provide \emph{quantitative} knowledge.
     More precisely, it is not only possible to prove
  that typability in one typing system characterizes
  termination of some particular evaluation strategy, \ie\ that a
  program is terminating if and only if it is typable in the corresponding system, but also an {\it
    upper bound} or {\it exact measure} for the time needed for its
  evaluation can be derived from its typing information
  (see~\cite{bucciarelli17} for a survey).

\paragraph{Contributions.}

The first contribution of this paper is a term
calculus implementing (full) node replication (\autoref{sec:calculus}) and internally
encoding skeleton extraction (\autoref{sec:ndlr}). We study some of
its main operational properties: termination of the
substitution calculus, confluence, and its relation with the $\lambda$-calculus.

Our second contribution is the use of the node replication paradigm to
give an alternative specification of two evaluation
strategies usually described by means of full or linear substitution:
call-by-name (\autoref{sec:whlr}) and weak fully lazy reduction
(\autoref{sec:ndlr}), based on the key notion of skeleton.
The former can be related to (weak) head
reduction, while the latter is a fully lazy version of (weak)
call-by-need. In contrast to other implementations of fully lazy
reduction relying on (external) meta-level definitions, our
implementation is based on formal operations internally defined over
the syntax of the calculus.

Our implementation of fully lazy reduction is based on
  a class of \emph{restricted} terms, called $\purelterms$,
 which simplifies the formal reasoning. This restriction is not ad-hoc in the sense it is stable through evaluations based on weak strategies, i.e. is an invariant property of evaluations starting from pure $\lambda$-terms.

Furthermore, while it is known that call-by-name and call-by-need
specified by means of full/linear substitution are observationally
equivalent~\cite{ariola97}, it was not clear at first whether the same
property would hold in our case. Our third contribution is a proof of this
result (\autoref{sec:observational}) using semantical tools coming
from proof theory --notably intersection types.  This proof
technique~\cite{kesner16} considerably simplifies other
approaches~\cite{ariola97,maraist98} based on syntactical
tools. Moreover, the use of intersection types has another important consequence: 
  standard call-by-name and call-by-need turn out to be
  observationally equivalent to call-by-name and call-by-need
  with node replication, as
  well as to the more semantical notion of neededness (see~\cite{KesnerRV18}).

Intersection types provide
quantitative information about fully lazy evaluation so that a fourth
contribution of this work is a measure based on type
derivations which turns out to be an upper bound to the
length of reduction sequences to normal forms in a fully lazy
implementation. 

More generally, our work bridges the gap between the Curry-Howard
theoretical understanding of node replication and concrete
implementations of fully lazy sharing.
Related works are presented in the concluding~\autoref{sec:conclusion-atomic}.

\section{A Calculus for Node Replication}%
\label{sec:calculus}

We now present the $\rcalc$-calculus (as in $\nrep$eplication).
From a syntactical point of view, we add two new constructors to the
$\lambda$-calculus: explicit substitutions and explicit distributors. From
an operational point of view, we provide a rewriting
system on $\rcalc$-terms together with a notion of \emph{levels}
which will play a key role in the next sections.

\subsection{Syntax}
\label{sec:syntax}

Given a countably infinite set of variables $x,y,z,...$, we consider the following grammars.
\[ \begin{array}{l@{\hspace{-.5mm}}rcll}
  \gramTitle{Terms} & t, u, r, s & \Coloneqq & x \mid \lx.t \mid t u \mid
  t \es x u \mid t \dist x {\ly.u}\\
  \gramTitle{Pure Terms}    &  p, q & \Coloneqq & x \mid \lx.p
  \mid p q\\
  \gramTitle{Term Contexts} & \fc & \Coloneqq &
  \ec \mid \lx.\fc \mid \fc t \mid t \fc \mid
  \fc \es x t \mid \fc \dist x {\ly.u} \mid
  t \es x \fc \mid t \dist x {\ly.\fc}\\
  \gramTitle{List Contexts} & \lc & \Coloneqq &
  \ec \mid \lc \es x u \mid \lc \dist x {\ly.u}
\end{array} \]
The set of terms is denoted by $\rterms$ and the subset of \deft{pure} terms is
denoted by $\pterms$. 
We write $\id$ for the identity function $\lx.x$, and $\tmsz{t}$ for the number of symbols of the term $t$.

The construction $\esub x u$ is an \deft{explicit substitution}, which
allows the \emph{sharing} of the subterm $u$.  A term $t \esub
x u$ can be simply seen as a more concise notation
for a \emph{let-binding} $\letb x u t$, where $u$ is shared across the
free occurences of the variable $x$ in $t$.  The second construction
  $\dist x {\ly.u}$ is an \deft{explicit distributor} (or
  simply \emph{distributor}), which is used specifically in the
  duplication of abstractions.

An \deft{explicit cut} is denoted by $\cut x u$, which is either $\esub x u$, or
$\dist x u$ when $u$ is $\ly.u'$, typically to factorize some definitions and
proofs where explicit substitutions and distributors behave similarly.
A characterization function $\isES \cdot$ on explicit cuts distinguishes these two cases: $\isES{\cut x u} = 1$ if $\cut x u=\esub x u$, and $0$ otherwise.
The application constructor associates to the left, \ie\ $t_1t_2  \ldots t_n$ means $((t_1t_2) \ldots t_n)$. Explicit cuts affects the rightmost constructor, \ie\ $tu\cut x s$ means $t (u\cut x s)$, otherwise we will  write parenthesis like $(t u)\cut x s$. 

Two notions of \deft{contexts} are used.  Term contexts $\fc$ extend
those of the $\lambda$-calculus to explicit cuts. A term context is said to be \deft{pure} if it does not contain any explicit cut.  List contexts $\lc$ denote an
arbitrary list of explicit cuts.  They will be used in particular to
implement reduction at a \textit{distance} (see \autoref{sec:operational-semantics}).

\deft{Free} and \deft{bound variables} of terms are defined as
expected, in particular $\fv {\lx.t} \eqdef \fv t \backslash \{x\}$
and $\fv {t \cut x u} \eqdef \fv t \backslash \{x\} \cup \fv u$.
We write $\nbocc x t$ to denote the number of free occurrences of
the variable $x$ in the term $t$.
These notions are extended
to contexts as expected. In particular $\fv\ec \eqdef \emptyset$
and $\fv{\lc \cut x u} \eqdef \fv\lc \backslash \{x\} \cup \fv{u}$.
The \deft{domain} of a list context is defined as
$\domlc\ec \eqdef \emptyset$
and $\domlc{\lc \cut x u} \eqdef \domlc{\lc} \cup \{ x \} $.

The standard notion of \emph{$\alpha$-conversion}~\cite{barendregt85} on
$\lambda$-terms is extended to
the full set of $\rcalc$-terms as expected,
and we systematically assume $\alpha$-conversion
whenever necessary to avoid capture of free
variables. We write $t \msub x u$ for the \deft{meta-level
substitution} operation simultaneously replacing all
the free occurrences of the variable $x$ in $t$ by the term $u$.

The \deft{application of a context
$\fc$ to a term $t$}, written $\ctx\fc t$, replaces the hole $\ec$ of $\fc$ by $t$.  Thus, for
instance, $\ctx\ec t = t$ and $\ctx{(\lx.\ec)} t = \lx.t$.
This operation is not defined modulo
$\alpha$-conversion, so that capture of variables eventually
happens, such as in the second example if $x \in \fv{t}$.  Thus,
another kind of application of contexts to
terms is also considered, identified by double brackets, and is only defined if
there is no capture of variables. For
instance,
$\ctx{(\ly.\ec)} x = \ly.x$ while $\ctxnc{(\lx.\ec)} x$ is undefined.

\subsection{Operational semantics}
\label{sec:operational-semantics}

Before presenting the dynamics of the calculus, let us introduce some notations
and definitions concerning reduction.
Let $\rew\rewrel$ be any reduction
relation. We write $t \rew\rewrel
  t'$ when  $(t,t') \in \;\rew\rewrel$, and $\rewn \rewrel$
(resp. $\rewp{\rewrel}$) for the reflexive-transitive (resp. transitive) closure
of $\rew\rewrel$.  A term $t$ is said to be \deft{$\rewrel$-confluent}
iff $t \rewn\rewrel u$ and $t \rewn\rewrel s$ implies there is $t'$
such that $u \rewn\rewrel t'$ and $s \rewn\rewrel t'$. The relation
$\rewrel$ is \deft{confluent} iff every term is $\rewrel$-confluent.
A term $t$ is said to be in \deft{$\rewrel$-normal form} (written also
$\rewrel$-nf) iff there is no $t'$ such that $t \rew\rewrel t'$. A
term $t$ is said to be \deft{$\rewrel$-terminating} or
\deft{$\rewrel$-normalizing} iff there is no infinite
$\rewrel$-sequence starting at $t$.  The reduction $\rewrel$ is said
to be \deft{terminating} iff every term is $\rewrel$-terminating.

Explicit substitutions may block some expected \emph{meaningful}
(\ie\ non-structural) reductions.  For instance, $\beta$-reduction is
blocked in $(\lx.t) \es y s u$ because an explicit substitution
lies between the function and its argument.  This situation does not happen in graphical representations
(\eg~\cite{girard96}), but it is typical in the sequential structure
of \emph{term} syntaxes.

There are at least two ways to handle this issue. The first one
is based on  \emph{structural/permutation} rules, as
in~\cite{gundersen13a}.  Therefore, in the previous example, the
substitution is first pushed outside the application node by means of
a permutation rule, as $(\lx.t) \es y s u \rew{} ((\lx.t) u) \es y s$,
so that $\beta$-reduction is finally unblocked.  The second,
less elementary, possibility is given by an operational semantics
\emph{at a distance}~\cite{accattoli10,accattoli14a}, where the
$\beta$-redex can be fired by a rule like $\ctx\lc{\lx.t} u
\rew{} \ctx\lc{t \es x u}$, where $\lc$ is an arbitrary list
context. The distant paradigm is therefore used to gather meaningful
and permutation rules in only one reduction step.  In $\rcalc$, we
combine these two complementary technical tools.  First, we consider
the following permutation rules:
\[ \begin{array}{llll}
  \ly.t \cut x u & \rrule{\permlr} & (\ly.t) \cut x u & \text{ if } y \notin \fv u\\
  t \cut x u s & \rrule{\permlr}  & (ts) \cut x u & \text{ if } x \notin \fv s \\
  t s \cut x u & \rrule{\permlr}  & (ts) \cut x u & \text{ if } x \notin \fv t\\
  t \cut x {u \cut y s} & \rrule\permlr & t \cut x u \cut y s & \text{ if } y \notin \fv t
\end{array} \]
The reduction relation $\rew\permlr$ is defined as the closure of the four
rules $\rrule\permlr$ under \emph{all} contexts.

\begin{exa}
  Let $t = x \esub x {w \dist {z_1}{\id} (\ly.y \esub {z_2}{z_3})}$.
  Both inner explict cuts $\dist {z_1}{\id}$ and $\esub {z_2}{z_3}$ are pushed
  outside the main explicit substitution, which results in a pure term followed
  by a list of explicit cuts.
  \[ \begin{array}{llllll}
    t
    &\rew\permlr & x \esub x {w \dist {z_1}{\id} (\ly.y) \esub {z_2}{z_3}} &
    &\rew\permlr & x \esub x {(w \dist{z_1}{\id} (\ly.y)) \esub{z_2}{z_3}}\\
    &\rew\permlr & x \esub x {w \dist{z_1}{\id}(\ly.y)} \esub{z_2}{z_3} &
    &\rew\permlr & x \esub x {(w(\ly.y)) \dist{z_1}{\id}} \esub{z_2}{z_3}\\
    &\rew\permlr & x \esub x {w(\ly.y)} \dist{z_1}{\id} \esub{z_2}{z_3}
  \end{array} \]
\end{exa}

Permutations do not hold any computational content, only a structural one.  Indeed, all terms in the
reduction sequence above could be naturally translated
to the same graphical notation. In order to highlight the
computational content of node replication, we combine distance and
permutations within a single rewriting semantics.
In addition, because operational semantics at a distance is
directly inspired from graph formalisms, there is a better  correspondence between the syntactic representation of terms and the graphical
representations of their associated reduction
notion~\cite{kesner07,accattoli18}.
The resulting  reduction relation
$\rcalc$ is given by
the closure under all  contexts of the following rules:
\label{d:nrep}%
\[ \begin{array}{lllll}
  \ctx\hlc{\lx.t} u & \rrule{\db}&   \ctx\hlc{t \es x u} \\
  t \es x {\ctx\hlc{us}} & \rrule{\apprr} & \ctx\hlc{t \msub x {yz} \es y u
    \es z s} &  \mbox{where $y$ and $z$ are fresh} \\
  t \es x {\ctx\hlc{\ly.u}} & \rrule{\distrr} & \ctx\hlc{t \dist x {\ly.z \es z u}} & \mbox{where $z$ is fresh}  \\
  t \dist x {\ly.u} & \rrule{\absrr}& \ctx\hlc{t \msub x {\ly.p}} & \mbox{where }
  u \rewn\permlr \ctx\hlc p \mbox{ and } y \notin \fv\hlc \\
  t \es x {\ctx\hlc y} & \rrule{\varrr} & \ctx\hlc{t \msub x y} \\
\end{array} \]
where the distant contexts are highlighted in \verdear{green} to make it easier to read.
The \deft{$\rcalc$-calculus} is defined by the set of terms  $\rterms$ equipped
with this reduction relation. 
In the five rules just above, a list context $\lc$ is pushed outside
the term. We assume in all these cases that there is no capture of
variables caused by this transformation, \eg\ in rule $\db$ this means
that $\domlc\lc \cap \fv{u} = \emptyset$.  Apart from the \emph{distant
${\redName B}$eta} rule $\db$ used to fire $\beta$-reduction, there are
four substitution rules used to copy nodes of \emph{pure} terms
pushing outside all the cuts surrounding the node to be copied.
Rule~$\apprr$ copies one application node, while
rule~$\varrr$ copies one variable node.
To copy abstractions, both rules $\distrr$ and $\absrr$ are needed. 
Notice that the (meta-level and capture-free) substitution is \emph{full}, in the sense that
it is performed simultaneously on all ocurrences of the free variable $x$ at the same time.

The reduction relation $\rrule{\nrsub}$ is defined as
  $\rrule\apprr \cup \rrule\distrr \cup \rrule\absrr \cup
  \rrule\varrr$, while the $\redName s$ubstitution relation $\rew
\nrsub$ (resp.  distant Beta relation $\rew \db$) is defined as the
closure of $\rrule{\nrsub}$ (resp. $\rrule \db$) under \emph{all}
contexts, and the reduction relation $\rew\nrep$ is the union of $\rew
\nrsub$ and $\rew \db$.

\begin{exa}
  \label{ex:lrapp}%
  This example illustrates the use of rules $\apprr$ and $\varrr$ to
  replicate application and variables nodes,
  as well as rule $\db$ to fire reduction.
  No distance is involved in this example.
  \[\begin{array}{lll}
    (\lx.xx)(yz)
    &\rew\db (xx) \esub x {yz}
    &\rew\apprr ((x_1x_2)(x_1x_2)) \esub {x_1} y \esub {x_2} z\\
    &\rew\varrr (yx_2)(yx_2) \esub {x_2} z
    &\rew\varrr (yz)(yz)
  \end{array} \]
\end{exa}

\begin{exa}
  \label{ex:lrabs}%
  Replication of abstractions is more involved, as illustrated below.
  \begin{align}
    (\lx.xx)(\ly.(ww)y)
    &\rew\db (xx) \esub x {\ly.(ww)y}\\
    &\rew\distrr (xx) \dist x {\ly.z \esub z {(ww)y}} \label{ex:lrabs-dist}\\
    &\rew\apprr (xx) \dist x {\ly.(z_1z_2) \esub {z_1}{ww} \esub {z_2} y}
    \label{ex:lrabs-app}\\
    &\rew\varrr (xx) \dist x {\ly.(z_1y) \esub {z_1}{ww}} \label{ex:lrabs-var}\\
    &\rew\apprr (xx) \dist x {\ly.((z_3 z_2)y) \verdear{\esub {z_3}{w} \esub {z_2}{w}}}
    \label{ex:lrabs-appw}\\
    &\rew\absrr ((\ly.(z_3z_2)y)(\ly.(z_3z_2)y)) \esub {z_3}{w} \esub{z_2} w
      \label{ex:lrabs-abs}\\
    &\rew\varrr ((\ly.(wz_2)y)(\ly.(wz_2)y)) \esub{z_2} w
      \label{ex:lrabs-varwone}\\
    &\rew\varrr (\ly.(ww)y)(\ly.(ww)y) \label{ex:lrabs-varwtwo}
  \end{align}
\end{exa}

The specificity in copying an abstraction $\ly.u$ is due to the
(binding) relation between the binder $\ly$ and all the free
occurrences of $y$ in its body $u$.  Abstractions are thus copied in
two stages.  The first one is implemented by the rule~$\distrr$, which
creates a distributor in which a potentially replicable abstraction is
placed, while moving its body inside a new explicit substitution.
Thus, in line~\eqref{ex:lrabs-dist}, we create a distributor
over the abstraction $\ly$, while $(ww)y$ is placed inside an explicit
substitution $\esub z {(ww)y}$.  Notice that this substitution is in
the scope of abstraction $\ly$.  The distributor is marking the
fact that the abstraction needs to be further duplicated.
There are then two kinds of potentially replicable nodes shared in the body of
the corresponding abstraction.

\begin{enumerate}
    \item All free occurences of the
      variable bound by the main abstraction (here $\ly$)
      must be replicated by means of the rule~$\varrr$ \eqref{ex:lrabs-var},
      so as to keep the correct binding structure.
      This means that all the nodes leading to these
      occurences must also be duplicated:
      this is why rule~$\apprr$ is first used in~\eqref{ex:lrabs-app}.
    \item All nodes which are neither a free
      occurence of the bound variable nor in the path to
      such a node can be arbitrarily copied inside the distributor
     (\eg\ the internal application node in line~\eqref{ex:lrabs-appw}),
      or replicated later (\eg\ the two variable nodes $w$
      in \eqref{ex:lrabs-varwone} and \eqref{ex:lrabs-varwtwo}).
\end{enumerate}
Components which are not replicated inside the
distributor form a list of explicit cuts, which can occur at
  different depths inside this distributor. Indeed, in~\eqref{ex:lrabs-appw},
  there are two explicit substitutions $\esub
  {z_3}{w}$ and $\esub {z_2} w$.
The cuts can be
gathered together into a list context, called $\lc$ in the definition
of rule~$\absrr$, which is pushed outside by using
permutation rules, before performing the substitution of the pure
body containing all the bound occurrences of $y$ (here $\ly.(z_1z_2)y$).
This operation is in general hard to specify using only distance since the cuts can
appear at arbitrary depth in the distributor, and this
is one of the reasons to introduce the use of permutation rules in
rule~$\absrr$.

Other choices are possible, such as replicating all the nodes, or only the
uppermost application and the node $y$ (corresponding to fully lazy
duplication), as long as at least all free occurences of $y$ are duplicated.

\label{page:example-reduction}

\begin{exa}
    This last example showcases different reduction steps with distance,
    highlighted in green.
    \begin{align*}
      (\lx.x) \verdear{\esub{z_4}{z_5}} (w \dist{z_1}{\id} (\ly.y\esub{z_2}{z_3}))
    &\rew\db x \esub x {w \dist{z_1}{\id} (\ly.y\esub{z_2}{z_3})}
    \esub{z_4}{z_5}\\
    &\rew\apprr (x_1x_2) \esub {x_1}{w \verdear{\dist{z_1}{\id}}}
    \esub{x_2}{\ly.y \esub{z_2}{z_3}} \esub{z_4}{z_5}\\
    &\rew\varrr (wx_2) \dist{z_1}\id \esub{x_2}{\ly.y \esub{z_2}{z_3}}
    \esub{z_4}{z_5}\\
    &\rew\distrr (wx_2) \dist{z_1}\id \dist{x_2}{\ly.x \esub x {y
    \verdear{\esub{z_2}{z_3}}}} \esub{z_4}{z_5}\\
    &\rew\varrr (wx_2) \dist{z_1}\id \dist{x_2}{\ly.y \verdear{\esub{z_2}{z_3}}} \esub{z_4}{z_5}\\
    &\rew\absrr (w(\ly.y)) \dist{z_1}\id \esub{z_2}{z_3} \esub{z_4}{z_5}
    \end{align*}
\end{exa}

Notice that an $\nrep$-step can be decomposed into some
$\permlr$-steps followed by a simpler step not involving any
list context.  Indeed, $t \rew\nrep u$ could be simulated by $t \rewn\permlr t'
\rew{\purelr\nrep} u$, where $\rew{\purelr\nrep}$ is the closure under
all contexts of the following set of rewriting rules:
\[ \begin{array}{llll}
  (\lx.t) u & \rrule{\purelr\db} & t \es x u &  \\
  t \es x {us} & \rrule{\purelr\apprr} & t \msub x {yz} \es y u \es z s & \\
  t \es x {\ly.u} & \rrule{\purelr\distrr} & t \dist x {\ly.z \es z u} &  \\
  t \dist x {\ly.p} &  \rrule{\purelr\absrr} & t \msub x {\ly.p} \\
  t \es x y & \rrule{\purelr\varrr} & t \msub x y & \\
\end{array} \]
For instance, step (\ref{ex:lrabs-abs}) in \autoref{ex:lrabs} can be decomposed
as follows, where $r = \ly.(z_3z_2)y$:
\[  (xx)
  \dist x {r \esub {z_3}{w} \esub {z_2}{w}}
  \rewn{\pi} (xx) \dist x {r} \esub {z_3}{w} \esub {z_2}{w}
\rew{\purelr\absrr}  (r r) \esub {z_3}{w} \esub {z_2}{w} \]
This decomposition will be useful in some of our proofs, but we prefer to integrate
distance inside the rules, as initially defined on
page~\pageref{d:nrep},  to highlight the computational behavior and
execute permutations only when strictly necessary.

\subsection{Levels}
\label{sec:levels}

This subsection introduces the syntactical notion of level and its
associated properties.  Intuitively, the level of a variable in a term
indicates the maximal depth (\emph{only} w.r.t. explicit substitutions
and not w.r.t. explicit distributors) of its free
occurrences. However, in order to be sound with respect to
the permutation rules, levels do not consider depth in the
usual sense only, but also across linked chains of explicit
substitutions. For instance, the level of $z$ in both
$(xx) \esub x {y \esub y z}$
and $(xx) \esub x y \esub y z$
is  the same.
Levels will play a key role in the next sections:
they will be the combinatorial witnesses of
the progress of $\nrsub$-substitution steps, necessary to prove
termination of the $\nrsub$-relation.
They will also be helpful to define a decreasing measure on typing derivations
in \autoref{sec:types}.
The \deft{level} of a variable $z$ in a term $t$ is defined  by induction:
\begin{align*}
  \lv z x &\eqdef  0  \\
  \lv z {t_1 t_2} &\eqdef \max (\lv z {t_1}, \lv z {t_2}) \\
  \lv z {\lx.t} &\eqdef \lv z t\\
  \lv z {t \cut x u} &\eqdef
  \begin{cases}
    \lv z t & \text{if } z \notin \fv u \\
    \max (\lv z t, \lv x t + \lv z u + \isES{\cut x u})
            & \text{otherwise}
          \end{cases} \\
\end{align*}
In the two last cases, we can always suppose $z \neq x$, because we work modulo
$\alpha$-conversion.
Notice that $\lv {z} t = 0$ whenever $z \notin \fv{t}$ or $t$ is pure.
We illustrate the concept of level by an example.
Consider $t = x \es x {z \es y w} \es w {w'}$,
then $\lv{z}{t} = 1$, $\lv{w'}{t} = 3$ and $\lv y t = 0$ because $y \notin
\fv t$.
This notion is also extended to contexts as expected, \ie\ $\lv \ec \fc =
\lv z {\ctx\fc z}$, where $z$ is a fresh variable.
Remark that for any variable $x$, $\lv \ec \fc \leq \lv x {\ctxnc\fc x}$
and $\lv x {\ctxnc\fc p} \leq \lv x {\ctxnc\fc x}$ for any $p \in \pterms$.

\begin{restatable}{lem}{levelsubstitution}
  \label{l:level-substitution}%
  Let $x \neq z$, $t \in \rterms$ and $p \in \pterms$:
  \begin{enumerate}
    \item \label{l:level-substitution-void}%
      If $ z \notin \fv p$, then $\lv
      z {t \msub x p}= \lv z t$.
    \item \label{l:level-substitution-non-void}%
      If $z \in \fv p$,
      then $\lv {z} {t \msub x p} = \max (\lv {z} t, \lv
      x t)$.
  \end{enumerate}
\end{restatable}
\begin{proof}
  By induction on $t$ \seeappendix{p:level-substitution}.
\end{proof}

\begin{restatable}{lem}{stabilitylevels}
  \label{l:stability-levels}%
  Let $t \in \rterms$ and $w$ be any variable.
  \begin{enumerate}
    \item \label{l:stability-levels-permutation}%
      If $t_0 \rew\permlr t_1$,
      then $\lv w {t_0} \geq \lv w {t_1}$.
    \item \label{l:stability-levels-substitution}%
      If $t_0 \rew\nrsub t_1$, then $\lv w {t_0} \geq \lv w {t_1}$.
  \end{enumerate}
\end{restatable}
\begin{proof}
  By induction on the reduction relation
  \seeappendix{p:stability-levels}.
\end{proof}

Notice that there are two cases when the level of a variable in a term may decrease:
\begin{itemize}
  \item Moving an explicit cut out of another one
    with a permutation rule when the
    first cut is a void cut, \ie\ its domain does not bind any other variable.
    Thus \eg\  if $t = x \es x {z \es y w} \es w {w'}
    \rew\permlr x \es x z \es y w \es w {w'} = u$,
    then $\lv{w'} t = 3 > 2=\lv{w'} u$.
  \item Using rule $\rrule\varrr$.
    Thus \eg\  if $t = (xx) \es {x} {y} \es
    {y} z \rew\varrr (yy) \es {y} z = u$, then
    $\lv z t = 2 > 1 = \lv z u$.
\end{itemize}

Hence, levels alone are not enough
to prove termination of $\rew\nrsub$. We thus define a decreasing measure for
$\rew\nrsub$ in which not only variables are indexed by a
level, but also constructors.  For instance, in the term $t \es x {\ly.yz}$,
we can consider that the level of
\emph{all} the constructors of $\ly.yz$,
including the abstraction and the application, have level $\lv x t$.
This will ensure that the level of an abstraction will
decrease when applying rule $\distrr$, as well as the level of an
application when applying rule $\apprr$.

\section{Operational Properties}%
\label{sec:properties}

We now prove  three key properties of the
$\rcalc$-calculus: termination of the
reduction system $\rew\nrsub$,
relation between $\rcalc$ and the $\lambda$-calculus, and confluence of the
reduction system $\rew\nrep$.

\paragraph{Termination of $\rew\nrsub$}%

Some (rather informal) arguments are provided in~\cite{gundersen13a} to justify
termination of the substitution subrelation of their  calculus.
We expand these ideas into an alternative full formal proof adapted to our case,
which is based on a measure being strictly decreasing w.r.t. $\rew\nrsub$.

We consider a set $\subobj$ of objects of the form $\subes  k n$ or
$\subdist  k\ (k,n \in \mathbb N)$,
which is equipped with the following ordering $\subge$
($\subgeq$ denotes its reflexive closure):
\[ \begin{array}{rll@{\hspace{.25cm}}rll}
  \subes{k}{n}
  & \subge \subes{k'}{n'}  &\mbox{if } k > k',
  \mbox{ or } (k = k' \mbox{ and } n > n')
  & \subdist k &\subge  \subes{k'}{n} &\mbox{if } k \geq k'\\
  \subes  k n &\subge  \subdist{k'} &\mbox{if } k >  k' &
  \subdist  k &\subge  \subdist{k'} &\mbox{if } k > k'
\end{array} \]

Notice that the symbols $\subesname$ and $\subdistname$ are
  just formal expressions, \ie\ $\subobj$ could be alternatively (but
  less clearly) defined as
  $(\mathbb{N} \times \mathbb{N}) \uplus \mathbb{N}$.
\begin{lem}
  \label{t:plus-S-wf}%
  The order $\subge$ on the set $\subobj$ is well-founded.
\end{lem}
\begin{proof}
  Let us consider the set $\mathbb N$ equipped with the
  standard order $>_{\mathbb N}$ on natural numbers.  Let us also consider
  the set $\mathbb N_\infty := \mathbb N \uplus \{\infty\}$ equipped with the
  order $>_{\infty} := >_{\mathbb N} \cup \{ \pair{\infty}{n} \mid n \in
  \mathbb N\}$.  Since $>_{\mathbb N}$ and $>_{\infty}$ are both WF, then the
  lexicographic order induced by $\pair{>_{\mathbb N}}{>_{\infty}}$ on
  $\mathbb N \times \mathbb N_\infty$, written
  $\subgelex$, is also WF. We show that $\subge$ is WF
  by projecting it into the WF order $\subgelex$,
  \ie\ we define a projection function $\subproj$
  such that $s \subge s'$ implies $\subproj(s) \subgelex \subproj(s')$, for
  any $s,s' \in \subobj$.
  Let us define $\subproj(s) = \pair{\subka(s)}{\suben(s)}$, where
  $\subka(\subes k n) \eqdef k$ and $\subka(\subdist k) \eqdef k$ while
  $\suben(\subes k n) \eqdef n$ and  $\suben(\subdist k) \eqdef \infty$.
  We reason by cases.
  \begin{itemize}
    \item $s_0=\subes k n \subge
      \subes {k'} {n'} = s_1$. Then $\pair{\subka(s_0)}{\suben(s_0)} = \pair{k}{n} \subgelex \pair{k'}{n'} =
      \pair{\subka(s_1)}{\suben(s_1)}$ holds by definition since
      either $k > k'$ or $k = k'$ and $n > n'$.
    \item $s_0=\subdist k \subge
      \subdist {k'}  = s_1$. Then $\pair{\subka(s_0)}{\suben(s_0)} = \pair{k}{\infty} \subgelex \pair{k'}{\infty} =
      \pair{\subka(s_1)}{\suben(s_1)}$ holds by definition since $k >
      k'$. 
    \item $s_0=\subes  k n \subge
      \subdist {k'}  = s_1$. Then $\pair{\subka(s_0)}{\suben(s_0)} = \pair{k}{n}
      \subgelex \pair{k'}{\infty} =
      \pair{\subka(s_1)}{\suben(s_1)}$ holds by definition since $k >
      k'$. 
    \item $s_0=\subdist k  \subge  \subes {k'} {n'} = s_1$. Then $\pair{\subka(s_0)}{\suben(s_0)} = \pair{k}{\infty} \subgelex \pair{k'}{n'} =
      \pair{\subka(s_1)}{\suben(s_1)}$ holds by definition
      since either $k > k'$ or $k = k'$ and $\infty > n$. 
      \qedhere
  \end{itemize}
\end{proof}

We write $\subgemult$ for the multiset extension of the
order $\subge$ on $\subobj$, which turns out to be
well-founded~\cite{Nipkow-Baader} by \autoref{t:plus-S-wf}.
Some operations on multisets are needed to build the measure
$\submeas{\_}$ on terms. Indeed,
let $M$ be a multiset of objects in $\subobj$.
Multiset sum is denoted $\sqcup$. Furthermore:
\begin{enumerate}
  \item The \deft{$\subesname$-elements} (resp. \deft{$\subdistname$-elements}) of
    the multiset $M$ are all the objects of the form
    $\subes{k}{n}$ (resp. $\subdist{k}$) in $M$. We then may write
    $M$ as $M_{\subesname} \sqcup M_{\subdistname}$, where $M_{\subesname}$ (resp.
    $M_{\subdistname}$)
    contains all the $\subesname$-elements (resp $\subdistname$-elements) of $M$.
  \item Given $K \in \mathbb N$, we write $M^{\leq K}$ (resp. $M^{> K}$)
    for the multiset containing all $o \in M$ such that the first element of
    $o$ is less than $K$ (resp. strictly greater than $K$).
    We write $M^{> K}_{\subesname}$ for $M^{> K} \sqcap M_{\subesname}$,
    where $\sqcap$ denotes multiset intersection.
  \item $M$  can thus be decomposed in three disjoint multisets
    $M_{\subdistname}, M_{\subesname}^{\leq K}$ and $M_{\subesname}^{> K}$, for every $K \in \mathbb N$.
  \item  We also define the following operation on $M$:
    \[ p \cdot M \eqdef  \mult{\subes{p+k}{n} \mid
      \subes{k}{n} \in M }
    \sqcup \mult{\subdist{p+k} \mid \subdist{k} \in M} \]
\end{enumerate}
We are now ready to (inductively) define our \deft{cuts level} measure
$\submeas\cdot$ on terms.
\[
  \renewcommand{\arraystretch}{1.15}
  \begin{array}{l@{\hspace{1cm}}l@{\hspace{1cm}}l}
    \submeas x \eqdef \emult &
    \submeas {\lx.t} \eqdef \submeas t &
    \submeas {t u} \eqdef \submeas t \sqcup \submeas u\\
    \multicolumn{3}{l}{
      \submeas {t \es x u} \eqdef
      \submeas t \sqcup ( (\lv x t + 1) \cdot \submeas u ) \sqcup [\subes{\lv x     t + 1}{\tmsz u}]}\\
    \multicolumn{3}{l}{
      \submeas {t \dist x u} \eqdef
    \submeas t \sqcup ( \lv x t \cdot \submeas u ) \sqcup [\subdist{\lv x t}]}
  \end{array}
\]
Intuitively, the integer $k$ in $\subes k n$
and $\subdist k$ counts the level of variables bound by explicit substitutions,
while $n$ counts the size of terms to be substituted by an ES.
Remark that for every pure term $p$ we have $\submeas p = \emult$.

\begin{exa}
    \label{ex:ms}%
    Consider the following reduction sequence:
    \begin{align*}
      t_0 = (yy) \es y {(\lz.x) w}
      &\rew\apprr (y_1 y_2)(y_1 y_2) \es {y_1}{\lz.x} \es {y_2} w &= t_1\\
      &\rew\varrr (y_1 w)(y_1 w) \es {y_1}{\lz.x} &= t_2\\
      &\rew\distrr (y_1 w)(y_1 w) \dist {y_1} {\lz.x' \es {x'} x} &= t_3\\
      &\rew\absrr ((\lz.x') w)((\lz.x') w) \es{x'}{x} &= t_4\\
      &\rew\varrr ((\lz.x) w)((\lz.x) w) &= t_5
    \end{align*}
    We have $\submeas {t_0} = \mult{\subes 1 4}$,
    $\submeas {t_1} = \mult{\subes 1 1, \subes 1 2}$,
    $\submeas {t_2} = \mult{\subes 1 2}$,
    $\submeas {t_3} = \mult{\subes 1 1, \subdist 0}$,
    $\submeas {t_4} = \mult{\subes 1 1}$ and $\submeas {t_5} = \emult$.
\end{exa}

\begin{fact}
  Some properties on multisets are straightforward:
  \begin{itemize}
    \item If $M_1 \subgemult{} M_2$, then $M_1 \sqcup M \subgemult{} M_2 \sqcup M$.
    \item If $M_1\subgemult{} M_2$, then $k \cdot M_1 \subgemult{} k \cdot M_2$
      for any $k \in \mathbb N$.
    \item $k_1 \cdot k_2 \cdot M = (k_1 + k_2) \cdot M$.
  \end{itemize}
\end{fact}

\begin{lem}
  \label{l:ms-lc}%
  If $\submeas t \subgemult{} \submeas u$ and $\lv x t \geq \lv x u$ holds for every
  $x \in \domlc\lc$, then $\submeas {\ctx \lc t} \subgemult{} \submeas {\ctx \lc u}$.
\end{lem}
\begin{proof}
  By induction on $\lc$. The property is straightforward.
\end{proof}

\begin{restatable}{lem}{mssubstitution}
  \label{l:ms-substitution}%
  Let $t$ be a term, $x$ a variable and $p$ a pure term.
  Let $K = \lv x t$.
  Then there is $N \in\mathbb N$ such that $\submeas {t \msub x p} \sqsubseteq
  \submeas {t}_{\subdistname} \sqcup \submeas {t}^{>K}_{\subesname} \sqcup
  \mult{\subes k n \mid k \leq K \mbox{ and } n \leq N}$.
\end{restatable}
\begin{proof}
  By induction on $t$
  \seeappendix{p:ms-substitution}.
\end{proof}

\begin{restatable}{lem}{mspi}
  \label{l:ms-stable-by-pi}
  Let $t \in \rterms$.
  Then $t \rew\permlr t'$ implies $\submeas t \subgeqmult \submeas{t'}$.
\end{restatable}
\begin{proof}
  By induction on the reduction $t \rew\permlr t'$.
  We only detail here the case where
  $t = s \esub y {r \esub x u} \rrule\permlr s \esub y r \esub x u = t'$ at root,
  where $x \notin \fv t$
  \seeappendix{p:ms-stable-by-pi}.
  \begin{align*}
    \submeas{t}
        &= \submeas s \sqcup (\lv y s + 1) \cdot \submeas{r \es x u}
        \sqcup \mult{\subes{\lv y s + 1}{\tmsz {r \es x u}}} \\
        &= \submeas s \sqcup (\lv y s + 1) \cdot
        \left(\submeas r \sqcup (\lv x r + 1) \cdot \submeas u
        \sqcup \mult{\subes{\lv x r + 1}{\tmsz u}} \right) \\
        &\quad \sqcup \mult{\subes{\lv y s + 1}{\tmsz {r \es x u}}} \\
        &= \submeas s \sqcup (\lv y s + 1) \cdot \submeas r 
        \sqcup (\lv y s + \lv x r + 2) \cdot \submeas u \\
        &\quad \sqcup \mult{
          \subes{\lv y s + \lv x r + 2}{\tmsz u},
        \subes{\lv y s + 1}{\tmsz{r \es x u}}} \\
        &= \left( \submeas s \sqcup (\lv y s + 1) \cdot \submeas r
        \sqcup \mult{\subes{\lv y s + 1}{\tmsz{r \es x u}}} \right) \\
        &\quad \sqcup (\lv y s + \lv x r + 2) \cdot \submeas u
        \sqcup \mult{\subes{\lv y s + \lv x r + 2}{\tmsz u}} \\
        & \subgemult{} \left( \submeas s \sqcup (\lv y s + 1) \cdot \submeas r
        \sqcup \mult{\subes{\lv y s + 1}{\tmsz r}} \right) \\
        &\quad \sqcup (\lv x {s \es y r} + 1) \cdot \submeas u
        \sqcup \mult{\subes{\lv x {s \es y r} + 1}{\tmsz u}} \\
        &= \submeas{t'}
        \end{align*}
      The $\subgemult{}$ inequation is justified by the following facts:
      \begin{itemize}
        \item $\tmsz {r\es x u} > \tmsz r$.
        \item $\lv y s + \lv x r +2  =  \max(0,\lv y s + \lv x r +1)+1 = \lv x {s \es y r} +1$.
          \qedhere
      \end{itemize}
\end{proof}

\begin{lem}
  \label{l:ms-stable-by-s}
  Let $t \in \rterms$.
  Then $t \rew\nrsub t'$ implies $\submeas t \subgemult \submeas{t'}$.
\end{lem}
\begin{proof}
  Let $t = \ctx \fc {t_0} \rew\nrsub \ctx \fc {t_1} = t'$,
  where $t_0 \rew\nrsub t_1$ is a reduction step at the root
  position. We proceed by induction on $\fc$.
  We detail the base case which is $\fc=\ec$.  In all such
  cases we use \autoref{l:ms-stable-by-pi} to push $\lc$ outside, \ie\ we
  can write $t_0 \rew{\nrsub} t_1$ as $t_0 \rewn\permlr
  \ctx{\lc}{t'_0} \rew{\nrsub'} \ctx{\lc}{t'_1} = t_1$,
  where $t'_0 \rew{\nrsub'} t'_1$ is a root step which does not
  push any list context outside. We then show the property for root
  steps $t'_0 \rew{\nrsub'} t'_1$, and we conclude
  with \autoref{l:ms-stable-by-pi} then \autoref{l:ms-lc}
  by $\submeas {t_0} \subgeqmult \submeas
  {\ctx{\lc}{t'_0}} \subgemult \submeas { \ctx{\lc}{t'_1}} = \submeas {t_1}$ since $\lv x {t'_0} \geq \lv x {t'_1}$ holds for every
  $x \in \domlc{\lc}$ by \autoref{l:stability-levels}.
  Let us analyse all  the cases $t_0' \rew{\nrsub'} t_1'$.
  \begin{enumerate}
    \item If $t'_0 = t \es x {us} \rew\apprr t \msub x {yz} \es y u \es z s = t'_1$,
      where $y$ and $z$ are fresh variables,
      then \[ \submeas{t'_0}
        = \submeas t \sqcup (\lv x t + 1) \cdot \submeas{us}
      \sqcup \mult{\subes{\lv x t + 1}{\tmsz{us}}} \text{ and} \]
      \begin{align*}
        \submeas{t'_1}
        &= \submeas{t \msub x {yz} \es y u} \sqcup (\lv z {t \msub x {yz} \es y
        u} + 1) \cdot \submeas s\\
        &\quad \sqcup \mult{\subes{\lv z {t \msub x {yz} \es y u} + 1}{\tmsz s}} \\
        & = (\submeas{t \msub x {yz}} \sqcup (\lv y {t \msub x {yz}} + 1) \cdot \submeas u
        \sqcup \mult{\subes{\lv y {t \msub x {yz}} + 1}{\tmsz u}}) \\
        &\quad \sqcup (\lv z {t \msub x {yz} \es y u} + 1) \cdot \submeas s
        \sqcup \mult{\subes{\lv z {t \msub x {yz} \es y u} + 1}{\tmsz s}} \\
        &= (\submeas{t \msub x {yz}} \sqcup (\lv x t + 1) \cdot \submeas u
        \sqcup \mult{\subes{\lv x t + 1}{\tmsz u}}) \\
        &\quad \sqcup (\lv x t + 1) \cdot \submeas s
        \sqcup \mult{\subes{\lv x t + 1}{\tmsz s}} \\
        &= \submeas{t \msub x {yz}} 
        \sqcup (\lv x t + 1) \cdot \submeas {us}
        \sqcup \mult{\subes{\lv x t + 1}{\tmsz u},
        \subes{\lv x t + 1}{\tmsz s}}
      \end{align*}
      By \autoref{l:ms-substitution},
      $\submeas{t \msub x {yz}} \sqsubseteq
      \submeas t_{\subdistname} \sqcup \submeas t_{\subesname}^{>\lv x
        t} \sqcup \mult{\subes k n \mid k \leq \lv x t, n \leq
          N}$ for some $N \in \mathbb{N}$.
      We also have that $\mult{\subes {\lv x t + 1}{\tmsz{us}}}
      \subgemult \mult{\subes{\lv x t + 1}{\tmsz u}, \subes{\lv x t + 1}{\tmsz s}}
      \subgemult {\mult{\subes k n \mid k \leq \lv x t, n \leq N}}$.
      Moreover, $\submeas t \sqsupseteq \submeas t_{\subdistname} \sqcup \submeas
      t_{\subesname}^{>\lv x t} $
      and $\mult{\subes{\lv x t + 1}{\tmsz{us}}}
      \subgemult \mult{\subes k n \mid k \leq \lv x t, n \leq N}$.
      Thus we conclude  $\submeas{t_0'} \subgemult \submeas{t_1'}$.

    \item If $t'_0 = t \es x {\lam y u} \rew\distrr t \dist x {\lam y {z \es z u}} = t'_1$,
      then \[ \submeas{t'_0}
        = \submeas t \sqcup (\lv x t + 1) \cdot \submeas u 
      \sqcup \mult{\subes{\lv x t + 1}{\tmsz u + 1}} \text{ and} \]
      \begin{align*}
        \submeas{t'_1}
        &= \submeas t \sqcup \lv x t \cdot \submeas{\lam y {z \es z u}}
        \sqcup \mult{\subdist{\lv x t}} \\
        &= \submeas t \sqcup \lv x t \cdot \submeas{z \es z u} 
        \sqcup \mult{\subdist{\lv x t}} \\
        &= \submeas t \sqcup \lv x t \cdot
        (\submeas z \sqcup (\lv z z + 1) \cdot \submeas u \sqcup \mult{\subes{\lv z z + 1}{\tmsz u}}) 
        \sqcup \mult{\subdist{\lv x t}} \\
        &= \submeas t \sqcup \lv x t \cdot
        (1 \cdot \submeas u \sqcup \mult{\subes 1 {\tmsz u}}) 
        \sqcup \mult{\subdist{\lv x t}} \\
        &= \submeas t \sqcup (\lv x t + 1) \cdot \submeas u 
        \sqcup \mult{\subes{\lv x t + 1}{\tmsz u}, \subdist{\lv x t}}
      \end{align*}
      $\submeas{t'_0} \subgemult \submeas{t'_1}$ because the multisets are
      the same except for $\subes{\lv x t + 1}{\tmsz u}$
      and $\subdist{\lv x t}$ on the right
      which are smaller than
      $\subes{\lv x t + 1}{\tmsz u + 1}$  on the left.

    \item If $t'_0 = t \dist x {\lam y u} \rew\absrr
      t \msub x {\lam y
      u} = t'_1$,  then we have $\submeas{t'_0} =
      \submeas t \sqcup \mult{\subdist{\lv x t}}.$ By
      \autoref{l:ms-substitution}, $\submeas{t_1'} \sqsubseteq \submeas t_{\subdistname} \sqcup
      \submeas t^{>\lv x t} \sqcup \mult{\subes k n \mid k \leq \lv x
        t, n \leq N}$ for some $N
        \in \mathbb{N}$.
      Since  $\mult{\subdist{\lv x t}} \subgemult \mult{\subes k n \mid k
      \leq \lv x t, n \leq N}$ and $\submeas t \sqsupseteq \submeas t_{\subdistname} \sqcup
      \submeas t^{>\lv x t} $, then we conclude $\submeas{t_0'} \subgemult \submeas{t_1'}$.

    \item If $t'_0 = t \es x y \rew\varrr t \msub x y = t'_1$,
      we have $\submeas{t'_0} = \submeas t \sqcup \mult{\subes  {\lv x t + 1} 1}$.
      By \autoref{l:ms-substitution},
      $\submeas {t'_1} \sqsubseteq
      \submeas t_{\subdistname} \sqcup \submeas t_{\subesname}^{>\lv x
        t} \sqcup \mult{\subes k n \mid k \leq \lv x t, n \leq N}$ for some $N
        \in \mathbb{N}$.
      Since 
      $\mult{\subes{\lv x t + 1} 1} \subgemult \mult{\subes k n \mid k
      \leq \lv x t, n \leq N}$ and $\submeas t \sqsupseteq \submeas t_{\subdistname} \sqcup \submeas t^{>\lv x t}
      $, we conclude $\submeas{t_0'} \subgemult \submeas{t_1'}$
  \qedhere
  \end{enumerate}
\end{proof}

The sequence of \autoref{ex:ms} illustrates this phenomenon:
  indeed, $\submeas{t_i} \subgemult \submeas{t_{i+1}}$ for
$0 \leq i < 5$.

\begin{cor}
  \label{l:s_term}%
  The reduction relation $\rew\nrsub$ is terminating.
\end{cor}

\paragraph{Simulations}%
\label{sec:simulation}

We show the relation between
$\rcalc$ and $\lcalc$,
as well as the atomic $\lambda$-calculus $\acalc$.
For that, we introduce a projection from $\rterms$ to
$\pterms$ implementing the unfolding of all the explicit cuts:
\[ \begin{array}{cccc}
  \maprl x \eqdef x &
  \maprl {(\lx.t)} \eqdef \lx.\maprl t &
  \maprl {(tu)} \eqdef \maprl t \maprl u &
  \maprl {(t \cut x u)} \eqdef \maprl t \msub x {\maprl u}.
\end{array} \]
Thus \eg\
$\maprl{x \es x {z \es y w} \es w {w'}}
= x \msub x {z \msub y w} \msub w {w'}
= z$.
The previous projection can be extended from list contexts to
substitutions  as follows:
$\maprl \ec  := \{ \} $ and
$\maprl {(\lc \cut x u)} := \maprl \lc \circ \msub x {\maprl u}$,
where $\circ$ denotes standard composition of substitutions.

\begin{restatable}{lem}{simullres}
  \label{l:simul-lr-es}%
  Let $t \in \rterms$.
  If $t \rew\nrep t'$, then $\maprl t \rewn\beta \maprl
  {t'}$. In particular, if either $t \rew\permlr t'$ or $t
  \rew\nrsub t'$, then $\maprl t = \maprl{t'}$.
\end{restatable}
\begin{proof}
  By induction on the stated relations.
\end{proof}

The relation $\rew \nrsub$ enjoys
\deft{full composition} on \emph{pure} terms.
Namely:
\begin{lem}
  \label{l:full-composition}
  For any  $p \in \pterms$, $t \es x p \rewp\nrsub t \msub x p$.
\end{lem}
\begin{proof}
  By induction on $p$.
  \begin{itemize}
    \item If $p = y$, then
      $t \es x y \rew\varrr t \msub x y$.
    \item If $p = p_1 p_2$, then
      \begin{align*}
        t \es x {p_1 p_2}
        &\rew\apprr t \msub x {yz} \es y {p_1} \es z {p_2} \\
        &\rewp\ih t \msub x {yz} \msub y {p_1} \es z {p_2}
        \rewp\ih t \msub x {yz} \msub y {p_1} \msub z {p_2} \\
        &= t \msub x {p_1 p_2}
      \end{align*}
    \item If $p = \lam y q$, then
      $t \es x {\lam y q}
      \rew\absrr t \dist x {\lam y {z \es z q}}
      \rewp\ih t \dist x {\lam y q}
      \rew\distrr t \msub x {\lam y q}$.
      \qedhere
  \end{itemize}
\end{proof}

This property does not hold in general if $p$ is not pure.
Indeed, if $t = xx$, then $(xx) \es x {y \es {y} {z}}$
does not $\nrsub$-reduce to $(y \es {y} {z})(y \es {y} {z})$, but to
$(yy) \es {y} z$.  However, full composition restricted to
pure terms is sufficient to prove simulation of the $\lambda$-calculus.

\begin{lem}[Simulation of the $\lambda$-calculus]
  \label{l:sim_beta}%
  Let $p_0 \in \pterms$.
  If $p_0 \rew\beta p_1$,
  then $p_0 \rew\db \rewp\nrsub p_1$.
\end{lem}

\begin{proof}
  Let $p_0 = \ctx\fc {t_0} \rew\beta \ctx\fc{t_1} = p_1$, where
  $t_0 = (\lam x q) p \rrule\beta q \msub x p = t_1$.
  By \autoref{l:full-composition}, $t_0 \rew\db q \es x p \rewp\nrsub t_1$.
  The inductive cases for $\fc$ are straightforward.
\end{proof}

The previous results have an important consequence relating the
atomic $\lambda$-calculus and the $\rcalc$-calculus.  Indeed, it can
be shown that reduction in the atomic $\lambda$-calculus is
captured by $\acalc$, and vice-versa.
More precisely, the $\rcalc$-calculus can be
simulated into the atomic $\lambda$-calculus by
\autoref{l:simul-lr-es} and~\cite{gundersen13a}, while the
converse holds by~\cite{gundersen13a} and
\autoref{l:sim_beta}.

A more structural correspondence between $\rcalc$ and $\acalc$ could also be
established. Indeed, $\rcalc$ can be first refined into a
(non-linear) calculus \emph{without} distance, let say $\rcalc'$,
so that permutation rules are integrated in the intermediate
calculus as independent rules. Then a structural relation can be
established between $\rcalc$ and $\rcalc'$ on one side, and
$\rcalc'$ and the atomic $\lambda$-calculus on the other side (as
for example done in~\cite{kesner07} for the $\lambda$-calculus).

\paragraph{Confluence}%
\label{sec:confluence}

By \autoref{l:s_term} the reduction relation $\rew\nrsub$ is terminating.
It is then  not difficult to prove confluence of $\rew\nrsub$
by using the unfolding function $\maprl \cdot$.

\begin{lem}
  \label{l:s-normal-pure}
  Let $t \in \rterms$. Then $t$ is in $\nrsub$-nf if and only if $t$ is pure.
\end{lem}
\begin{proof}
  It is obvious that a pure term is $\nrsub$-normal.
  Let us show the left-to-right implication
  and consider a $\nrsub$-normal term $t$.
  We reason by induction on $t$.
  Suppose that $t$ is not pure, so that  $t =  \ctx\fc{t_0 \cut x u}$.
  If the explicit cut is an explicit substitution, then one of the
  rules $\apprr,\distrr,\varrr$ apply, which contradicts the
  hypothesis. Otherwise the cut is a distributor, and $u$ is an
  abstraction $\lam y {u'}$, where $u'$ is in particular a
  $\nrsub$-normal form.
  By the  \ih\  $u'$ is pure so that the rule $\absrr$ applies,
  which contradicts the hypothesis again.
\end{proof}

\begin{cor}
  \label{l:nfs}
  Let $t \in \rterms$. If $t$ is in $\nrsub$-nf, then
  $\maprl t = t$.
\end{cor}

\begin{lem}
  \label{l:s_conf} \mbox{}
  The reduction relation $\rew\nrsub$ is terminating and confluent.
\end{lem}
\begin{proof}
  Termination holds by \autoref{l:s_term}.
  For confluence, suppose $t \rewn\nrsub t_1$ and $t \rewn\nrsub t_2$.
  Let $t_1 \rewn\nrsub t'_1$ and
  $t_2 \rewn\nrsub t'_2$, where
  $t'_1$ and $t'_2$ are in $\nrsub$-nf.
  Then by \autoref{l:nfs}, $\maprl{(t_i')} = t_i'$
    for both $i = 1,2$.
    By \autoref{l:simul-lr-es},
    $\maprl {(t_i')} = \maprl{t_i} = \maprl t$ so that $t_1' = t_2'$, closing the diagram.
\end{proof}

By termination of $\rew\nrsub$ any $t \in \rterms$ has
a $\nrsub$-nf, and by confluence this $\nrsub$-nf is
unique.
By \autoref{l:simul-lr-es} and \autoref{l:nfs} one obtains:
\begin{cor}
  \label{l:nfs_eq_maprl}
  Let $t \in \rterms$. Then the unique $\nrsub$-nf of $t$ is $\maprl t$.
\end{cor}

\begin{thm}
  The reduction relation $\rew\nrep$ is confluent.
\end{thm}

\begin{proof}
  Let $t \in \rterms$  such that $t \rewn\nrep
  t_1$ and $t \rewn\nrep t_2$. By
  simulation (\autoref{l:simul-lr-es}), we have $\maprl t
  \rewn\beta \maprl {t_1}$ and $\maprl t \rewn\beta \maprl {t_2}$. By 
  \autoref{l:s_conf}, there exist $t'_1$ (resp.~$t'_2$) the unique $\nrsub$-nf of $t_1$ (resp.~$t_2$).
  By \autoref{l:nfs_eq_maprl} we have
  $t'_1 = \maprl {t_1}$ and $t'_2 = \maprl {t_2}$.
  Because $\rew\beta$ is confluent, there is $u$ such
  that $\maprl {t_1} \rewn\beta u$ and $\maprl {t_2} \rewn\beta u$,
  and by
  \autoref{l:sim_beta}, $\maprl {t_1}  \rewn\nrep u$ and $\maprl {t_2} 
  \rewn\nrep u$. The diagram is then closed by $t_1 \rewn \nrsub t'_1
  = \maprl {t_1} \rewn\nrep u $ and $t_2 \rewn \nrsub t'_2  = \maprl {t_2}
  \rewn\nrep u$. Graphically,
  
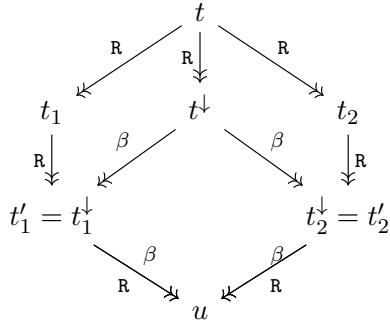
\[
    \begin{tikzcd}[baseline=(u.base)]
  & t  \arrow[->>]{dl}[above left]{\nrep}
  \arrow[->>]{dr}[above right]{\nrep}
  \arrow[->>]{d}[left]{\nrep} \\
  t_1 \arrow[->>]{d}[left]{\nrep}
  & \maprl{t} \arrow[->>]{dr}[above right]{\beta}
  \arrow[->>]{dl}[above left]{\beta}
  & t_2 \arrow[->>]{d}[right]{\nrep} \\
  t'_1 = \maprl{t_1} \arrow[->>]{dr}[below left]{\nrep}
  \arrow[->>]{dr}[bend right]{\beta}
  &
  & \maprl{t_2} = t'_2 \arrow[->>]{dl}[above right]{\beta}
  \arrow[->>]{dl}[bend left]{\nrep} \\
  & |[alias=u]| u
    \end{tikzcd}
    \tag*{\qedhere}
  \]
\end{proof}

\section{Encoding Evaluation Strategies}

In the theory of programming languages~\cite{plotkin75}, the notion of
\emph{calculus} is usually based on a non-deterministic rewriting
relation, while the deterministic notion of \emph{strategy} is
associated to a concrete machinery being able to implement a specific
evaluation procedure.  Typical evaluation strategies are call-by-name,
call-by-value and call-by-need, to name a few.

Although the atomic $\lambda$-calculus was
introduced as a technical tool to implement full laziness, only
its (non-deterministic) equational theory
was studied. In this paper we bridge the gap between
the theoretical presentation of the atomic $\lambda$-calculus and concrete
specifications of evaluation strategies.  Indeed, we use the
$\rcalc$-calculus to investigate two concrete cases: a call-by-name strategy
implementing weak head reduction, based on full substitution, and the
call-by-need fully lazy strategy, which uses linear
substitution.

In this work, we choose to implement full laziness for pure terms, that is, for
the usual $\lambda$-calculus without cuts.
Indeed, we see explicit cuts as a tool for a fully lazy implementation of the $\lambda$-calculus.
We thus keep in line with the definitions found in the literature.
Defining full laziness for terms with explicit cuts also brings technical
difficulties, which might divert from the main point: using node replication to
implement a fully lazy strategy.

We then restrict the set of terms to a subset
$\purelterms$, which simplifies the formal reasoning of explicit cuts inside
distributors. Indeed, distributors will all be of the
shape $\dist x {\ly.\ctx\llc p}$, where $p$ is a pure
term containing the constructors that have been (symbolically) shared in the
distributor, and $\llc$ is a \emph{commutative list} (defined
below).
We argue that this restriction is natural
in a weak implementation of the $\lambda$-calculus: it
is true on pure terms and is preserved through
evaluation.
We consider the following grammars.
\[ \begin{array}{llll}
  \gramTitle{Linear Cut Values} & \llset & \Coloneqq & \lx.\ctx\llc p
  \mbox{ where } y \in \domlc\llc \implies \nbocc y p = 1 \\
  \gramTitle{Commutative Lists} & \llc & \Coloneqq & \ec \mid \llc \es x p \mid \llc \dist x \llset
  \mbox{ where } \nbocc x \llc = 0 \\
  \gramTitle{Values}  & v & \Coloneqq & \lx.p \\
  \gramTitle{Restricted Terms} & \purelterms & \Coloneqq &
  x \mid v \mid \purelterms\ \purelterms \mid \purelterms \es x \purelterms
  \mid \purelterms \dist x \llset
\end{array} \]
 A term $t$ generated by any of the grammars $G$ defined above is
written $t \in G$.
Thus \eg\ $\lx.(yz) \es y \id \es z \id \in \llset$
but $\lx.(yy) \es y \id \notin \llset$,
$\ec \es x {yz} \es {x'} \id \in \llc$
but $\ec \es x {yz} \es y \id \notin \llc$,
and $(yz) \dist y \id \in \purelterms$ but
$(yz) \dist y {\lx.(yy) \es y \id} \notin \purelterms$.

The set $\llset$ is stable by the relation $\rew\nrsub$ (\autoref{l:t-stable}),
but $\purelterms$ is clearly not stable under the whole $\rew\nrep$ relation, where $\db$-reductions
may occur under abstractions. For instance, let
$t_1 = (yz) \dist y {\lx.(\ly.yy) \id} \rew\db
(yz) \dist y {\lx.(yy) \es y \id} = t_2$.
Then $t_1 \in \purelterms$ but $t_2 \notin \purelterms$, since $ \nbocc y {yy} = 2$.
However, $\purelterms$ is stable under both weak strategies
to be defined: call-by-name and call-by-need.
We factorize the proofs by proving stability for a more general relation
$\rew{\factorise}$, defined
as the relation $\rew\nrep$ with $\db$-reductions forbidden under abstractions and inside
distributors.

\begin{lem}
  \label{l:t-stable}
  If $t \in \llset$ and $t \rew\nrsub t'$, then $t' \in \llset$.
\end{lem}
\begin{proof}
  We first show a more general statement, namely that
  $t = \ctx{\llc_0}{p_0}$ with $\nbocc y {p_0} = 1$
  for every $y \in \domlc{\llc_0}$, 
  and $t \rrule\nrsub t'$ imply
  $t' = \ctx{\llc_1}{p_1}$
  with $\nbocc y {p_1} = 1$ for every $y \in \domlc{\llc_1}$.
  In the following rules $\varrr$, $\apprr$ and  $\distrr$, there is no $\lc$ context inside the explicit substitutions because lists in $\llc$
  only contain pure terms by definition.
  \begin{itemize}
    \item $t = u \es x {z} \rrule\varrr u \msub x {z} = t'$.
      This is straightforward.
    \item $t= \ctx\llc p \es x {q_1q_2} \rrule\apprr \ctx\llc {p \msub x {x_1x_2}} \es {x_1}{q_1} \es {x_2}{q_2}=t'$.
      Freshness of both $x_1$ and $x_2$ implies $\nbocc{x_1}{p \msub x {x_1x_2}}
      = \nbocc{x_2}{p \msub x {x_1x_2}} = \nbocc x p = 1$,
      and $\nbocc {x_2}{q_1} = 0$.
    \item $t=\ctx\llc p \es x {\lam {z} {p'}} \rrule\distrr \ctx\llc p
      \dist x {\lam {z} {w} \es {w} {p'}}= t'$.
      By hypothesis $\nbocc x p = 1$,
      $\nbocc x \llc =0$ and $\lam {z} {p'}$ is pure. Then,
      $\lam {z} {w \es {w} {p'}} \in \llset$ because
      $p'$ is pure and $\nbocc {w} {w} = 1$.

    \item $t = \ctx\llc p \dist x {\lam {z} {\ctx {\llc'} {p'}}}
      \rrule\absrr \ctx {\llc'}{\ctx\llc p \msub x {\lam {z} {p'}}} = t'$.
      By hypothesis $\lam {z} {\ctx{\llc'} {p'}} \in \llset$ thus $\lam
      {z} {p'}$
      and $p \msub x {\lam {z} {p'}}$ are pure. We conclude since $\nbocc x \llc=0$ by hypothesis.
  \end{itemize}
  Now we can lift the property to $\llset$
  by observing that we necessarily have $t = \lam x u \rew\nrsub \lam x {u'} $,
  where $u \rew \nrsub u'$.
  Then we conclude by the previous point.
\end{proof}

\begin{lem}
  \label{l:u-stable}
  If $t \in \purelterms$ and $t \rew{ \factorise} t'$, then $t' \in \purelterms$.
\end{lem}
\begin{proof}
    Straightforward by induction on the reduction relation.
\end{proof}

\subsection{Call-by-name}%
\label{sec:whlr}

The \deft{call-by-name} (CBN) strategy $\rew\whlr$
(\autoref{f:whlr}) is defined on the set of terms $\purelterms$ as
the union of the following relations $\rew\whdblr$ and $\rew\whslr$.
The strategy is \emph{weak} as there is no reduction under abstractions.
It is also worth noticing (as a particular case of \autoref{l:u-stable})
that $t \in \purelterms$ and $t \rew\whlr t'$ implies $t' \in \purelterms$.

\begin{figure}[h]
  \begin{mathpar}
    \inferrule*[right=\whlrdbroot]{
    t \rrule\db t'}{
    t \rew\whdblr t'} \and
    \inferrule*[right=\whlrdbapp]{
    t \rew\whdblr t'}{
    tu \rew\whdblr t'u} \and
    \inferrule*[right=\whlrdbsub]{
    t \rew\whdblr t'}{
    t \cut x u \rew\whdblr t' \cut x u}\\
    \inferrule*[right=\whlrsubroot]{
    t \rrule\nrsub t'}{
    t \rew\whslr t'} \and
    \inferrule*[right=\whlrsubapp]{
    t \rew\whslr t'}{
    tu \rew\whslr t'u} \and
    \inferrule*[right=\whlrsubsub]{
    t \rew\whslr t'}{
    u \dist x {\ly.t} \rew\whslr u \dist x {\ly.t'}}
  \end{mathpar}
  \caption{Call-by-Name Strategy}
  \label{f:whlr}
\end{figure}

\begin{exa}
  This example follows a call-by-name evaluation.
    The name of the contextual rule is written in the superscript of the
  arrow symbol, and the redex is underlined.
  \[ \begin{array}{rll}
    \underline{(\lx_1.\id (x_1\id))(\ly.(\id\id) y)}
    &\rew{}^{\textsc{db}}& \underline{(\id (x_1\id)) \es {x_1}
    {\ly.(\id\id) y}}\\
    &\rew{}^{\textsc{s}}& (\id (x_1\id)) \dist {x_1} {\ly.\underline{z \es
    z {(\id\id) y}}}\\
    &\rew{}^{\textsc{subs}}& (\id (x_1\id)) \dist {x_1} {\ly.\underline{(z_1z_2) \es {z_1}
    {\id\id} \es {z_2} y}}\\
    &\rew{}^{\textsc{subs}}& \underline{(\id (x_1\id)) \dist {x_1} {\ly.(z_1y)
    \es {z_1} {\id\id}}}\\
    &\rew{}^{\textsc{s}}& (\underline{\id ((\ly.z_1y) \id)}) \esub {z_1}
    {\id\id}\\
    &\rew{}^{\textsc{subdb}}& x_2 \es {x_2} {(\ly.z_1y) \id} \es {z_1}
    {\id\id}\\
    &\rewp{}& ((\ly.z_1y)\id) \es {z_1} {\id\id}\\
    &\rewp{}& \ly.(\id\id)y
  \end{array} \]
\end{exa}

The strategy $\rew\whlr$ does not impose duplication of all nodes in the
body of an abstraction inside the distributor: only the skeleton of the
abstraction $\ly.(\id\id)y$ is replicated.
But the strategy forbids $\db$-reductions inside explicit cuts, so that there is
no benefit gained by keeping shared terms such as $\id\id$.
Indeed, the main idea behind full laziness is that shared terms are only reduced once.
The CBN strategy, on the contrary, duplicates arguments before reducing them.
The absence of optimization is reflected by the fact that the strategy, although not deterministic,
enjoys the remarkable \emph{diamond} property, guaranteeing in particular that
all reduction sequences starting from $t$ and ending in a normal form have the
same length.

\begin{restatable}[Diamond]{prop}{diamond}
  \label{l:diamond}
  The CBN strategy enjoys the diamond property,
  \ie\ for any terms $t, u, s \in \purelterms$
  such that $t \rew\whlr u$, $t \rew \whlr s$ and $u \neq s$,
  there exists $t'$ such that $u \rew\whlr t'$ and $s \rew\whlr t'$.
\end{restatable}
\begin{proof}
  By separate inductions on the reduction relations
   $\pair{\rew\whdblr}{\rew\whdblr}$, $\pair{\rew\whslr}{\rew\whslr}$ and $\pair{\rew\whdblr}{\rew\whslr}$ \seeappendix{p:diamond}.
\end{proof}

It is worth noticing that  call-by-name in the $\lambda$-calculus
can be simulated by call-by-name in $\rcalc$.
The former can be defined by weak-head reduction, denoted $\rew\whl$, and generated by the following rules:
\begin{mathpar}
  \inferrule{\phantom{}}{(\ly.t)u \rew\whl t \msub x u} \sep\sep
  \inferrule{t \rew\whl t'}{tu \rew\whl t'u}
\end{mathpar}
There is in particular a one-to-one relation between $\beta$-steps and
$\whdblr$-steps.
\begin{lem}[Relating Call-by-Name Strategies]
  \label{l:rel_strat}\mbox{}
  \begin{itemize}
    \item Let $p_0 \in \pterms$. If $p_0 \rew\whl p_1$, then
      $p_0 \rew\whdblr\rewp\whslr p_1$ (thus $p_0 \rewp\whlr
      p_1$).
    \item Let $t_0 \in \purelterms$. If $t_0 \rew{\whdblr} t_1$, then
      $\maprl {t_0} \rew\whl \maprl {t_1}$.
      If $t_0 \rew\whslr t_1$, then $\maprl {t_0} = \maprl {t_1}$.
  \end{itemize}
\end{lem}
\begin{proof} \hfill
  \begin{itemize}
    \item By induction on $\rew\whl$.
      \begin{itemize}
        \item Let $p_0 = (\lam x p) q \rew\beta p\msub x q = p_1$. Then $(\lam x p) q \rew\whdblr p \es x q$ and we need to verify that $p \es x q \rewp{\whslr} p \msub x q$. The proof of $t \es x q \rewp{\whslr} t \msub x q$ for any $t\in \purelterms$ and pure term $q$ is by induction on $q$:
          \begin{itemize}
            \item If $q = y$ then $t \es x y \rew\whslr t \msub x y$. 
            \item If $q = q_0 q_1$ then $t \es x q \rew\whslr t \msub x {z_0z_1} \es {z_0} {q_0} \es {z_1} {q_1}$. By the \ih\ we have
              \[ \begin{array}{c}
                t \msub x {z_0z_1} \es {z_0} {q_0} \es {z_1} {q_1}
                \rewp\whslr  (t \msub x {z_0z_1} \es {z_0} {q_0}) \msub {z_1} {q_1} = 
                t \msub x {z_0 q_1} \es {z_0} {q_0}\\
                \mbox{ and  } \\
                t \msub x {z_0 q_1} \es {z_0} {q_0} \rewp\whslr t \msub x {z_0 q_1} \msub {z_0} {q_0} = t \msub x {q_0 q_1}
              \end{array} \]
              Therefore, $t \es x {q_0 q_1} \rewp\whslr t \msub x {q_0q_1}$.
            \item If $q = \lam y q'$ then $t \es x q \rew\whslr t \dist x {\lam y z \es z {q'}}$.
              By the \ih\ we have that $z \es z {q'} \rewp\whslr z \msub z {q'} = q'$ thus $t \dist x {\lam y z \es z {q'}} \rewp\whslr t \dist x {\lam y q'} \rew\whslr  t \msub x {\lam y q'}$. Therefore, $t \es x {\lam y q'} \rewp\whslr t \msub x {\lam y q'}$.
          \end{itemize} 
        \item Let $p_0 =  pq \rew\whl p'q = p_1$ where $p \rew\whl p'$. By the
          \ih\ we have that $p \rewp\whlr p'$ then, by $\whlrdbapp$ and
          $\whlrsubapp$, $p_0 =  pq \rewp\whlr p'q = p_1$.
      \end{itemize}
    \item By case analysis on $\rew\whlr$. If $t_0 \rew\whslr t_1$ then $\maprl
      {t_0} = \maprl{t_1}$ by \autoref{l:simul-lr-es}. If $t_0\rew\whdblr t_1$ then we prove the property by induction on $\rew\whdblr$.
      \begin{itemize}
        \item Let $t_0 = (\lam x t) u \rew\whdblr t \es x u = t_1$. Then
          $\maprl {t_0} = (\lam x \maprl t) \maprl u \rew\beta \maprl t \msub x {\maprl u} = \maprl{t_1}$. Note that both $\maprl t$ and $\maprl u$ are pure terms. 
        \item Let $t_0 = t u \rew\whdblr t' u = t_1$ where $t \rew\whdblr t'$. Then
          $\maprl t \rew\whl \maprl{t'}$ by the \ih, thus $\maprl {t_0} = \maprl t \maprl u \rew\whl \maprl{t'} \maprl u = \maprl {t_1}$. 
        \item Let $t_0 = t \cut x u \rew\whdblr t' \cut x u = t_1$ where $ t \rew\whdblr t'$.
          Then $\maprl t \rew\whl \maprl{t'}$ by the \ih, thus $\maprl {t_0} = \maprl t \msub x {\maprl u} \rew\whl \maprl{t'} \msub x {\maprl u} = \maprl {t_1}$. Note that the result depends on the closure of $\rew\whl$ by (implicit) substitutions, which has a straightforward proof by induction on (pure) term $\maprl {t}$, using substitution composition.
          \qedhere
      \end{itemize}
  \end{itemize}
\end{proof}

The following grammar $\nmnormg$ is meant to characterize
normal forms with respect to the $\rew\whlr$ strategy:
\[ \begin{array}{lcl}
  \nmnormg &\Coloneqq& \lam x p \mid \nmneutg \\
  \nmneutg &\Coloneqq& x \mid \nmneutg \, t
\end{array} \]
  Notice that all normal forms are pure terms: we unfold all explicit
substitutions with $\nrsub$-steps.
\begin{lem}\label{l:name_nf}
  Let $t \in \purelterms$. Then $t \in \nmnormg$ iff $t$ is in $\whlr$-nf.
\end{lem}
\begin{proof}
  The left-to-right implication is straightforward.
  The right-to-left implication proof is by induction on $\purelterms$.
  \begin{itemize}
    \item $t = x$. By definition, $t \in \nmneutg$.
    \item $t = \lam x p$. Then $t\in \nmnormg$ by definition.
    \item $t = t' u$, where $t', u \in \purelterms$.
      By definition of $\rew\whlr$,  $t$ in
      $\whlr$-nf implies $t'$   is also in $\whlr$-nf
      and $t'$ is neither an
      explicit cut nor an abstraction. Thus $t' \in \nmneutg$
      by the \ih\ and we can conclude $t \in \nmneutg$.
    \item $t = t' \esub x u$, where $t', u \in \purelterms$.
      This is not possible because there is always an
      applicable structural rule which
      would contradict $t$ to be in $\whlr$-nf.
    \item $t = t' \dist x {\lam y u}$, where $\lam y u =\lam y {\ctx \llc {p}} \in \llset $.
      Then either we can apply a structural rule
      on $u$,
      or $u$ is pure (\ie\ $\llc=\ec$) and we can apply rule $\rew\absrr$. In both cases we  would have a contradiction with
      $t$ in $\whlr$-nf.
      \qedhere
  \end{itemize}
\end{proof}

\subsection{Call-by-need}%
\label{sec:ndlr}

We now specify a deterministic strategy $\ndlr$ implementing
demand-driven computations and only linearly replicating nodes of \emph{values}
(\ie\ pure
abstractions). Given a value $\lx.p$, only the piece of structure
containing the paths between the binder $\lx$ and all the free
occurrences of $x$ in $p$, named \emph{skeleton}, will be copied. All
the other components of the abstraction will remain shared, thus
avoiding some future duplications of redexes, as explained in the
introduction.
By copying only the smallest possible substructure of the abstraction,
the strategy $\ndlr$ implements an optimization of call-by-need
called \emph{\texttt fully \texttt lazy sharing}~\cite{wadsworth71}.
First, we formally define the key notions we are going to use.

A \deft{free expression}~\cite{peytonjones87,balabonski12b}
of a \emph{pure} term $p$ is a strict subterm $q$ of $p$ such that
every free occurrence of a variable in $q$ is also
a free occurrence of the variable in $p$.  A
\deft{free expression} of $p$ is \deft{maximal} if it is not a subterm
of another free expression of $p$.  From now on, we will consider the
(ordered) list of all maximal free expressions (MFE) of a term.
Thus \eg\ the MFEs of $\ly.p$, where $p = (\id y) \id (\lz.z y w)$,
is given by the list
$\omult{\id;\id;w}$.

An \deft{$n$-ary (pure) context} ($n \geq 0$) is a
(pure) context with $n$ holes $\ec$. A skeleton is
an $n$-ary pure context where the maximal free expressions w.r.t. a
variable set $\theta$ are replaced with holes.  We introduce two
different yet equivalent notions of skeleton, together with a
  corresponding operation of skeleton extraction: we argue that they
entail respectively a big-step and a small-step semantics. More
  precisely, in the big-step semantics the skeleton extraction process
  can be seen as a meta-operator, defined by operations that are
  external to the calculus itself, as in~\cite{ariola97}, while in the
  small-step semantics the process of extraction is defined by an
  explicit reduction relation encoded in the calculus itself. 

\paragraph{A first definition of skeleton}

The first notion of skeleton runs as follows. Given  any set of variables $\theta$,
the \deft{$\theta$-skeleton} of a pure term
is an n-ary pure (\ie\ without explicit cuts) context defined as
$\skel \theta p \eqdef \ec$ if $\theta \cap \fv p = \emptyset$; otherwise:
\[ \begin{array}{c@{\hspace{3em}}c@{\hspace{3em}}c}
  \skel \theta x \eqdef x &
  \skel \theta {\lam x p} \eqdef \lam x {\skel {\theta \cup \{x\}} p}&
  \skel \theta {p_1p_2} \eqdef \skel \theta {p_1} \skel \theta {p_2}
\end{array} \]

Thus \eg\ if  $p = (I y) I (\lam z z y w)$ as above, then   $\skel
{\{y\}} p = (\ec y) \ec (\lam z z y \ec)$.

Function $\skel \theta {\_}$ is (implicitly) intended to give a
context whose holes correspond to the MFE's
that are abstracted out.  Splitting a term into a
skeleton and a multiset of MFEs is at the core of full laziness. This
can naturally be implemented in the node replication model, as
observed in~\cite{gundersen13a}.  Here, we give two different
(alternative) operational semantics to achieve it.  The first one
(\autoref{f:swbg}), written $\skelbs\theta$, uses big-step semantics
and implements the first definition of skeleton introduced above.
\begin{figure}[h]
  \begin{mathpar}
    \inferrule{x \text{ fresh}}{p \skelbs\theta x \es x {p}}
    \quad\text{when $\fv p \cap \theta = \emptyset$;
    otherwise:}\\
    \inferrule{ }{x \skelbs\theta x} \and
    \inferrule{p \skelbs{\theta \cup \{x\}} \ctx\lc{p'}} {\lam x {p}
      \skelbs\theta \ctx\lc{\lam x {p'}}} \and \inferrule{ p
      \skelbs\theta \ctx{\lc_1}{p'} \and q \skelbs\theta
      \ctx{\lc_2}{q'}} {pq \skelbs\theta
      \ctx{\lc_2}{\ctx{\lc_1}{p'q'}}}
  \end{mathpar}
\caption{Relation $\skelbs\theta$: Splitting Skeleton and MFEs in Big-Step Semantics}
\label{f:swbg}
\end{figure}

Each of the rules in \autoref{f:swbg} corresponds to a
different case in the first definition of $\theta$-skeleton.  In the
first rule, since there is no free variable of $p$ in $\theta$, $p$ is
thus an MFE kept shared in an explicit substitution.  The other
three rules correspond to each possible constructor,
where all the explicit cuts
created during the inductive cases are pushed out.
\begin{exa}
  \label{ex:skelbs}%
  Let $y,z \notin\fv{t}$, so that $t$ is the MFE of $\ly.x \es x {\lz.(yt)z}$. Then,
  \begin{mathpar}
    \inferrule{
      \inferrule{
        \inferrule{
          \inferrule{ }{y \skelbs{\{y,z\}} y}
          \and \inferrule{ }{t \skelbs{\{y,z\}} x \esub x t}}{
          yt \skelbs{\{y,z\}} (yx) \esub x t}
          \and \inferrule{ }{z \skelbs{\{y,z\}} z}}{
      (yt)z \skelbs{\{y,z\}} ((yx)z) \esub x t}}{
    \lz.(yt)z \skelbs{\{y\}} (\lz.(yx)z) \esub x t}
  \end{mathpar}
\end{exa}

\begin{lem}[Correctness of $\skelbs\theta$]
  \label{l:skelbs_correct}%
  If $p \in \pterms$,
  then $\exists n \geq 0$ s.t.  $p \skelbs\theta
  \ctx{\skel \theta p}{x_1, \dots, x_n} \es {x_i} {t_i}_{i \leq n}$,
  where $\ctx{\skel \theta p}{t_1, \dots, t_n} = p$,
  and $(x_i)_{1 \leq i \leq n}$ are fresh pairwise distinct variables.
  Moreover, $\fv {t_i} \cap \theta = \emptyset$ for all $1 \leq i \leq n$.
\end{lem}
\begin{proof}
  If $\fv p \cap \theta = \emptyset$, then $p \skelbs\theta x_1
  \es {x_1}  p$
  and $\skel \theta p = \ec$, so that $\ctx{\skel \theta p} p = p$
  trivially holds.
  Otherwise, we reason by induction on $p$:
  \begin{itemize}
    \item If $p = x$, then $\skel \theta x = x$,
      so the property holds for $n=0$ because $x \skelbs\theta
      x$.
    \item If $p = p_1p_2$, then
      $\skel \theta p = \skel \theta {p_1} \skel \theta {p_2}$.
      By the \ih\ we have
      \begin{align*}
        p_1 \skelbs\theta \ctx{\skel\theta{p_1}}{x_1,\dots,x_k} \es {x_i}{t_i}_{i \leq k}
        &\text{ and }
        p_2 \skelbs\theta \ctx{\skel\theta{p_2}}{x_{k+1},\dots,x_n} \es
        {x_i}{t_i}_{k < i \leq n},
        \text{ where}\\
        \ctx{\skel\theta{p_1}}{t_1,\dots,t_k} = p_1
        &\text{ and }
        \ctx{\skel\theta{p_2}}{t_{k+1},\dots,t_n} = p_2.
      \end{align*}
      Hence:
      \begin{align*}
        p_1p_2
        &\skelbs\theta
        (\ctx{\skel\theta{p_1}}{x_1,\dots,x_k}
        \ctx{\skel\theta{p_2}}{x_{k+1},\dots,x_n})
        \es {x_i} {t_i}_{i \leq k} \es {x_i}{t_i}_{k < i \leq n}\\
        &= \ctx{\skel\theta p}{x_1,\dots,x_n} \es {x_i}{t_i}_{i \leq n}
      \end{align*}
    \item If $p = \lx.p'$,then
      $\skel\theta p = \lx.\skel{\theta\cup\{x\}}{p'}$.
      By the \ih\  we have
      \[ p' \skelbs{\theta\cup\{x\}}
      \ctx{\skel{\theta\cup\{x\}}{p'}}{x_1,\dots,x_n} \es {x_i}{t_i}_{i \leq n}. \]
      Moreover, $x \notin \bigcup_{i\leq n} \fv{t_i}$ by definition of $\skelbs{}$ and every $x_i$ is different from $x$.
      Hence:
      \[ \lx.p' \skelbs\theta
        (\lx.\ctx{\skel{\theta\cup\{x\}}{p'}}{x_1,\dots,x_n})
        \esub {x_i}{t_i}_{i \leq n}
      = \ctx{\skel\theta{\lx.p'}}{x_1,\dots,x_n} \es {x_i}{t_i}_{i \leq n}.
      \qedhere \]
  \end{itemize}
\end{proof}

The correcteness lemma states in particular that
$p \skelbs\theta \ctx{\lc}{p'}$ implies $p'$ is pure and $\fv
\lc \cap \theta = \emptyset$.

\paragraph{An alternative definition of skeleton}

An \emph{alternative} definition of  $\theta$-skeleton can be given by
\emph{removing} the maximal free expressions from a term.
Indeed, the \deft{$\theta$-skeleton} $\askel\theta p$ of a pure term
$p$, where $\theta=\{x_1 \ldots x_n\}$,  is the n-ary pure context
$\askel \theta p$ such that $\ctx{\askel\theta p}{q_1,\dots,q_n} =
p$, for $\omult{q_1;\dots;q_n}$
the maximal free expressions of $\lx_1.
\ldots \lx_n.p$~\footnote{The order of the abstractions is irrelevant.}.
It is easy to show that both notions of skeleton are equivalent, \ie\ $\skel \theta p = \askel \theta p$.
Thus, for the same $p$ as before, $\ly.\askel {\{y\}} p = \ly.(\ec y) \ec (\lz.z y \ec)$.

The second strategy to split a term into a skeleton and its MFEs is the small-step strategy
$\rew\skelss$ on the set of terms $\llset$
(\autoref{f:swalt}), which is indeed a subset of
the reduction relation $\rew\nrep$.
It implements the second definition of skeleton we have introduced.
The relation $\rew\skelss$ makes use of four basic rules which are
parameterized by the variable
$y$ upon which the skeleton is built, written $\rrule{}^y$.
There are also two contextual (inductive) rules.
\begin{figure}[h]
  \begin{mathpar}
    \inferrule{ }{
    t \es x y \rrule\varrr^y t \msub x y} \and
    \inferrule{
    y \in \fv{p_1p_2}}{
    t \es x {p_1p_2} \rrule\apprr^y t \msub x {x_1 x_2} \es {x_1}{p_1} \es
  {x_2}{p_2}} \and
  \inferrule{
  y \in \fv{\lz.p}}{
  t \es x {\lz.p} \rrule\distrr^y
t \dist x {\lz.w \es w p}} \and
\inferrule{
y \in \fv {\lz.\ctx\llc p} \and z \notin \fv \llc}{
t \dist x {\lz.\ctx\llc p}
\rrule\absrr^y \ctx\llc{t \msub x {\lz.p}}} \and
\inferrule*[right=\skelsstop]{
t \rrule{ }^y t'  \and y \notin \fv \llc}{
\ly.\ctx\llc t \rew\skelss \ly.\ctx\llc{t'}} \and
\inferrule*[right=\skelssdeep]{
t \rew\skelss t' \and y \in \fv t   \and y \notin \fv\llc}{
  \ly.\ctx\llc{u \dist x t}
\rew\skelss \ly.\ctx\llc{u \dist x {t'}}}
\end{mathpar}
\caption{Relation $\rew\skelss$: Splitting Skeleton and MFEs in Small-Step Semantics}
\label{f:swalt}
\end{figure}
\begin{exa}
  \label{ex:skelss}%
  Let $\ly.x \es x {\lz.(yt)z}$ be as in \autoref{ex:skelbs}.
  \begin{align*}
    \ly.\underline{x \es x {\lz.(yt)z}}
    & \rew\distrr^y \ly.x \dist x {\lz.\underline{w \es w {(yt) z}}}
    \rew\apprr^z \ly.x \dist x {\lz.\underline{(w_1 w_2) \es {w_1}{yt} \es {w_2} z}}\\
    & \rew\varrr^z \ly.\underline{x \dist x {\lz. {(w_1 z) \verdear{\es {w_1}{yt}}}}}
    \rew\absrr^y \ly.\underline{(\lz.{w_1z}) \es {w_1}{yt}} \\
    & \rew\apprr^y \ly.\underline{(\lz.{(x_1x_2)z}) \es {x_1} y} \es {x_2} t
    \rew\varrr^y \ly.(\lz.{(yx_2)z}) \es {x_2} t
  \end{align*}
\end{exa}
Notice that the focused variable changes from $y$ to $z$, then back to $y$.
This is because $\rew\skelss$ constructs the innermost skeletons first.
The small-step approach allows to parametrize the reduction
relation by only one variable at a time, instead of a set.

\begin{lem}
  \label{l:llset_sw}
  If $t \in \llset$ and $t \rew\skelss t'$, then $t' \in \llset$.
\end{lem}
\begin{proof}
  For the root $\rrule{}^y$ rules, we first show that
  if $t = \ctx{\llc_0}{p_0}$ with $\nbocc {z} {p_0} =
  1$ for all $z \in \domlc{\llc_0}$,  and $t \rrule{}^y t'$,
  then $t' = \ctx{\llc_1}{p_1}$
  with $\nbocc {z} {p_1} = 1$ for all $z \in \dom{\llc_1}$.
  \begin{itemize}
    \item If $t \rrule\varrr^y t'$, this is straightforward.
    \item If $t=\ctx\llc p \es x {q_1q_2} \rrule\apprr^y
      \ctx\llc p \msub x {x_1x_2} \es {x_1}{q_1} \es {x_2}{q_2}=t'$,
      then, since $x \notin \fv\llc$, we have $\llc \msub x {x_1x_2} = \llc$.
      Moreover, freshness of $x_1,x_2$ 
      implies $\nbocc{x_1}{\ctx\llc {p'}} = \nbocc{x_2}{\ctx\llc {p'}} = \nbocc x {\ctx\llc
      p} = 1$, where $p' = p \msub x {x_1x_2}$,
      and $\nbocc{x_2}{q_1} = 0$.
    \item If $t=u \es x {\lx'.p} \rrule\distrr^y
    u \dist x {\lx'.w \es {w} p}= t'$, this is true
      by hypothesis,
      where in particular $\nbocc x u = 1$, and
      $\lx'.w \es {w} p \in \llset$ because $p$
      is pure and $\nbocc {w} {w} = 1$.
    \item If $t =\ctx{\llc_1}{p_1} \dist x
      {\lz.\ctx{\llc_2} {p_2}} \rrule\absrr^y
      \ctx{\llc_2}{\ctx{\llc_1}{p_1}\msub x {\lz.p_2}} =t'$.
      By hypothesis $\nbocc x {p_1} = 1$ and $\nbocc x {\llc_1} = 0$,
      so that $t' = \ctx{\llc_2}{\ctx{\llc_1}{p_1 \msub x
      {\lz.p_2}}} = \ctx {\llc'_1} {p'}$,
      since for all $z_1 \in \domlc{\llc_1}$ and all $z_2 \in
        \domlc{\llc_2}$, $\nbocc{z_1}{p_1} = \nbocc{z_2}{p_2} =
      1$ and, by $\alpha$-conversion,  $\nbocc {z_2} {p_1}
      = \nbocc {z_1} {p_2} = 0$
      so that $\nbocc {z'} {p_1 \msub x {\lz.p_2}} = 1$
      for any $z' \in \domlc{\llc'}$.
\end{itemize}
Then, for the contextual rules, we show by induction on $t \rew\nrsub t'$:
if $t \in \llset$ and $t \rew\nrsub t'$, then $t' \in \llset$.
\begin{itemize}
  \item In the case of $\skelsstop$,
    we have $t = \ly.\ctx\llc{t_0} \rew\nrsub \ly.\ctx\llc{t_1}$.
    By the hypothesis that $t \in \llset$ follows $t_0 =
    \ctx{\llc_0}{p_0}$.
    By the previous case analysis, $t_1 =
    \ctx{\llc_1}{p_1}$.
    Therefore $t' \in \llset$.
  \item In the case of $\skelssdeep$,
    we have $t = \ly.\ctx\llc{u \dist x {t_0}} \rew\nrsub
    \ly.\ctx\llc{u \dist x {t_1}}$.
    By the hypothesis that $t \in \llset$ follows $t_0 \in \llset$.
    By induction hypothesis, $t_1 \in \llset$.
    Therefore $t' \in \llset$.
    \qedhere
\end{itemize}
\end{proof}

\begin{lem}
  \label{l:sw-conf-sn}
  The reduction relation $\rew\skelss$ is confluent and terminating.
\end{lem}
\begin{proof}
  To show termination it is sufficient to notice that
  $t \rew\skelss t'$ implies $t \rew \nrsub t'$.
  Since $\rew\nrsub$ is terminating
  (\autoref{l:s_term}) then we conclude termination of $\rew\skelss$.
  Next, we show that $\rew\skelss$ is confluent
  by observing that it is deterministic. Indeed,
  \begin{itemize}
    \item The base rules $\rrule{}^y$ only reduce the outermost cut and they are all distinct:
      there is one rule for an outermost distributor,
      and three rules for outermost explicit substitutions,
      one for each possible form (variable, application, abstraction).
    \item Because of the condition  $y \notin \fv \llc$ in rules
      $\skelsstop$
      and $\skelssdeep$
      the base rules are always applied from right to left
      inside an abstraction.
    \item Moreover, rule $\skelssdeep$ does not overlap with
      any other rule, in particular with $\rrule\absrr^y$.
      Indeed, for a term $u \dist x {\lam
      z {\ctx\llc p}}$, there are only two possibilites.
      Either $z$ is a free variable of $\llc$, and we
      cannot apply $\rrule\absrr^y$, or $z$ is not a free variable of $\llc$, and we can apply $\rrule\absrr^y$.
      In the latter, there is in particular no cut of $\llc$
      for  which $z$ is free.
      Therefore, we cannot apply any base-rule recursively inside the distributor.
      So, we cannot apply rule $\skelssdeep$.
  \end{itemize}
  Since rule application is deterministic, then there is no possible diverging diagram, and thus confluence is trivial.
\end{proof}

Thus, from now on, we denote by $\skelbs{}_\skelss$
the function relating a term of $\llset$ to its unique $\skelss$-nf.
For instance, from \autoref{ex:skelss} we deduce
$\ly.x \esub x {\lz.(yt)z} \skelbs{}_\skelss
\ly.(\lz.{(yx_2)z}) \es {x_2} t$.
\begin{lem}
  \label{l:st-extraction}%
  If $p$ is a pure term and $\llc$ a (commutative) list context
  where $y \notin \fv\llc$, then there exists $n$ and an n-ary pure
  context $c$ such that
  \[ \ly.\ctx\llc{t \es z p}
    \rewn\skelss \ly.\ctx\llc{
  t \msub z {\ctx c {x_1,\dots,x_n}} \es {x_i} {q_i}_{1 \leq i \leq n}} \]
  where the variables $x_1, \dots, x_n$ are fresh pairwise distinct and $\omult{q_1;\dots;q_n}$ are the MFE of $\ly.p$ such that $\ctx c
  {q_1,\dots,q_n} = p$.
\end{lem}
\begin{proof}
  If $y \notin \fv p$, then $p$ is the MFE
  of $\ly.p$ and the property is satisfied by the empty reduction,
  with $n=1$, $c = \ec$, and $q_1 = p$.
  Otherwise, we reason by induction on $p$.
  \begin{itemize}
    \item If $p = y$, then $\ly.p$ has no MFE and
      $\ly.\ctx\llc{t \es z y} \rew\varrr^y \ly.\ctx\llc{t \msub z y}$.
      Then the property holds for $n=0$ and the 0-ary context $y$.
    \item If $p=p_1p_2$, then by the \ih\ on $p_2$ and on $p_1$ we have:
      \begin{align*}
        \ly.\ctx\llc {t \es z {p_1p_2}}
        & \rew\apprr^y \ly.\ctx\llc{t \msub z {z_1z_2} \es {z_1}{p_1} \es {z_2}{p_2}} \\
        &\rewn\skelss \ly.\ctx\llc{t \msub z {z_1 \ctx{c_2}{x_{k+1},\dots,x_n}} \es {z_1}{p_1} \es {x_i}{q_i}_{k < i \leq n}} \\
        &\rewn\skelss \ly.\ctx\llc{t \msub z {\ctx{c_1}{x_1,\dots,x_k} \ctx{c_2}{x_{k+1},\dots,x_n}} \es{x_i}{q_i}_{1 \leq i \leq k}
        \es {x_i}{q_i}_{k < i \leq n}} \\
        &= \ly.\ctx\llc{t \msub z {\ctx c {x_1,\dots,x_n}}
        \es{x_i}{q_i}_{1 \leq i \leq n}}
      \end{align*}
      where $ \ctx{c}{x_1,\dots,x_n} = \ctx{c_1}{ x_1,\dots,x_k}
      \ctx{c_2}{x_{k+1},\dots,x_n}$, and the variables $x_1,
      \ldots, x_n$ are chosen to be pairwise distinct.  To apply the
      \ih\ on $p_1$, we take $\llc$ to be $\ctx\llc{\ec
      \es{x_i}{q_i}_{k < i \leq n}}$, which verifies the
      hypothesis of the statement since by definition of the
      MFEs, $y \notin \cup_{k < i \leq n}
      \fv{q_i}$.  We can conclude since 
      the maximal free
      expressions of $\ly.p_1p_2$ can be computed by
      considering the MFEs of $\ly.p_1$ and $\ly.p_2$ respectively,
      \ie\  $\omult{q_1;\dots;q_n}$.
    \item If $p = \lx.p'$, then by the \ih\ on $p'$ we have:
      $\lx.z' \es {z'}{p'} \rewn\skelss
      \lx.\ctx{c'}{x_1,\dots,x_n} \es {x_i}{q_i}_{1 \leq i \leq n}$,
      where the terms $\omult{q_1;\dots;q_n}$ are the MFEs of $\lx.p'$, so in particular $x \notin \cup_{1 \leq i \leq n} \fv{q_i}$.
      We can then apply
      the \ih\ on $q_n, \ldots, q_1$, thus for $t_0 = \ly.\ctx\llc{t \es z {\lx.p'}}$ we have:
      \begin{align*}
      t_0
        & \rew\distrr^y \ly.\ctx\llc{t \dist z {\lx.z' \es {z'}{p'}}} \\
        &\rewn\skelss \ly.\ctx\llc{t \dist z {\lx.\ctx{c'}{x_1,\dots,x_n} \es{x_i}{q_i}_{1 \leq i \leq n}}} \\
        & \rew\absrr^y \ly.\ctx\llc{t \msub z {\lx.\ctx{c'}{x_1,\dots,x_n}}
        \es {x_i}{q_i}_{1 \leq i \leq n}} \\
        &\rewn\skelss \ly.\ctx\llc{t \msub z {\lx.\ctx{c'}{x_1,\dots,x_{n-1},\ctx{c_n}{x_n^1,\dots,x_n^{m_n}}}}
        \es {x_i}{q_i}_{1 \leq i < n} \es{x_n^j}{q_n^j}_{1 \leq j \leq m_n}} \\
        &\rewn\skelss \ly.\ctx\llc{t \msub z {\lx.\ctx{c'}{
              \ctx{c_1}{x_1^1,\dots,x_1^{m_1}}, \dots,
          \ctx{c_n}{x_n^1,\dots,x_n^{m_n}}}}
        \es{x_i^j}{q_n^j}_{1 \leq j \leq m_i, 1 \leq i \leq n}} \\
        &= \ly.\ctx\llc{t \msub z {\ctx c {
          x_1^1,\dots,x_n^{m_n}}}
        \es{x_i^j}{q_n^j}_{1 \leq j \leq m_i, 1 \leq i \leq n}}
      \end{align*}
      where $\ctx c {
      x_1^1,\dots,x_n^{m_n}} = \lx.\ctx{c'}{
      \ctx{c_1}{x_1^1,\dots,x_1^{m_1}}, \dots,
    \ctx{c_n}{x_n^1,\dots,x_n^{m_n}}}$
    and the variables $x_1^1$ to $x_n^{m_n}$ are taken pairwise distinct.
    To apply the \ih\ on $q_k$ ($1 \leq k \leq n$),
    we take the linear context to be $\ctx\llc{\ec \es{x_i^j}{q_i^j}_{1 \leq j \leq m_i, k < i \leq
    n}}$, which verifies the hypothesis of the statement since by definition of
    the MFEs,
    $y \notin \cup_{1 \leq j \leq m_i, k < i \leq n} \fv{q_i^j}$.
    By the \ih\  $\omult{q_i^1;\dots;q_i^{m_i}}$ are the MFEs of $\ly.q_i$ for each $i$.
    Therefore, since $\omult{q_1;\dots;q_n}$ are the MFEs
    of $\lx.p'$, the terms $\omult{q_1^1;\dots;q_n^{m_n}}$ are also the MFEs of $\ly.\lx.p'$.
    \qedhere
\end{itemize}
\end{proof}

\begin{cor}[Correctness of $\rew\skelss$]
  \label{l:skelss_correct}%
  Let $p \in \pterms$ and $\omult{q_1;\dots;q_n}$ be the MFEs of $\ly.p$.
  Then $\ly.z \es z p
  \skelbs{}_\skelss  \ly.\ctx{\askel{\{y\}} p}{x_1,\dots,x_n} \es {x_i}{q_i}_{i \leq n}$
  where the variables $x_1, \dots, x_n$ are fresh and pairwise distinct.
\end{cor}
\begin{proof}
  By \autoref{l:st-extraction}, there is an n-ary pure context $c$ such that $\lam
  y {z \es z p} \rewn\whslr t = \lam y {\ctx c {x_1,\dots,x_n}
  \es{x_i}{q_i}_{1 \leq i \leq n}}$,
  where $\omult{q_1;\dots;q_n}$ are the MFEs of $\lam y p$.
  Thus, by the alternative definition of skeleton, $c$ is  $\askel{\{y\}} p$.
  Moreover, $t$ is the $\whslr$-nf
  of $\lam y {z \es z p}$ because no more base
  $\rrule{}^y$-reduction steps can be applied
  to  the list of explicit substitutions since $y$ is not free
  in $q_1, \ldots, q_n$ by definition of MFE.
\end{proof}

From the fact that the two definitions of skeleton are
equivalent, and from both proofs of correctness
(\autoref{l:skelbs_correct} and \autoref{l:skelss_correct}), we
infer the equivalence between the small-step and the big-step splitting
semantics (\autoref{f:swalt} and \autoref{f:swbg} resp.).
Since the small-step semantics is contained in $\rcalc$,
we use it to build our call-by-need strategy using node replication.

Another interesting question concerns the splitting semantics for terms with explicit cuts.
It is not always clear what the maximal free expressions are,
as this notion depends on the position of the explicit cuts in the
term.
For instance, take the term
$t = \ly.z_1 \esub w {xy} z_2$.
What should be the MFEs of $t$? It could be $\omult{z_1;x;z_2}$, or
$\omult{z_1z_2;x}$, or even $\omult{(z_1z_2) \esub w x}$.
Similarly for the skeleton, should it be respectively (1)~$\ly.\ec \esub w {x\ec} \ec$,
(2)~$\ly.\ec\ec$ or~(3) $\ly.\ec$?
Solution~(1) proposes to keep explicit substitutions in the skeleton.
This is not coherent with the semantics of $\rcalc$ and $\acalc$,
which only substitute pure terms.
Solution~(2) consists in unfolding the explicit cuts, so that the
skeleton is pure.
This can easily be obtained by adding the following rule to the definition.
\[ \skel \theta {t \esub x u} \eqdef
\begin{cases}
  \skel {\theta \cup \{x\}} t \msub x {\skel \theta u}, &\text{ if } \theta
  \cup \fv u \neq \emptyset\\
  \skel \theta t, &\text{ otherwise}
\end{cases} \]
Indeed, this is the definition of skeleton adopted for the atomic $\lambda$-calculus
in \cite{gundersen13a},
where the authors prove that the skeleton of a term with explicit substitutions
(but without explicit distributors) can be split from the MFEs.
Unfortunately, they do not provide a splitting rewriting relation.

In cases involving explicit cuts binding no variable,
like $\esub w {xy}$ in the term $t$ above,
this definition is a cause of inefficiency: we would prefer solution~(3),
which avoids duplication of the application node.
More generally, many nodes can be duplicated inside a term to reach a bound
variable that will finally be erased.
For instance, in the term $\ly.x_1 \esub w y x_2 x_3 \dots x_n$,
$n-1$ applications nodes will need to be duplicated, and the skeleton would be
considered $\ly.\ec \ec \ec \dots \ec$ ($n$ times) following~(2), and simply
$\ly.\ec$ following~(3).
As another example, the skeleton of $\ly.(\lz.z \esub w y)x$ would be considered
$\ly.(\lz.z)\ec$ following~(2) and $\ly.\ec$ following~(3).
Unfortunately, this definition is hard to specify
inductively (and therefore in
a big-steps semantics) without modifying the term first by permuting the cuts.
Interestingly though, giving a small-steps semantics for it simply amounts to
allowing $\rew\skelss$-reduction deep inside the distributors.

\paragraph{The Call-by-Need Strategy}
The \deft{call-by-need strategy} $\rew\ndlr$ (\autoref{f:ndlr})
is defined on the set of terms $\purelterms$, by using
closure under the \emph{need contexts}, given by the grammar
\[ \ndc \Coloneqq \ec \mid \ndc t \mid \ndc \cut x t
\mid \ctxnc\ndc x \esub x \ndc\]
where $\ctxnc\ndc \_$ denotes capture-free application of contexts
(\autoref{sec:syntax}).
Like call-by-name (\autoref{sec:whlr}),
the call-by-need strategy is \emph{weak}, because no \emph{meaningful} reduction steps are performed under abstractions.

\begin{figure}[h]
  \[ \begin{array}{llll}
    \ctx\lc{\lx.p} u & \rrule\ndlrdb & \ctx\lc{p \es x u}\\
    \ctxnc\ndc x \es x {\ctx\lc{\ly.p}} & \rrule\ndlrdist &
    \ctx\lc{\ctx\llc{\ctxnc\ndc x \dist x {\ly.p'}}} &
    \,\text{ if } \ly.{z \es z p} \skelbs{}_\skelss \ly.\ctx\llc{p'}\\
    \ctxnc\ndc x \dist x v & \rrule\ndlrabs  & \ctxnc\ndc v \dist x v
  \end{array} \]
  \caption{Call-by-Need Strategy}
  \label{f:ndlr}
\end{figure}

Rule $\db$ is the same one used to define $\whlr$.
Rule $\ndlrdist$ (named after splitting) only uses node replication operations to
compute the skeleton of the abstraction, while rule $\ndlrabs$  implements
one-shot \emph{linear} substitution.
There is no rule to substitute a variable, as it is usually done in call-by-need for closed terms~\cite{ariola97}.

Linear substitution (replacing one free occurrence of a
  variable at a time) as implemented by rule~$\ndlrabs$
above is not captured by the
calculus $\rcalc$.  This shows a limitation of $\rcalc$ and $\acalc$,
both using full substitution to implement fully lazy sharing.  Yet,
the demand-driven philosophy of call-by-need is generally understood
as replacing only some desired instance of one
variable~\cite{ariola97}.  This corresponds in particular to the
behavior of abstract machines, which make explicit some of the
implementation features.  Nonetheless, remark that the substitution
used in the small-steps semantics $\rew\skelss$ is linear, thanks to
the restriction on terms.  Therefore, designing $\ndlr$ as a strategy
of a linear calculus with node replication is straightforward.

Notice that as a particular case of \autoref{l:u-stable},
$t \in \purelterms$ and $t \rew\ndlr t'$ implies $t' \in \purelterms$.
Another interesting property is that $t \rew \ndlrabs t'$ implies $\lv z t \geq \lv z {t'}$.  Moreover, $\rew\ndlr$ is deterministic.

\begin{lem}[Determinism]\label{l:need_determinism}
  The strategy $\rew \ndlr$ is deterministic.
\end{lem}
\begin{proof}
  The left hand sides of the rules $\db$, $\distrr$ and $\ndlrabs$ are
  disjoint. On the other hand, the reduction relation $\rew\skelss$
  is confluent and terminating by \autoref{l:sw-conf-sn} so that
  $\skelbs{}_\skelss$ defines a function, thus the relation $\rew
  \ndlr$ is deterministic.
\end{proof}

\begin{exa}\label{ex:ndlr}
  Let $t_0= (\lx. (\id (\id x))) (\ly.y\id)$. Needed variable occurrences are highlighted in $\anaranjear{\mbox{orange}}$.
  \[ \begin{array}{lllllllll}
    t_0 & \rew{\db}  (\underline{\id (\id x)}) \es{x}{\ly.y\id}  \rew{\db}
    \anaranjear{x_1} \es{x_1}{\underline{\id x}} \es{x}{\ly.y\id}
    \\
        &\rew{\db}
        \underline{x_1 \es{x_1}{x_2 \es{x_2}{\anaranjear{x}}} \es{x}{\lam
        y { y\id}}} 
        \rew{\ndlrdist}  \underline{x_1 \es{x_1}{x_2 \es{x_2}{\anaranjear{x}}} \dist{x}{\ly. { yz_1}}}
        \es{z_1}{\id} \\
                        & \rew{\ndlrabs}
                        \underline{x_1 \es{x_1}{\anaranjear{x_2} \es{x_2}{\ly. { yz_1}}}} \dist{x}{\ly. { yz_1}}
                        \es{z_1}{\id} \\
                        & \rew{\ndlrdist} x_1 \es{x_1}{\underline{\anaranjear{x_2} \dist{x_2}{\ly. { yz_2}}} \es{z_2}{z_1}} \dist{x}{\ly. { yz_1}}
                        \es{z_1}{\id} \\
                        & \rew{\ndlrabs} \underline{\anaranjear{x_1} \es{x_1}{(\ly. { yz_2}) \verdear{\dist{x_2}{\ly. { yz_2}} \es{z_2}{z_1}}}} \dist{x}{\ly. { yz_1}}
                        \es{z_1}{\id}   \\
                        & \rew{\ndlrdist} \underline{\anaranjear{x_1} \dist{x_1}{\ly. { yz_3} }}\es{z_3}{z_2}\dist{x_2}{\ly. { yz_2}}\es{z_2}{z_1} \dist{x}{\ly. { yz_1}}
                        \es{z_1}{\id} \\
                        & \rew{\ndlrabs} (\ly. { yz_3})  \dist{x_1}{\ly. { yz_3} }\es{z_3}{z_2}\dist{x_2}{\ly. { yz_2}}\es{z_2}{z_1} \dist{x}{\ly. { yz_1}} \es{z_1}{\id} 
      \end{array} \]
    \end{exa}

In order to characterize $\ndlr$-nfs, we use
the notion of \deft{needed free variables}, defined as:
\[ \begin{array}{l@{\hspace{1cm}}l}
  \ndv x := \{x\} &\quad
  \ndv{t \es y u} :=
  \begin{cases}
    (\ndv t \setminus \{y\}) \cup \ndv u & \text{if } y \in \ndv t \\
    \ndv t & \text{if } y \notin \ndv t
  \end{cases} \\
  \ndv{tu} := \ndv t &\quad \ndv{t \dist x u} := \ndv t \\
  \ndv{\lam x t} := \emptyset
\end{array} \]

Thus \eg\ $\ndv {x \dist y {\id} {\id}} = \{ x \}$
and $\ndv {(x y_1) \es x {z y_2}} = \{z\}$.
In particular, $x \in \ndv{t}$ implies $x \in \fv{t}$.

\begin{restatable}{lem}{ndvprop}
  \label{l:ndv_prop}%
  Let $t \in \purelterms$. Then $x \in \ndv t$ iff there exists a context $\ndc$
  such that $t = \ctxnc\ndc x$.
\end{restatable}
\begin{proof}
  By induction on $t$ for the left-to-right implication and by induction on $\ndc$
  for the other one
  \seeappendix{p:ndv_prop}.
\end{proof}

Terms of $\purelterms$ in $\ndlr$-nf can be characterized by the grammar $\normg$,
defined upon the grammar of neutral terms $\neutg$.
Notice that $\whlr$-nfs are also $\ndlr$-nfs.
\begin{align*}
  \normg &\Coloneqq \ctx\lc{\lam x t} \mid \neutg \\
  \neutg, \neutg_0 &\Coloneqq x \mid \neutg \, t \mid \neutg \cut x u \,\, x \notin
  \ndv\neutg \mid \neutg \esub x {\neutg_0} \,\, x \in \ndv \neutg
\end{align*}

\begin{restatable}{lem}{neednf}
  \label{l:need_nf}%
  Let $t \in \purelterms$. Then $t \in \normg$ iff $t$ is in $\ndlr$-nf.
\end{restatable}
\begin{proof}
  By induction on the grammars
  \seeappendix{p:need_nf}.
\end{proof}

\section{A Type System for \texorpdfstring{$\rcalc$}{λR}}
\label{sec:types}

This section introduces a quantitative type system $\sysWR$
for $\rcalc$. Non-idempotent
intersection~\cite{gardner94} has one main advantage over the
idempotent model \cite{BDS13}: it gives \emph{quantitative}
information about the length of reduction sequences to normal
forms~\cite{Carvalho07}. Indeed, not only typability and normalization
can be proved to be equivalent, but a measure based on type
derivations provides an \emph{upper bound} to normalizing reduction
sequences. This was extensively investigated in different
logical/computational
frameworks~\cite{accattoli19,BucciarelliKRV20,CarraroG14,ehrhard2012,kesner16,kesner20a}. However,
no quantitative result based on types exists in the literature for the node replication model,
including the attempts done for deep inference~\cite{guerrieri21}.
The typing rules of our system are in
themselves not surprising (see for example~\cite{kesner14} where a similar system is used for a $\lambda$-calculus with explicit substitutions interpreting the logical cut rule), but they provide
a handy quantitative characterization of fully lazy
normalization (\autoref{sec:observational}).

Types are built on the following grammar of types and
(finite) multi-types, where $\alpha$ ranges over a set of base types,
$\ans$ is a special type constant used to type
terms reducing to normal abstractions and $I$ ranges over
  finite sets.
\[ \begin{array}{llll}
  \gramTitle{Types} &   \sigma  & := &  \ans \mid \alpha \mid \MM \to \sigma\\
  \gramTitle{Multi-Types} &    \MM   & := &  \multii{\sigma_i}
\end{array} \]
We write $\msetsz\MM$ to denote  the \deft{size of a multi-type} $\MM$.
\deft{Typing environments}, written $\Gamma$, $\Delta$, $\Sigma$ are
functions from variables to multi-types, assigning the empty
multiset to all but a finite set of variables. The domain of
$\Gamma$ is given
by $\dom{\Gamma} :=\{ x \mid \Gamma(x) \neq \emult \}$.
The \deft{union of environments}, written $\Gamma \inter \Delta$, is defined
by $(\Gamma\inter\Delta)(x) := \Gamma(x) \sqcup \Delta(x)$, where
$\sqcup$ denotes multiset union. For instance,
$(x: \mult{\sigma},y: \mult{\tau}) \inter (x: \mult{\sigma},z: \mult{\tau}) =
(x:\mult{\sigma,\sigma}, y: \mult{\tau},z: \mult{\tau})$. This
notion is extended to several environments as expected, so that
$\inter_{\iI}\Gamma_i$ denotes a finite union of environments,
and the empty environment
when $I = \emptyset$. We write $\Gamma; \Delta$ for $\Gamma\inter\Delta$ when
$\dom{\Gamma} \cap \dom{\Delta} = \emptyset$.
\deft{Type judgments} have the form $\seq \Gamma t \sigma$, where
$\Gamma$ is a typing environment, $t$ is a term and $\sigma$ is a type.
\begin{figure}[!htbp]
  \begin{mathpar}
    \inferrule*[right=\ruleAxR]{ }{
    \seq {x:\mult\sigma} x \sigma} \and
    \inferrule*[right=\ruleAbsR]{
    \seq{\Gamma; x:\MM} t \sigma}{
  \seq \Gamma {\lx.t} {\MM \to \sigma}} \and
  \inferrule*[right=\ruleAppR]{
    \seq \Gamma t {\MM \to \sigma} \and
  \seq \Delta u \MM}{
\seq {\Gamma \inter \Delta} {tu} \sigma} \and
\inferrule*[right=\ruleAnsR]{ }{
\seq{\emptyset}{\lx.t}\ans} \and
\inferrule*[right=\ruleCutR]{
  \seq{\Gamma; x:\MM} t \sigma \and
\seq \Delta u \MM}{
\seq{\Gamma \inter \Delta}{t \cut x u} \sigma} \and\!
\inferrule*[right=\many]{
(\seq{\Gamma_i} t {\sigma_i})_{\iI}}{
\seq{\interii \Gamma_i} t {\multii{\sigma_i}}}
\end{mathpar}
  \caption{Typing System $\sysWR$}
  \label{f:typing-rules}
\end{figure}
A \deft{(typing) derivation} is a tree obtained by applying the
(inductive) typing rules of system $\sysWR$
(\autoref{f:typing-rules}), introduced
in~\cite{kesner14}.   The notation $\dem \Phi \Gamma t \sigma$
means there is a derivation named $\Phi$ of the judgment
$\seq \Gamma t \sigma$ in system $\sysWR$. A term $t$ is
typable in system $\sysWR$, or $\sysWR$-typable,
iff there is an environment $\Gamma$ and a type $\sigma$
such that $\dem \Phi \Gamma t \sigma$.
The \deft{size of a type derivation} $\dersz\Phi$
is defined as the number of its rules $\ruleAbsR$, $\ruleAppR$ and $\ruleAnsR$.
The typing system is \emph{relevant}:

\begin{lem}[Relevance]
  \label{l:relevance}
  If $\dem \Phi {\Gamma} t \sigma$, then
  $\dom{\Gam} \subseteq \fv{t}$.
\end{lem}
\begin{proof}
  Straightforward by induction on the typing derivation.
\end{proof}

\begin{exa}
  \label{ex:derivation}%
  The following tree is a type derivation (called $\Phi_u$) in system~$\sysWR$
  for the term $u = x \esub x {yz}$.
  \[\inferrule*[Right=\ruleCutR]{\inferrule*[Right=\ruleAxR]{
      }{\seq{x: \mult\tau} x \tau}\\
      \inferrule*[Right=\ruleAppR]{ \inferrule*[Right=\ruleAxR]{
        }{\seq{y : \mult{\mult\tau \to \tau}} y {\mult\tau \to \tau}} \\
        \inferrule*[Right=\many]{\inferrule*[Right=\ruleAxR]{
        }{\seq{z: \mult\tau} z \tau}}{\seq{z : \mult\tau} z
      {\mult\tau}}}{\inferrule*[Right=\many]{\seq{y:\mult{\mult\tau \to \tau},
    z:\mult\tau}{yz} \tau}{\seq{y:\mult{\mult\tau \to \tau},
  z:\mult\tau}{yz}{\mult\tau}}}}{\seq{y: \mult{\mult\tau \to \tau}, z:
  \mult\tau}{x \es x {yz}} \tau}\]
\end{exa}

Type derivations can be measured by triples.
We use a + operation on triples as pointwise addition: 
$(n_1,n_2,n_3) + (m_1,m_2,m_3) = (n_1+m_1,n_2+m_2,n_3+m_3)$.
These triples are computed
by a \deft{weighted derivation level} function
defined on typing derivations as $\dermeas \Phi:= \dermeasaux \Phi 1$,
where $\dermeasaux - -$ is inductively defined below.
In the cases $\ruleAbsR$, $\ruleAppR$ and $\ruleCutR$, we let $\Phi_t$ (resp. $\Phi_u$)
be the subderivation of the type of $t$ (resp. $\Phi_u$)
and in $\many$ we let $\Phi_t^i$ be the $i$-th derivation of the type of $t$ for each $\iI$.
\begin{itemize}
  \item For \ruleAxR, $\dermeasaux {\Phi_x} m = (0, 0, 1)$,
  \item For \ruleAbsR,  $\dermeasaux {\Phi_{\lx.t}} m = \dermeasaux {\Phi_t} m + (1, m, 0)$.
  \item For \ruleAnsR,  $\dermeasaux {\Phi_{\lx.t}} m = (1, m, 0)$.
  \item For \ruleAppR,
    $\dermeasaux {\Phi_{tu}} m = \dermeasaux {\Phi_t} m + \dermeasaux {\Phi_u} m + (1, m, 0)$.
  \item For \ruleCutR,
    $\dermeasaux {\Phi_{t \cut x u}} m = \dermeasaux {\Phi_t} m + \dermeasaux
    {\Phi_u} {m + \lv x t + \isES{\cut x u}}$.
  \item For \many, $\dermeasaux {\Phi_{t}} m = \sum_{\iI} \dermeasaux {\Phi^i_t} m$.
\end{itemize}

Intuitively, the first (resp. third) component of the
$3$-tuple $\dermeasaux \Phi m$ counts the number of application/abstraction
(resp. axiom) rules in the typing derivation and do not depend on $m$.  The second one takes into account the
number of application/abstraction rules as well, but \emph{weighted}
by the level of the constructor.
The $3$-tuples are ordered lexicographically.

\begin{exa}
  Take the derivation $\Phi_u$ from \autoref{ex:derivation}.
  Its measure is $\dermeas {\Phi_u} = (1,2,3)$.  Moreover, for $x \es x {yz}
  \rew\apprr (x_1x_2) \es {x_1} y \es {x_2} z$ we have $\dem {\Phi_{u'}}
  {y:\mult{\sigma}, z:\mult{\tau}} {(x_1x_2) \es {x_1} y \es {x_2} z} \tau$
  and $\dermeas {\Phi_{u'}} = (1,1,4)$.
\end{exa}

\begin{restatable}{lem}{moreweight}
  \label{l:more_weight}%
  For all derivation $\Phi$ and all $m, n \in \mathbb N$ with $m > n$,
  $\dermeasaux \Phi m = \dermeasaux \Phi n + (0, (m-n) * \dersz \Phi, 0)$.
\end{restatable}
\begin{proof}
  By induction on $\Phi$ \seeappendix{p:more_weight}.
\end{proof}

\begin{lem}[Split]
  \label{l:split}%
  Let $\dem \Phi \Delta u \MM$ such that $ \MM = \multunion_\iI \MM_i$ for $I
  \neq \emptyset$.
  Then there are derivations
  $\dem {\Phi_i} {\Delta_i} u {\MM_i}$ such that $\Delta = \inter_{\iI} \Delta_i$
  and $\dermeasaux \Phi m  = \sum_{\iI} \dermeasaux {\Phi^i} m$.
\end{lem}

\begin{proof}
  Straightforward.
\end{proof}

\section{Observational Equivalence}%
\label{sec:observational}

The type system $\sysWR$ characterizes normalization of both
$\whlr$ and $\ndlr$ strategies as follows: every typable
term normalizes and every normalisable term is typable.  In this
sense, system $\sysWR$ can be seen as a (quantitative)
\emph{model}~\cite{bucciarelli01}
of our call-by-name and call-by-need strategies. We
prove these results by studying the appropriate lemmas, notably
weighted subject reduction and weighted subject expansion.  We then
deduce observational equivalence between the $\whlr$ and the
$\ndlr$ strategies from the fact that their associated normalization
properties are both fully characterized by the same typing system.

\paragraph{Soundness}%
\label{sec:soundness-rcalc}

Soundness of system $\sysWR$ w.r.t. both $\rew\whlr$ and
$\rew\ndlr$ is investigated in this section. More precisely, we show
that typable terms are normalizing for both strategies.  In contrast
to reducibility techniques needed to show this kind of result for
simple types~\cite{gundersen13b}, soundness is achieved here by relatively simple
combinatorial arguments based again on decreasing measures.  We start
by studying the interaction between system $\sysWR$ and linear as well as
full substitution.

\begin{restatable}[Partial Substitution]{lem}{partialsub}
  \label{l:partial-substitution}%
  Let $\dem \Phi {\Gamma ; x:\MM}{\ctxnc\fc x} \sigma$
  and $\sqsubseteq$ denote multiset inclusion.  Then,
  there exists $\MN \sqsubseteq \MM$ such that for every $\dem {\Phi_u}
  \Delta u \MN$ we have $\dem \Psi {\Gamma \inter \Delta; x:\MM
  \setminus \MN}{\ctxnc\fc u} \sigma$ and, for every $m \in \mathbb N$,
  $\dermeasaux \Psi m =
  \dermeasaux \Phi m + \dermeasaux {\Phi_u} {m + \lv \ec \fc} - (0,0,\msetsz \MN)$.
\end{restatable}
\begin{proof}
  By induction on $\Phi$ \seeappendix{p:partial-substitution}.
\end{proof}

\begin{cor}[Substitution]
  \label{l:qsub}%
  If $\dem {\Phi_t} {\Gamma ; x: \MM} t \sigma$ and
  $\dem {\Phi_u} \Delta u \MM$, then
  $\dem \Phi {\Gamma \inter \Delta}{t \msub x u} \sigma$, and
  for all $m \in \mathbb N$ we have $\dermeasaux \Phi m \leq \dermeasaux {\Phi_t} m + \dermeasaux {\Phi_u} {m + \lv x t}$.
  Moreover, $\msetsz\MM > 0$ iff the inequality is strict.
\end{cor}
\begin{proof}
  The proof is by induction on $\nbocc x t$.
  If $\nbocc x t = 0$, then by the relevance Lemma~\ref{l:relevance} $\MM = \emult$,
  so that $\Phi_u$ necessarily comes from a $\many$ rule without any premise and
  thus $\Phi = \Phi_t$.
  We have $\dermeasaux \Phi m = \dermeasaux{\Phi_t} m + \dermeasaux {\Phi_u} {m  + \lv x t}$ because
  $\dermeasaux {\Phi_u} {m  + \lv x t} = (0,0,0)$.
  Otherwise, $\nbocc x t > 0$ and we can write $t$ as $\ctxnc\fc x$.
  By the partial substitution \autoref{l:partial-substitution},
  there exists $\MN \sqsubseteq \MM$ such that for all $\dem {\Phi_u^0}
  {\Delta_0} u \MN$,
  there is $\dem {\Phi'} {\Gamma \inter \Delta_0; x:\MM \setminus \MN}{\ctxnc\fc
  u} \sigma$.
  By the split \autoref{l:split}, there are derivations
  $\dem {\Phi_u^1} {\Delta_1} u \MN$ and
  $\dem {\Phi_u^2} {\Delta_2} u {\MM \setminus \MN}$,
  where $\Delta = \Delta_1 \inter \Delta_2$
  so that we can apply the partial substitution Lemma to
  $\Phi_t$ and $\Phi^1_u$, and we obtain
  $\dem {\Phi'} {\Gamma \inter \Delta_1; x:\MM \setminus \MN}{\ctxnc\fc u} \sigma$.
  Since $\lv \ec \fc \leq \lv x t$,  we have
  $\dermeasaux{\Phi'} m
  = \dermeasaux{\Phi_t} m + \dermeasaux{\Phi_u^1}{m + \lv \ec \fc} -
  (0,0,\msetsz\MN)
  \leq_{L.\ref{l:more_weight}} \dermeasaux{\Phi_t} m + \dermeasaux{\Phi_u^1}{m + \lv x t}
  - (0,0,\msetsz\MN) \leq \dermeasaux{\Phi_t} m + \dermeasaux{\Phi_u^1}{m + \lv x t}$.
  Because $x \notin \fv u$, $\nbocc x {\ctxnc\fc u} = \nbocc x t - 1$.
  We conclude by applying the \ih\  on $\Phi'$ and $\Phi_u^2$.
  We get $\dem \Phi {\Gamma \inter \Delta_1 \inter \Delta_2}{\ctxnc\fc u \msub x u}\sigma
  = \seq{\Gamma \inter \Delta}{t \msub x u} \sigma$.
  For the measure, we use $\lv x {\ctxnc\fc u} \leq \lv x t$ to  get
  $\dermeasaux \Phi m
  \leq \dermeasaux{\Phi'} m + \dermeasaux{\Phi_u^2}{m + \lv x {\ctxnc\fc u}}
  \leq \dermeasaux{\Phi_t} m + \dermeasaux{\Phi_u^1}{m + \lv x t} + \dermeasaux{\Phi_u^2}{m + \lv x t}
  =_{L.~\ref{l:split}} \dermeasaux{\Phi_t} m + \dermeasaux{\Phi_u}{m + \lv x t}$.
  If $\MM \neq \emult$, then either $\MN$ or $\MM \setminus \MN$ is non-empty,
  so at least one of the two previous inequalities is strict.
\end{proof}

The key idea to show soundness is that the measure $\dermeas \cdot$
decreases w.r.t. the reduction relations $\rew\whlr$ and $\rew\ndlr$:

\begin{restatable}[Weighted Subject Reduction for $\rew\permlr$]{lem}{srpermlr}
  \label{l:sr_permlr}%
  Let $\dem {\Phi_{t_0}} \Gamma {t_0} \sigma$ and $t_0 \rew\permlr t_1$.
  Then there exists $\dem {\Phi_{t_1}} \Gamma {t_1} \sigma$ such that
  $\dermeasaux {\Phi_{t_0}} m = \dermeasaux {\Phi_{t_1}} m$ for every $m
  \in \mathbb N$.
\end{restatable}
\begin{proof}
  Let $t_0 = \ctx\fc{t_0'}$ and $t_1 = \ctx\fc{t_1'}$,
  where $t_0' \rew\permlr t_1'$ is a root step.
  We reason by induction on $\fc$.
  We only detail one base case (where $\fc = \ec$) and one inductive case
  \seeappendix{p:sr_permlr}.
  \begin{itemize}
    \item $t_0' = t \cut x {u \cut y s} \rrule\permlr t \cut x u \cut y s =t_1'$,
      where $y \notin \fv{t}$. Let $\Phi_i$ be
      \begin{mathpar}
        \inferrule*[right=\ruleCutR]{
                \dem {\Phi_u^i} {\Delta_u^i;y:\MN_i} u {\rho_i}
                \and \inferrule*[Right=\many]{
                  \left(\dem{\Phi_s^{i,j}} {\Delta_s^{i,j}} s {\delta_j}
                \right)_{\jJi}}{
            \seq{\Delta_s^i}{s}{\MN_i}}}{
          \seq{\Delta_u^i \inter \Delta_s^i}{u \cut y s}{\rho_i}}
       \end{mathpar}
      then the typing derivation $\Phi$ is of the form
      \begin{mathpar}
        \inferrule*[right=\ruleCutR]{
          \dem {\Phi_t} {\Gamma'; x:\MM} t \sigma
          \and \inferrule*[Right=\many]{
            \left( \dem {\Phi_i}  {\Delta_u^i \inter \Delta_s^i} {u \cut y
                s} {\rho_i} \right)_{\iI}
      }{\seq{\Delta_u \inter \Delta_s}{u \cut y s}{\MM}}
    }{\seq{\Gamma' \inter\Delta_u \inter\Delta_s}{t \cut x {u \cut y s}}{\sigma}}
  \end{mathpar}
  where $\MM = \mult{\rho_i}_{\iI}$, $\MN_i= \mult{\delta_j}_{\jJi}$,
  $\Delta_u = \inter_{\iI} \Delta_u^i$,  $\Delta_s^i =  \inter_{\jJi} \Delta_s^{i,j}$, and
  $\Delta_s = \inter_{\iI} \Delta_s^i$.
  Let $\Phi_s$ be
 \[\inferrule*[Right=\many]{
    \left(
      \dem {\Phi_s^{i,j}} {\Delta_s^{i,j}} s {\delta_j}
  \right)_{\jJi, \iI}}{
  \seq{\Delta_s}{s}{\MN}}\]
  We then construct the following derivation $\Psi$.
  \begin{mathpar}
    \inferrule*[right=\ruleCutR]{
      \inferrule*[right=\ruleCutR]{
        \dem {\Phi_t} {\Gamma'; x:\MM} t \sigma
        \and \inferrule*[Right=\many]{
          \left(
            \dem {\Phi_u^i} {\Delta_u^i;y:\MN_i} u {\rho_i}
        \right)_{\iI}}{
    \seq{\Delta_u;y:\MN}{u}{\MM}}}{
  \seq{\Gamma' \inter \Delta_u; y:\MN}{t \cut x u}{\sigma}}
  \and \dem {\Phi_s} {\Delta_s}{s}{\MN}}{
\seq{\Gamma' \inter \Delta_u \inter \Delta_s}{t \cut x u \cut y s}{\sigma}}
\end{mathpar}

where $\MN = \multunion_{\iI} \MN_i$, so that
$\MN = \mult{\delta_j}_{\jJi, \iI}$. Moreover, because $y \notin \fv t$, we have that
$\lv y {t \cut x u} = \lv x t + \lv y u + \isES{\cut x u}$ if $y \in
\fv u$, and 
$\lv y {t \cut x u} = 0$ otherwise.
Now, we show that $\dermeasaux{\Phi_s^{i,j}}{m+\lv x t + \isES{\cut x u} + \lv y
u+ \isES{\cut y s}} =
\dermeasaux{\Phi_s^{i,j}}{m + \lv y {t \cut x u} + \isES{\cut y s}}$.
If $y \in \fv u$, this is immediate.
Otherwise, by the Relevance \autoref{l:relevance} we
have $J_i = \emult$ for any $i$ thus $s$ is necessarily typed with rule $\many$ and no premise, 
so that both measures are equal to (0,0,0).
Then, \begin{align*}
  \dermeasaux \Phi m
        &= \dermeasaux{\Phi_t} m
        + \sum_{\iI} \dermeasaux{\Phi_u^i}{m+\lv
        x t+\isES{\cut x u}}\\
        &\qquad + \sum_{\iI} \sum_{\jJi} \dermeasaux{\Phi_s^{i,j}}{m+\lv x t + \isES{\cut x
        u} + \lv y u + \isES{\cut y s}} \\
        & = \dermeasaux{\Phi_t} m + \sum_{\iI}
        \dermeasaux{\Phi_u^i}{m+\lv x t+\isES{\cut x u}}\\
        &\qquad + \sum_{\iI} \sum_{\jJi} \dermeasaux{\Phi_s^{i,j}}{m+\lv y {t \cut x
        u}+\isES{\cut y s}} \\
        &= \dermeasaux \Psi m
\end{align*}
\item If $\fc = u \cut x {\fc'}$, then
  we have $\dem {\Phi_u} {\Delta;x:\MM} u \sigma$
  and $\dem {\Phi'} {\Gamma'}{\ctx{\fc'}o} \MM$.
  By the \ih\  there is
  $\dem {\Psi'} {\Gamma'}{\ctx{\fc'}{o'}}{\MM}$,
  so $\dem \Psi {\Gamma' \inter\Delta}{u \cut x {\ctx{\fc'}{o'}}} \sigma$.
  Moreover,
  \begin{align*}
    \dermeasaux \Phi m
    &= \dermeasaux{\Phi_u} m
    + \dermeasaux{\Phi'}{m + \lv x u + \isES{\cut x u}}\\
    &=_{\ih} \dermeasaux{\Phi_u} m
    + \dermeasaux{\Psi'}{m + \lv x u + \isES{\cut x u}}\\
    &= \dermeasaux \Psi m
    \qedhere
  \end{align*}
\end{itemize}
\end{proof}

\begin{lem}[Weighted Subject Reduction for $\rew\nrsub$]
  \label{l:sr_nrsub}%
  Let $\dem {\Phi_{t_0}} \Gamma {t_0} \sigma$.
  If $t_0 \rew\nrsub t_1$,
  then there exists $\dem {\Phi_{t_1}} \Gamma {t_1} \sigma$
  such that $\dermeasaux {\Phi_{t_0}} m \geq  \dermeasaux {\Phi_{t_1}} m$
  for every $m \in \mathbb N$.
\end{lem}
\begin{proof}
  As remarked in \autoref{sec:operational-semantics}, $t_0 \rew \nrsub t_1$ implies
  $t_0  \rewn\permlr t' \rew{\purelr\nrsub} t_1$. By \autoref{l:sr_permlr},
  weighted subject reduction holds for
  $t_0  \rewn\permlr t'$, so  it is sufficient to show the statement for
  the relation $\rew{\purelr\nrsub}$. We reason by induction
  on this relation.
  We show the base cases for $\rrule{\purelr{\apprr}}$ and $\rrule{\purelr{\distrr}}$,
  the cases $\rrule{\purelr{\absrr}}$ and
  $\rrule{\purelr{\varrr}}$ are simply by the substitution
  \autoref{l:qsub}, and the inductive cases are straightforward by the \ih
  \begin{itemize}
    \item $t_0 = {t \es x {us}} \rew{\apprr} t \msub x {yz} \es y u \es
      z s = t_1$, where $y$ and $z$ are fresh variables. Let
      $\Phi_i$ be of the form
      \[\inferrule*[right=\ruleAppR]{\dem {\Phi^i_u} {\Gam^i_u}{u}{\MN_i  \ft  \rho_i}
        \and \dem {\Phi^i_s} {\Gam^i_s}{s}{\MN_i}}
      {\seq{\Gam^i_{u} \inter \Gam^i_{s}}{us }{ \rho_i}}\]
      then the typing derivation $\Phi_{t_0}$ is of the form
      \[ \inferrule*[right=\ruleCutR]{\dem {\Phi_t} {\Gam'; x:\MM}{t}{ \sig}
          \and \inferrule*[right=\many]{\left( \dem {\Phi_i}
          {\Gam^i_{u} \inter \Gam^i_{s}}{us }{ \rho_i} \right )_{\iI} }
        {\seq{\Gam_u \inter \Gam_s}{us }{\MM}}}
      {\seq{\Gam'  \inter  \Gam_u \inter \Gam_s}{ t {\es x {us}}}{ \sig} } \]
      where $\MM = \mult{\rho_i}_{\iI}$,
      $\Gam_u  = \inter_{\iI} \Gam^i_u$ and $\Gam_s= \inter_{\iI} \Gam^i_s$.
      We have
      \begin{align*}
        \dermeasaux{\Phi_{t_0}} m
        &= \dermeasaux {\Phi_t} m\\
        &+ \sum_{\iI} \left(\dermeasaux {\Phi_u^i}{m + \lv x t + 1}
        + \dermeasaux {\Phi_s^i}{m + \lv x t +1} + (1, m + \lv x t + 1, 0)\right)
      \end{align*}

      Let us consider $\dem{\Phi'} {\Gam'; y:\MN_u; z:\MN_s}{ t \msub {x} {yz}}{ \sig}$,
      obtained by \autoref{l:qsub} from $\Phi_t$ and
      $\dem {\Phi_{yz}} {y:\MN_u; z:\MN_s}{yz}{ \mult{\rho_i}_{\iI}}$,
      where $\MN_u = \mult{\MN_i  \ft  \rho_i}_{\iI}$ and $\MN_s =
      \multunion_{\iI} \MN_i$. Let
      \[\Phi_u =   \inferrule*[right=\ruleCutR,vcenter]{(\dem{\Phi^i_u} {\Gam^i_u}{u}{\MN_i  \ft  \rho_i}
      )_{\iI}}{\seq{\Gam_u}{u}{\MN_u}} \quad \Phi_s = \inferrule*[right=\ruleCutR,vcenter]{(\dem{\Phi^i_s}{\Gam^i_s}{s}{\MN_i})_{\iI}}{\seq{\Gam_s}{s}{\MN_s}} \]
      We construct the following derivation $\Phi_{t_1}$ with two applications
      of rule $\ruleCutR$.
      \[
        \inferrule*{\inferrule*{
            \dem {\Phi'} {\Gam'; z:\MN_s; y:\MN_u}{ t \msub {x} {yz}}{ \sig} \\
          \dem {\Phi_u} {\Gam_u}{u}{\MN_u}}
          { \seq{(\Gam'; z:\MN_s) \inter \Gam_u} {t \msub {x} {yz} \es y
          u}{\sig}} \and \dem {\Phi_s} {\Gam_s}{s}{\MN_s}
        }
      {\seq{\Gam'\inter\Gam_u\inter\Gam_s}{t \msub {x} {yz} \es y u \es z s}{\sig}} \]

    We consider two cases to conclude:
      \begin{itemize}
        \item If $\MM = \emult$, then
          \begin{align*}
            \dermeasaux {\Phi_{t_1}} m
             &= \dermeasaux {\Phi'} m \\
             & + \sum_{\iI} \dermeasaux{\Phi_u^i}{m + \lv y {t \msub
             x {yz}} + 1} + \sum_{\iI} \dermeasaux{\Phi_s^i}{m + \lv z {t \msub x
         {yz} \es y u} + 1} \\
             &= \dermeasaux {\Phi'} m
             =_{L.~\ref{l:qsub}} \dermeasaux {\Phi_t} m
             = \dermeasaux {\Phi_{t_0}} m
          \end{align*}
        \item If $\MM \neq \emult$, then
          \begin{align*}
            \dermeasaux {\Phi_{t_1}} m
            &= \dermeasaux {\Phi'} m \\
            &+ \sum_{\iI} \dermeasaux{\Phi_u^i}{m + \lv y {t \msub
            x {yz}} + 1} + \sum_{\iI} \dermeasaux{\Phi_s^i}{m + \lv z {t \msub x
        {yz} \es y u} + 1} \\
            &= \dermeasaux {\Phi'} m \\
            &+ \sum_{\iI} \dermeasaux{\Phi_u^i}{m + \lv y {t \msub x {yz}} + 1} + \sum_{\iI} \dermeasaux{\Phi_s^i}{m + \lv z {t \msub x {yz}} + 1} \\
            &=_{L.~\ref{l:level-substitution}:\ref{l:level-substitution-non-void}} \dermeasaux{\Phi'} m + \sum_{\iI} \left(\dermeasaux{\Phi_u^i}{m+\lv x t + 1} + \dermeasaux{\Phi_s^i}{m + \lv x t + 1}\right) \\
            & \leq _{L.~\ref{l:qsub}}
            \dermeasaux{\Phi_t}m  + \dermeasaux{\Phi_{yz}}{m + \lv
            x t} \\ 
            &+ \sum_{\iI} \left(\dermeasaux{\Phi_u^i}{m+\lv x t + 1} +
            \dermeasaux{\Phi_s^i}{m + \lv x t + 1}\right) \\
            & = \dermeasaux{\Phi_t} m + (1, m + \lv x t, 2) \\ 
            &+ \sum_{\iI}
            \left(\dermeasaux{\Phi_u^i}{m+\lv x t + 1} + \dermeasaux{\Phi_s^i}{m + \lv x
            t + 1} \right) \\
            & < \dermeasaux{\Phi_t} m \\
            &+ \sum_{\iI} \left(\dermeasaux{\Phi_u^i}{m+\lv x t + 1} + \dermeasaux{\Phi_s^i}{m + \lv x t + 1} + (1, m + \lv x t, 2)\right) \\
            &< \dermeasaux{\Phi_{t_0}} m.
          \end{align*}
      \end{itemize}

    \item $t_0 =  {t \es x {\lam y u}} \rew{\distrr} t \dist x {\lam y
      z \es z u} = t_1$.
      The typing derivation is of the form
      \[ \Phi_{t_0} = \inferrule*[right=\ruleCutR,vcenter]{\dem {\Phi_t} {\Gam'; x:\MM} t \sigma \and
        \dem {\Phi_{\lam y u}} {\Gam_{\lam y u}}{\lam y u} \MM}
      {\seq{\Gam'  \inter  \Gam_{\lam y u}}{t {\es x {\lam y u}}} \sigma} \]
      where  \[
        \Phi_{\lam y u} = \inferrule*[right=\many,vcenter]{
          \left( \inferrule*[right=\ruleAbsR]{ \dem {\Phi^i} {\Gam^i_{\lam y u}; y:\MN_i }{u}{\rho_i}}
          {\seq{\Gam^i_{\lam y u}}{\lam y u }{ \MN_i \ft \rho_i}} \right)_{\iI}
          \left( \inferrule*[right=\ruleAnsR]{ }
          {\seq{}{\lam y u}{ \ans}} \right)^k
        }{ \seq{\Gam_{\lam y u}}{\lam y u }{ \MM}}
      \]
      and $\MM = \mult{\MN_i \ft \rho_i}_{\iI} \multunion \mult{\underbrace{\ans,\dots,\ans}_{k}}$,
      with $\Gam_{\lam y u} = \inter_\iI \Gam^i_{\lam y u}$.
      Moreover,
      \begin{align*}
        \dermeasaux{\Phi_{t_0}} m
        & = \dermeasaux {\Phi_t} m
        + k * (1, m + \lv x t + 1, 0)\\
        & + \sum_{\iI} \left(\dermeasaux {\Phi^i}{m + \lv x t + 1} + (1,
        m + \lv x t + 1, 0)\right)
      \end{align*}
      We construct the following derivation
      \[ \Phi_{t_1} = \inferrule*[right=\ruleCutR,vcenter]{\dem {\Phi_t} {\Gam'; x:\MM}{t}{ \sig} \and
        \dem {\Phi_\lambda} {\Gam_{\lam y u}}{\lam y z \es z u  }{ \MM }}
      {\seq{\Gam' \inter   \Gam_{\lam y u}}{t \dist x {\lam y {z \es z u}}}{ \sig} } \]
      where \[
        \Phi_\lambda =
        \inferrule*[right=\many,vcenter]{
          \left( \inferrule*[right=\ruleAbsR,vcenter]{
            \dem {\Phi_\lambda^i} {\Gam^i_{\lam y u}; y:\MN_i}{z \es z u}{ \rho_i}}{
          \seq{\Gam^i_{\lam y u}}{\lam y z \es z u }{ \MN_i \ft \rho_i}}
        \right)_{\iI}
        \left( \inferrule*[right=\ruleAnsR]{ }
      {\seq{}{\lam y {z \es z u}}{ \ans}} \right)^k }{
  \seq{\Gam_{\lam y u}}{\lam y z \es z u  }{ \MM}} \]
  with $\Phi_\lambda^i$ of the form
  \[ \inferrule*[right=\ruleCutR]{
      \inferrule*[right=\ruleAxR]{  }{\seq{z:\mult{\rho_i}}{z}{\rho_i}}
      \and
    \dem{\Phi^i} {\Gam^i_{\lam y u}; y:\MN_i }{u}{\rho_i}}{
\seq{\Gam^i_{\lam y u}; y:\MN_i}{z \es z u}{ \rho_i}} \]
We have
\begin{align*}
  \dermeasaux{\Phi_{t_1}} m
  &= \dermeasaux{\Phi_t} m
  + k * (1, m + \lv x t, 0) \\
  & + \sum_{\iI} \left(\dermeasaux{\Phi^i}{m + \lv x t + 1} + (0,0,1) +
  (1,m+\lv x t,0)\right)\\
  &\leq \dermeasaux{\Phi_{t_0}} m
  \qedhere
\end{align*}
\end{itemize}
\end{proof}

\begin{lem}[Weighted Subject Reduction for $\rew\whdblr$]
  \label{l:sr_whdblr}%
  Let $\dem {\Phi_{t_0}} \Gamma {t_0} \sigma$.
  If $t_0 \rew\whdblr t_1$,
  then there exists $\dem {\Phi_{t_1}} \Gamma {t_1} \sigma$
  such that $\dermeasaux {\Phi_{t_0}} m > \dermeasaux {\Phi_{t_1}} m$ for
  every $m \in \mathbb N$.
\end{lem}
\begin{proof}
  We prove that $\dermeasaux {\Phi_{t_0}} m > \dermeasaux {\Phi_{t_1}} m$ by
  showing in particular that it is the first component of the
  $3$-tuple that strictly decreases.

  We reason by induction on the reduction relation $\rew\whdblr$.
  \begin{itemize}
    \item If $t_0 = \ctx \lc {\lam x t} u \rew{\db} \ctx \lc {t \es x u} = t_1$,
      then we reason by induction on $\lc$.
      The inductive step follows from \autoref{l:sr_permlr}, so we only show the base case $\lc = \ec$.
      The typing derivation $\Phi_{t_0}$ is of the form
      \[\inferrule*[right=\ruleAppR]{
          \inferrule*[right=\ruleAbsR]{
          \dem {\Phi_t} {\Gamma'; x:\MM}{t}{ \sigma}}{
        \seq{\Gamma'}{\lam x t}{ \MM \ft \sigma}}
      \and \dem {\Phi_u} {\Gamma_u}{u }{ \MM}}{
  \seq{\Gamma' \inter  \Gamma_u}{ (\lam x t)u}{ \sigma}} \]
  and $\dermeasaux{\Phi_{t_0}} m = \dermeasaux {\Phi_t} m + \dermeasaux {\Phi_u} m + (2, 2*m, 0)$.

  We construct the following derivation.
  \[ \Phi_{t_1} =
    \inferrule*[right=\ruleCutR,vcenter]
    {\dem {\Phi_t} {\Gamma'; x:\MM} t \sigma \and
    \dem {\Phi_u} {\Gamma_u} u \MM}
  {\seq{\Gamma' \inter \Gamma_u}{t \es x u} \sigma} \]
  We have
  \begin{align*}
    \dermeasaux{\Phi_{t_1}} m
    &= \dermeasaux {\Phi_t} m + \dermeasaux {\Phi_u}{m + \lv x t + 1}\\
    &=_{L.~\ref{l:more_weight}} \dermeasaux {\Phi_t} m + \dermeasaux {\Phi_u} m
    + (0, (\lv x t+1) * \dersz {\Phi_u}, 0) \\
    & < \dermeasaux{\Phi_{t_0}} m
  \end{align*}
  Notice that it is the first component of the first
  $3$-tuple that strictly
  decreases by  $2$.

  \item If $t_0 = tu \rew{\whdblr} t'u = t_1$, where $t \rew{\whdblr} t'$, then
    the property trivially holds by the \ih
  \item If $t_0 = t \es x u \rew{\whdblr} t' \es x u = t_1$, where $t \rew{\whdblr} t'$, then $\Gamma = \cmin{\Gamma'}{x} \inter   \Del$ and
    $\dem {\Phi_t} {\Gamma';x:\MM}{t}{\sigma}$
    and $\dem {\Phi_u} {\Delta}{u }{\MM}$.
    Moreover, $\dermeasaux{\Phi_{t_0}} m
    = \dermeasaux {\Phi_t} m +  \dermeasaux{\Phi_u}{m + \lv x t + 1}$.
    By the \ih\ we have
    $\dem {\Phi_{t'}} {\Gamma';x:\MM}{t'}{\sigma}$
    and $ \dermeasaux {\Phi_t} m >_{\ih} \dermeasaux{\Phi_{t'}} m$, where
    in particular it is the first component of the first $3$-tuple
    that strictly decreases. Derivation $\Phi_{t_1}$ is then
    obtained by rule $\ruleCutR$ from $\Phi_{t'} $ and $\Phi_u$.
    We can conclude since:
    \begin{align*}
      \dermeasaux{\Phi_{t_1}} m
      &=  \dermeasaux {\Phi_{t'}} m + \dermeasaux {\Phi_u}{m + \lv x {t'} + 1} \\
      & =_{L.~\ref{l:more_weight}} \dermeasaux {\Phi_{t'}} m +  \dermeasaux
      {\Phi_u} m + (0, (\lv x {t'}+1) * \dersz {\Phi_u}, 0) \\
      & <_{\ih} \dermeasaux {\Phi_{t}} m +  \dermeasaux {\Phi_u}{m}  + (0, (\lv
      x {t}+1) * \dersz {\Phi_u}, 0) \\
      & =_{L.~\ref{l:more_weight}} \dermeasaux {\Phi_{t}} m +  \dermeasaux {\Phi_u}{m + \lv x {t} + 1} \\
      & = \dermeasaux{\Phi_{t_0}} m
    \end{align*}
    Note that even when $\lv x {t'} > \lv x {t}$, the
    inequation $ \dermeasaux{\Phi_{t_1}} m < \dermeasaux{\Phi_{t_0}} m$ is determined by
    the strict relation between the first components of the $3$-tuples,
    that is, the unweighted number of abstraction and application rules.
    \qedhere
\end{itemize}
\end{proof}

\begin{lem}
  \label{l:typed-in-ndc}%
  Let $\dem \Phi \Gamma {\ctxnc\ndc x} \tau$.
  Then there exists $\Gamma'$, $I \neq \emptyset$ and $\mult{\sigma_i}_{\iI}$
  such that $\Gamma = \Gamma' \inter x:\mult{\sigma_i}_{\iI}$
  and for any variable $z$ there is a proof
  $\dem {\Phi'} {\Gamma'\inter z:
  \mult{\sigma_i}_{\iI}}{\ctxnc\ndc z} \tau$. In
  particular, if $z$ is fresh, then $\Gamma' \inter z:
  \mult{\sigma_i}_{\iI} = \Gamma'; z:
  \mult{\sigma_i}_{\iI}$. 
\end{lem}
\begin{proof}
  By induction on $\ndc$.
  \begin{itemize}
    \item If $\ndc = \ec$, this is straightforward by taking $\Gamma' = \emptyset$ and $\mult{\sigma_i}_{\iI} = \mult\tau$.
    \item If $\ndc = \ndc' t$ or $\ndc = \ndc' \cut x t$,
      then there is a derivation
      $\dem {\Phi'} {\Gamma_1}{\ctxnc{\ndc'} x}{\tau'}$,
      such that $\Gamma = \Gamma_1 \inter \Gamma_2$ and $\tau' = \MM \ft \tau$
      or $\tau = \tau'$, respectively.
      By the \ih\  $\Gamma_1 = \Gamma_1' \inter x:\mult{\sigma_i}_{\iI}$, so that
      $\Gamma' = \Gamma_1' \inter \Gamma_2$.
    \item If $\ndc = \ctxnc{\ndc_1} y \es y {\ndc_2}$,
      the derivation is as follows.
      \[ \inferrule*[right=\ruleCutR]{
        \seq{\Gamma_1 ; y:\mult{\rho_j}_{\jJ}}{\ctxnc{\ndc_1}y}{\tau}
        \\ \inferrule*[right=\many]{
          \left(\seq{\Gamma_j}{\ctxnc{\ndc_2} x}{\rho_j}\right)_{\jJ}}{
          \seq{\inter_\jJ \Gamma_j}{\ctxnc{\ndc_2} x}{ \mult{\rho_j}_{\jJ}}}}{
        \seq \Gamma {\ctxnc{\ndc_1} y \es y {\ctxnc{\ndc_2} x}}{\tau}} \]
      Where $\Gamma = \Gamma_1 \inter \Gamma_2$ and $\Gamma_2 = \inter_{\jJ}
      \Gamma_j$.
      By the \ih\ on $\ndc_1$, $\Gam_1;y:
        \mult{\rho_j}_{\jJ} = \Gam' \inter y:\mult{\rho_j}_{\jJ'}$ for some
        $\emptyset \neq J' \subseteq J$. Thus $J \neq \emptyset$.
      By the \ih\ on $\ndc_2$, for every $\jJ$ we have
      $\Gamma_j = \Gamma_j' \inter x:\mult{\sigma_i}_{i \in I_j}$,
      where $I_j \neq \emptyset$ and a proof $\dem {\Phi_j} {\Gamma'_j\inter z:
      \mult{\sigma_i}_{i \in I_j}}{\ctxnc{\ndc_2} z}{\rho_j}$ for a
      variable $z$.
      We then take $I = \cup_{\jJ} I_j$ and $\Gamma' = \Gamma_1 \inter_{\jJ} \Gamma_j'$.
      \qedhere
  \end{itemize} 
\end{proof}

\begin{lem}[Weighted Subject Reduction for $\ndlr$]
  \label{l:sr_flneed}%
  Let $\dem {\Phi_{t_0}} \Gamma {t_0} \sigma$.
  If $t_0 \rew\ndlr t_1$, then there exists
  $\dem {\Phi_{t_1}} \Gamma {t_1} \sigma$
  such that $\dermeasaux {\Phi_{t_0}} m > \dermeasaux {\Phi_{t_1}} m$ for
  every $m \in \mathbb N$.
\end{lem}
\begin{proof}
  We prove that $\dermeasaux {\Phi_{t_0}} m > \dermeasaux {\Phi_{t_1}} m$ by
  showing in particular that the first component of the
  first $3$-tuple strictly decreases when the reduction
  is $\db$. We reason by induction on the reduction relation, \ie\ by induction
  on the context $\ndc$ where the root reduction takes place.
  We first detail the base case when $\ndc=\ec$.
  \begin{itemize}
    \item $t_0 = \ctx\lc{\lam x t} u \rew{\db} \ctx\lc{t \es x u} = t_1$.
      This case is the same as for $\whlr$.

    \item $t_0 = \ctxnc\ndc x \es x {\lam y p} \rew\ndlrdist
      \ctx\llc{\ctxnc\ndc x \dist x {\lam y {p'}}} = t_1$,
      where $\lam y {z \es z p} \skelbs{}_\skelss \lam y {\ctx\llc{p'}}$.
      The typing derivation $\Phi_{t_0}$ is of the
      form
      \[ \inferrule*[right=\ruleCutR]{
          \dem {\Phi} {\Gamma';x:\MN}{\ctxnc\ndc x}{ \sigma}
          \and \inferrule*[Right=\many]{
          (\dem{\Phi_{\lam y p}^i} {\Delta_i}{\lam y p}{ \sigma_i})_{\iI}}{
      \seq \Delta {\lam y p}{ \MN}}}{
  \seq {\Gamma' \inter  \Delta}{\ctxnc\ndc x \es x {\lam y p}}{\sigma}} \]
  where $\MN = \mult{\sig_i}_{\iI}$, $\Del = \inter_{\iI} \Del_i$ and $\Gam =
  \Gam' \inter\Delta$. Moreover, $I \neq \emptyset$ by \autoref{l:typed-in-ndc}. For each $\sigma_i$ we build the following derivations $\Phi_{p_0}^i$:
  \begin{itemize}
    \item if $\sigma_i = \MM_i \ft \tau_i$ then
      $\Phi_{p_0}^i$ is of the form
      \[
          \inferrule*[Right=\ruleCutR]{
            \inferrule*[right=\ruleAxR]{ }{\seq{z:\mult{\tau_i}}{z}{\tau_i}}
          \\ \dem {\Phi_p^i} {\Delta_i;y:\MM_i}{p}{\tau_i}}{
      \inferrule*[right=\ruleAbsR]{\seq{\Delta_i;y:\MM_i}{z \es z p}{\tau_i}}{
  \seq{\Delta_i}{\lam y {z \es z p}}{\MM_i \ft \tau_i}}} \]
  where $\Phi_p^i $ is obtained from $\Phi_{\lam y p}^i$ by reversing the $\ruleAbsR$ rule.
\item if $\sigma_i = \ans$ then
  $\Phi_{p_0}^i =
  \inferrule*[right=\ruleAnsR,vcenter]{ }{\seq{}{\lam y {z \es z p}}{\ans}}$.
\end{itemize}
By hypothesis, $\lam y {z \es z p} \rewn\skelss \lam y {\ctx\llc{p'}}$.
Since $\rew\skelss$ is included in $\rew\nrsub$, then
we know
by \autoref{l:sr_nrsub} that there are derivations
$\dem {\Phi_{p_1}^i} {\Delta_i}{\lam y {\ctx\llc{p'}}}{\sig_i}$
such that $\dermeasaux {\Phi^i_{p_0}} m \geq \dermeasaux {\Phi^i_{p_1}} m$.
Thus, we can build the following derivation.
\[ \Phi'_{t_1} =
  \inferrule*[right=\ruleCutR]{
    \dem \Phi {\Gamma';x:\MN}{\ctxnc\ndc x} \sigma
    \and \inferrule*[Right=\many]{
    (\dem {\Phi_{p_1}^i} {\Delta_i}{\lam y {\ctx\llc{p'}}}{ \sigma_i})_{\iI}}{
\seq \Delta {\lam y {\ctx\llc{p'}}} \MN}}{
\seq {\Gamma' \inter \Delta}{\ctxnc\ndc x \dist x {\lam y {\ctx\llc{p'}}}}{\sigma}} \]
Let $n = \lv x {\ctxnc\ndc x}$.
We begin showing that
$\dermeasaux{\Phi_{\lam y p}^i}{m + n + 1} > \dermeasaux{\Phi_{p_0}^i}{m + n}$ for every $i \in I$.
There are two cases.
\begin{enumerate}
  \item If $\sigma_i = \ans$, then
    $\dermeasaux{\Phi_{\lam y p}^i}{m + n + 1} = (1,m+n+1,0)$,
    while $\dermeasaux{\Phi_{p_0}^i}{m + n} = (1,m+n,0)$.
  \item If $\sigma_i = \MM_i \ft \tau_i$, then
    $\dermeasaux{\Phi_{\lam y p}^i}{m + n + 1} = (1,m+n+1,0) + \dermeasaux{\Phi_p^i}{m+n+1}$ and
    \begin{align*}
      \dermeasaux{\Phi_{p_0}^i}{m+n}
            &= (1,m+n,0) + (0,0,1) + \dermeasaux{\Phi_p^i}{m+n+\lv z z+1} \\
            &= (1,m+n,1) + \dermeasaux{\Phi_p^i}{m+n+1}.
    \end{align*}
    So that
    $\dermeasaux{\Phi_{\lam y p}^i}{m + n + 1} > \dermeasaux{\Phi_{p_0}^i}{m + n}$ since $(1,m+n+1,0) > (1,m+n,1)$.
\end{enumerate}
Finally, we have:
\begin{align*}
  \dermeasaux{\Phi'_{t_1}} m &= \dermeasaux \Phi m + \sum_{\iI} \dermeasaux{\Phi_{p_1}^i}{m + n} \\
                             &\leq_{L.~\ref{l:sr_nrsub}} \dermeasaux \Phi m + \sum_{\iI} \dermeasaux{\Phi_{p_0}^i}{m + n} \\
                             &< \dermeasaux \Phi m + \sum_{\iI} \dermeasaux{\Phi_{\lam y p}^i}{m + n + 1} \\
                             &= \dermeasaux{\Phi_{t_0}} m
\end{align*}

By \autoref{l:sr_permlr}, we can finally construct
$\dem {\Phi_{t_1}} {\Gamma' \inter \Delta}{\ctx\llc{\ctxnc\ndc x \dist
x {\lam y {p'}}}}{\sigma}$, where $\dermeasaux {\Phi_{t_1}} m = \dermeasaux
{\Phi'_{t_1}} m$.

\item $t_0 = \ctxnc\ndc x \dist x v \rew\ndlrabs \ctxnc\ndc v \dist x v = t_1$.
  The typing derivation $\Phi_{t_0}$ is of the form
  \[
    \inferrule*[right=\ruleCutR]{
      \dem {\Phi} {\Gamma'; x: \MM}{\ctxnc\ndc x}{\sigma}
      \\ \inferrule*[Right=\many]{
      (\dem {\Phi^i_v} {\Delta_i} {v}{\tau_i})_{\iI}}{
  \dem {\Phi_v} {\Delta}{v}{\MM}}}{
\seq{\Gamma' \inter \Delta }{\ctxnc\ndc x \dist x v}{\sigma}} \]
where $\MM = \mult{\tau_i}_{\iI}$ and $\Delta = \inter_{\iI} \Delta_i$.
By \autoref{l:typed-in-ndc} we know that there is a non-empty $\MN \sqsubseteq \MM$ which types the variable $x$
in the hole of the context $\ndc$.
We can then write $\MM$ as $\MN  \sqcup \MN'$.
By \autoref{l:split} there are two derivations $\dem {\Phi_{v_1}}
{\Del_1} {v}{\MN}$ and $\dem {\Phi_{v_2}} {\Del_2} {v}{\MN'}$ such that $\Del = \Del_1 \inter \Del_2$
and $\dermeasaux {\Phi_v} m =
\dermeasaux {\Phi_{v_1}} m + \dermeasaux {\Phi_{v_2}} m $.  Using
\autoref{l:partial-substitution}, we can construct:
\[ \Phi_{t_1} =
  \inferrule*[right=\ruleCutR,vcenter]{
    \dem {\Psi} {\Gamma' \inter \Delta_1;  x:\MN'} {\ctxnc\ndc v}{\sigma} \\
  \dem {\Phi_{v_2}} {\Del_2}{v}{\MN'}} {\seq{\Gamma' \inter \Delta;
x:\MN'}{\ctxnc\ndc v \dist x v}{\sigma}} \]
We clearly have  $\lv \ec \ndc \leq \lv x {\ctxnc\ndc x}$
and, because $x \notin
\fv{v}$,
we also have $\lv x {\ctxnc\ndc v} \leq \lv x
{\ctxnc\ndc x}$. Then, 
\begin{align*}
  \dermeasaux{\Phi_{t_1}} m
        &= \dermeasaux \Psi m + \dermeasaux{\Phi_{v_2}}{m + \lv x {\ctxnc\ndc {v}}} \\
        &=_{L.~\ref{l:partial-substitution}} \dermeasaux \Phi m +
        \dermeasaux{\Phi_{v_1}}{m + \lv \ec \ndc} - (0,0,\msetsz \MN)
        + \dermeasaux{\Phi_{v_2}}{m + \lv x {\ctxnc\ndc
        {v}}} \\
        &\leq \dermeasaux \Phi m + \dermeasaux{\Phi_{v_1}}{m + \lv x {\ctxnc\ndc x}}
        - (0,0,\msetsz \MN) + \dermeasaux{\Phi_{v_2}}{m + \lv x {\ctxnc\ndc x}} \\
        & <  \dermeasaux \Phi m + \dermeasaux{\Phi_{v_1}}{m + \lv x {\ctxnc\ndc x}}
        + \dermeasaux{\Phi_{v_2}}{m + \lv x {\ctxnc\ndc x}} \\
        &= \dermeasaux{\Phi_{t_0}} m
\end{align*}
\end{itemize}

Now, we analyse all the inductive cases of the form
$t_0 = \ctx{\ndc}{t'_0} \rew\ndlr\ctx{\ndc}{t'_1} = t_1$,
where $t'_0 \rew \ndlr t'_1$.
\begin{enumerate}
  \item If $\ndc = \ndc' u$, then we have
    $\dem {\Phi_{t_0'}} {\Gamma'}{\ctx{\ndc'}{t_0'}}{\MN \ft \sigma}$
    and $\dem {\Phi_u} \Delta u \MN$.
    By the \ih\ there is $\dem {\Phi_{t_1'}} {\Gamma'}{\ctx{\ndc'}{t_1'}}{\MN \ft \sigma}$,
    so $\dem {\Phi_{t_1}} {\Gamma'\inter\Delta}{\ctx{\ndc'}{t_1'}u}{\sigma}$.
    Moreover, $\dermeasaux {\Phi_{t_0}} m
    = \dermeasaux{\Phi_{t_0'}} m + \dermeasaux{\Phi_u} m + (1,m,0)
    >_{\ih} \dermeasaux{\Phi_{t_1'}} m + \dermeasaux{\Phi_u} m + (1,m,0)
    = \dermeasaux {\Phi_{t_1}} m$.
  \item If $\ndc = \ndc' \cut x u$, then we have
    $\dem {\Phi_{t_0'}} {\Gamma';x:\MM}{\ctx{\ndc'}{t_0'}}{\sigma}$
    and $\dem {\Phi_u} {\Delta}{u}{\MM}$.
    By the \ih\ there is
    $\dem {\Phi_{t_1'}} {\Gamma';x:\MM}{\ctx{\ndc'}{t_1'}}{\sigma}$,
    so $\dem {\Phi_{t_1}} {\Gamma'\inter\Delta}{\ctx{\ndc'}{t_1'} \cut x u}{\sigma}$. We distinguish three different cases:
    \begin{itemize}
      \item If $t'_0 \rew \ndlr t'_1$  is a $\db$-step: then we know by the \ih\ that
        $\dermeasaux{\Phi_{t_0'}} m > \dermeasaux{\Phi_{t_1'}} m $ strictly decreases
        the first component of the first $3$-tuple. We then have
        \begin{align*}
          \dermeasaux{\Phi_{t_1}} m
            &= \dermeasaux {\Phi_{t'_1}} m + \dermeasaux {\Phi_u}{m + \lv x {\ctx{\ndc'}{t'_1}} + 1} \\
            & =_{L.~\ref{l:more_weight}} \dermeasaux {\Phi_{t'_1}} m +
            \dermeasaux {\Phi_u} m + (0, (\lv x {\ctx{\ndc'}{t'_1}}+1) * \dersz {\Phi_u}, 0) \\
            & <_{\ih} \dermeasaux {\Phi_{t'_0}} m +  \dermeasaux {\Phi_u}{m}  +
            (0, (\lv x {\ctx{\ndc'}{t'_0}}+1) * \dersz {\Phi_u}, 0) \\
            & =_{L.~\ref{l:more_weight}} \dermeasaux {\Phi_{t'_0}} m +  \dermeasaux {\Phi_u}{m + \lv x {\ctx{\ndc'}{t'_0}} + 1} \\
            & = \dermeasaux{\Phi_{t_0}} m
        \end{align*}
      \item If $t'_0 \rew \ndlr t'_1$  is an $\ndlrdist$-step, then
        $t'_0 \rewn{\nrsub}  t'_1$, so that
        $\ctx{\ndc'}{t'_0}\rewn \nrsub\ctx{\ndc'}{t'_1}$,
        and thus 
        $\lv x {\ctx{\ndc'}{t'_0}} \geq \lv x {\ctx{\ndc'}{t'_1}}$
        holds by \autoref{l:stability-levels}. We then conclude by:
        \begin{align*}
          \dermeasaux{\Phi_{t_1}} m
            &=  \dermeasaux {\Phi_{t'_1}} m + \dermeasaux {\Phi_u}{m + \lv x {\ctx{\ndc'}{t'_1}} + 1} \\
            & <_{\ih} \dermeasaux {\Phi_{t'_0}} m + \dermeasaux {\Phi_u}{m + \lv x {\ctx{\ndc'}{t'_1}} + 1} \\ 
            & \leq \dermeasaux {\Phi_{t'_0}} m + \dermeasaux {\Phi_u}{m + \lv x {\ctx{\ndc'}{t'_0}} + 1} \\ 
            & = \dermeasaux{\Phi_{t_0}} m
        \end{align*}
      \item If $t'_0 \rew \ndlr t'_1$ is a $\ndlrabs$-step, then we
        know that $\ctx {\ndc'}{t'_0}\rew \ndlr\ctx
        {\ndc'}{t'_1}$ also holds, then $\lv x {\ctx
        {\ndc'}{t'_0}} \geq \lv x {\ctx {\ndc'}{t'_1}}$.
        We conclude as before. 
    \end{itemize}
  \item If $\ndc = \ctxnc{\ndc_1} x \es x {\ndc_2}$, then
    we have $\dem {\Phi_1} {\Delta;x:\MM}{\ctxnc{\ndc_1} x}{\sigma}$
    and $\dem {\Phi_{t_0'}} {\Gamma'}{\ctx{\ndc_2}{t_0'}}{\MM}$.
    By the \ih\ there is
    $\dem {\Phi_{t_1'}} {\Gamma'}{\ctx{\ndc_2}{t_1'}}{\MM}$,
    so $\dem {\Phi_{t_1}} {\Gamma'\inter \Delta}{\ctxnc{\ndc_1} x \cut x
    {\ctx{\ndc_2}{t_1'}}}{\sigma}$.
    Moreover, $\dermeasaux {\Phi_{t_0}} m
    = \dermeasaux{\Phi_1} m + \dermeasaux{\Phi_{t_0'}}{m + \lv x
    {\ctxnc{\ndc_1} x} + \isES{\cut x {\ctx{\ndc_2}{t_0'}}}}
    >_{\ih} \dermeasaux{\Phi_1} m + \dermeasaux{\Phi_{t_1'}}{m + \lv x
    {\ctxnc{\ndc_1} x} + \isES{\cut x {\ctx{\ndc_2}{t_1'}}}}
    = \dermeasaux {\Phi_{t_1}} m$.
    \qedhere
\end{enumerate}
\end{proof}

\begin{exa}
  Consider the following reduction sequence:
  \[ (\id (x_1 \id)) \es {x_1} {\ly. \id y}
    \rew{\db}   x_2 \es {x_2} {x_1 \id} \es {x_1} {\ly. \id y} 
    \rew{\ndlrdist}  x_2 \es  {x_2} {x_1 \id} \dist {x_1} {\ly. z y}
  \es z \id \]
  We have $\dem {\Phi_1} \emptyset {(\id (x_1 \id)) \es {x_1} {\ly. \id y}} \ans$
  with $\Phi_1$ of the form \[
    \inferrule*[Right=\ruleCutR]{
      \inferrule*[right=\ruleAppR]{
      \Phi_\id \\ \\ \\ \Phi_{x_1\id}}{
    \seq{x_1:\mult{\mult{\ans} \to \ans}}{\id (x_1 \id)} \ans} \\
    \inferrule*[Right=\ruleAppR]{
      \Phi_\id \\
      \inferrule*[Right=\many]{
        \inferrule*[Right=\ruleAxR]{ }{
      \seq{y:\mult{\ans}}{y}{\ans}}}{
  \seq{y:\mult{\ans}}{y}{\mult{\ans}}}}{
  \inferrule*[Right=\ruleAbsR]{
  \seq{y:\mult{\ans}}{\id y}{\ans}}{
  \inferrule*[Right=\many]{\seq\emptyset{\ly. \id y}{\mult{\ans} \to \ans}}{
\seq\emptyset{\ly. \id y}{\mult{\mult{\ans} \to \ans}}}}}}{
\seq\emptyset{(\id (x_1 \id)) \es {x_1} {\ly. \id y}}\ans}\]
where \[ \Phi_\id =
  \inferrule*[right=\ruleAbsR,vcenter]{
    \inferrule*[Right=\ruleAxR]{ }{
  \seq{x:\mult\ans} x \ans}}{
\seq\emptyset{\id}{\mult{\ans} \to \ans}} \]
  and
  \[ \Phi_{x_1\id} =
\inferrule*[Right=\ruleAppR,vcenter]{
  \inferrule*[right=\ruleAxR]{ }{
  \seq{x_1}{\mult{\mult{\ans} \to \ans}}{x_1}{\mult{\ans} \to \ans}} \\
  \inferrule*[Right=\many]{
    \inferrule*[Right=\ruleAnsR]{ }{
\seq\emptyset{\id}{\ans}}}{\seq\emptyset{\id}{\mult{\ans}}}}{
\inferrule*[Right=\many]{
\seq{x_1:\mult{\mult{\ans} \to \ans}}{x_1 \id}{\ans}}{
\seq{x_1:\mult{\mult{\ans} \to \ans}}{x_1 \id}{\mult{\ans}}}} \]
We also have $\dem {\Phi_2} \emptyset {x_2 \es {x_2} {x_1 \id} \es {x_1}
{\ly. \id y}} \ans$ with $\Phi_2$ of the form
\[ \inferrule*[Right=\ruleCutR]{
    \inferrule*[right=\ruleCutR]{
      \inferrule*[right=\ruleAxR]{ }{
      \seq{x_2: \mult{\ans}}{x_2}{\ans}} \\
    \Phi_{x_1\id}}{
  \seq{x_1:\mult{\mult{\ans} \to \ans}}{x_2 \es {x_2} {x_1 \id}}{\ans}} \\
  \inferrule*[Right=\ruleAppR]{
    \Phi_\id \\
    \inferrule*[Right=\many]{
      \inferrule*[Right=\ruleAxR]{ }{
    \seq{y:\mult{\ans}}{y}{\ans}}}{
\seq{y:\mult{\ans}}{y}{\mult{\ans}}}}{
\inferrule*[Right=\ruleAbsR]{
\seq{y:\mult{\ans}}{\id y}{\ans}}{
\inferrule*[Right=\many]{
\seq\emptyset{\ly. \id y}{\mult{\ans} \to \ans}}{
\seq\emptyset{\ly. \id y}{\mult{\mult{\ans} \to \ans}}}}}}{
\seq\emptyset{x_2 \es {x_2} {x_1 \id} \es {x_1} {\ly. \id y}}{\ans}}\]
Concerning the measures we have
$\dermeas {\Phi_1} = (7,10,4) >  (5,13,4) = \dermeas {\Phi_2} $.
The first element of the $3$-tuple decreases from $7$ to $5$ because
we lost an abstraction and an application constructors
during $\db$-reduction. Note also that in $\Phi_1$ we
have $\dermeasaux {\Phi_{x_1\id}} 1 = (2,2,1)$ while in $\Phi_2$ we have
$\dermeasaux {\Phi_{x_1\id}} 2 = (2,4,1) = \dermeasaux {\Phi_{x_1\id}} 1 + (0,\dersz {\Phi_{x_1\id}},0)$. 
Besides, we have $\dem {\Phi_3} \emptyset {x_2 \es  {x_2} {x_1 \id} \dist {x_1} {\ly. z y}
\es z \id} \ans$ where $\Phi_3$ is of the form
\[ \inferrule*[Right=\ruleCutR]{
    \inferrule*[right=\ruleCutR]{
      \inferrule*[right=\ruleCutR]{
        \inferrule*[right=\ruleAxR]{ }{
        \seq{x_2: \mult{\ans}}{x_2}{\ans}} \and
      \Phi_{x_1\id}}{
    \seq{x_1:\mult{\mult{\ans} \to \ans}}{x_2 \es {x_2} {x_1 \id}}{\ans}}
    \and  \Phi'_3}{\seq{z:\mult{\mult{\ans} \to \ans}}{x_2 \es {x_2} {x_1 \id} \dist {x_1} {\ly. z
y}}{\ans}} \and
\Phi_\id}{
\seq\emptyset{x_2 \es  {x_2} {x_1 \id} \dist {x_1} {\ly. z y} \es z
  \id}{\ans}} \]
where $\Phi'_3$ is
\[ \inferrule*[Right=\ruleAppR]{
      \inferrule*[right=\ruleAxR]{ }{
      \seq{z:\mult{\mult{\ans} \to \ans}}{z}{ \mult{\ans} \to \ans}} \and
      \inferrule*[Right=\many]{
        \inferrule*[Right=\ruleAxR]{ }{
      \seq{y:\mult{\ans}}{y}{\ans}}}{
  \seq{y:\mult{\ans}}{y}{\mult{\ans}}}}{
  \inferrule*[Right=\ruleAbsR]{
  \seq{z:\mult{\mult{\ans} \to \ans}; y:\mult{\ans}}{z y}{\ans}}{
  \inferrule*[Right=\many]{
  \seq{z:\mult{\mult{\ans} \to \ans}}{\ly. z y}{\mult{\ans} \to \ans}}{
\seq{z:\mult{\mult{\ans} \to \ans}}{\ly. z y}{\mult{\mult{\ans} \to \ans}}}}}\]
Therefore $\dermeas{\Phi_3} = (5,11,5) <   (5,13,4) = \dermeas{\Phi_2}$,
where the second element of the $3$-tuple
has decreased  from $13$ to $11$ because
two nodes of the term $\ly.\id y$, namely the binder and
the application, have moved from the explicit substitution of level $3$
to the distributor of level
$2$.
\end{exa}

\begin{thm}[Typability implies $\whlr$-Normalization]
  \label{t:soundness-name}%
  Let $\dem {\Phi_t} \Gamma t \sigma$.
  Then $t$ is $\whlr$-normalizing.
  Moreover, the first element of $\dermeas {\Phi_t}$ is an upper bound for
  the number of $\db$-steps to $\whlr$-nf.
\end{thm}
\begin{proof}
  Suppose $t$ is not $\whlr$-normalizing.
  Since $\rew\nrsub$ is terminating by \autoref{l:s_term},
  then every infinite $\rew\whlr$-reduction sequence starting at $t$
  must necessarily have an infinite number of $\db$-steps.
  Moreover, all terms in such an infinite sequence are typed
  by \autoref{l:sr_whdblr} and
  \autoref{l:sr_nrsub}.
  Therefore, these lemmas
  guarantee that all
  $\db$/$\nrsub$ reduction steps involved in such
  $\rew\whlr$-reduction sequence do not increase the measure $\dermeas
  \cdot$, and that, in particular, $\db$-steps strictly decrease it by
  decreasing the first element of the triple.
  This leads to a contradiction because the order $>$ on 3-tuples $\dermeas
  \cdot$ is well-founded. Then $t$ is necessarily $\whlr$-normalizing.
\end{proof}

\begin{thm}[Typability implies $\ndlr$-Normalization]
  \label{t:soundness-need}%
  Let $\dem {\Phi_t} \Gamma t \sigma$.
  Then $t$ is $\ndlr$-normalizing. Moreover, the first element of
  $\dermeas {\Phi_t}$ is an upper bound for the number of $\db$-steps to $\ndlr$-nf.
\end{thm}
\begin{proof} The property trivially holds
  by \autoref{l:sr_flneed} since the lexicographic order
  on $3$-tuples is well-founded.
\end{proof}

\paragraph{Completeness}%
\label{sec:completeness}

We address here completeness of system $\sysWR$ with respect to
$\rew\whlr$ and $\rew\ndlr$. More precisely, we show that
normalizing terms in each strategy are typable.  The
basic property in showing that
consists in guaranteeing that normal forms are typable.

\begin{restatable}[$\ndlr$-nfs are Typable]{lem}{ndlrnormtyp}
  \label{l:ndlr-norm-typ}%
  Let $t$ be in $\ndlr$-nf.
  Then there exists a derivation $\dem \Phi \Gamma t \tau$
  such that for any $x \notin \ndv t$, $\Gamma(x) = \emult$.
\end{restatable}
\begin{proof}
  By \autoref{l:need_nf} we can reason by induction on the grammar $\normg$
  \seeappendix{p:ndlr-norm-typ}.
\end{proof}

\begin{exa}
  Remember that $\ndv{(xy_1) \esub x z {y_1}} = \{z\}$ and note that $\ndv{xy_1}
  = \{x\}$.
  \begin{mathpar}
    \inferrule{
      \inferrule{
        \inferrule{ }{\seq{x:\emult \to \tau} x {\emult \to \tau}}
      \and \inferrule{ }{\seq \emptyset {y_1} \emult}}
      {\seq{x:\emult \to \tau} {xy_1} \tau}
      \and \inferrule{
        \inferrule{ }{\seq{z:\emult \to \emult \to \tau} z {\emult \to \emult \to \tau}}
      \and \inferrule{ }{\seq \emptyset {y_2} \emult}}
    {\seq{z:\emult \to \emult \to \tau}{zy_2}{\emult \to \tau}}}
    {\seq{z:\emult \to \emult \to \tau}{(xy_1) \esub x {zy_2}} \tau}
  \end{mathpar}
\end{exa}

Because $\whlr$-nfs are also $\ndlr$-nfs,
we infer the following corollary for free.
\begin{cor}[$\whlr$-nfs are Typable]
  \label{l:whlr-norm-typ}%
  Let $t$ be in $\whlr$-nf. Then there is a derivation $\dem \Phi \Gamma t \tau$.
\end{cor}

We need lemmas stating the behavior of partial and full (anti-)substitution
w.r.t. typing.

\begin{restatable}[Partial Anti-Substitution]{lem}{partialasub}
  \label{l:partial-anti-substitution}%
  Let $\ctxnc\fc x$ and $u$ be terms s.t. $x \notin \fv{u}$ and $\dem \Phi
  \Gamma {\ctxnc\fc u} \sigma$.
  Then $\exists \Gamma'$, $\exists \Delta$, $\exists \MM$, $\exists \Phi'$, $\exists \Phi_u$
  s.t. $\Gamma = \Gamma' \inter \Delta$, $\dem {\Phi'} {\Gamma' \inter
  x:\MM}{\ctxnc\fc x} \sigma$ and $\dem {\Phi_u} \Delta u \MM$.
\end{restatable}
\begin{proof}
  By induction on $\fc$
  \seeappendix{p:partial-anti-substitution}.
\end{proof}

\begin{cor}[Anti-Substitution]
  \label{l:anti-substitution}%
  Let $u$ be a term s.t. $x \notin \fv{u}$ and
  $\dem \Phi \Gamma {t \msub x u} \sigma$. Then $\exists \Gamma'$,
  $\exists \Delta$, $\exists \MM$, $\exists \Phi'$,  $\exists \Phi_u$
  s.t. $\Gamma = \Gamma' \inter \Delta$, $\dem {\Phi'} {\Gamma';x:\MM} t \sigma$
  and $\dem {\Phi_u} \Delta u \MM$.
\end{cor}
\begin{proof}
  The proof is by induction on $\nbocc x t$.
  \begin{itemize}
    \item If $\nbocc x t = 0$ then $t \msub x u = t$ and, by
      \autoref{l:relevance}, $x \notin \dom{\Gamma}$ then $\Gamma =
      \Gamma;x:\emult$. Therefore, for $\Gamma' := \Gamma$, $\Delta :=
      \emptyset$, $\MM = \emult$, $\Phi' := \Phi$ and $\Phi_u := \inferrule{
      }{\seq{ } u \emult}$ the result holds.
    \item If $\nbocc x t \geq 1$ then let $\ctxnc\fc x$ such that $t\msub{x}{u} = \ctxnc\fc u$.
      For any fresh $y$, we have that $t\msub{x}{u} = \ctxnc\fc y \msub{y}{u}$
      where $\ctxnc\fc y = t'\msub{x}{u}$ s.t. $t = t'\msub{y}{x}$. Note that
      $\nbocc x {t'} < \nbocc x t$. Then by
      \autoref{l:partial-anti-substitution} $\exists \Gamma''$, $\exists
      \Delta'$, $\exists \MN$, $\exists \Phi''$,  $\exists \Phi'_u$ s.t. $\Gamma
      = \Gamma'' \inter \Delta'$, $\dem {\Phi''} {\Gamma'' \inter y:\MN}{\ctxnc\fc y}
      \sigma$ and $\dem {\Phi'_u} {\Delta'} u \MN$ where, by freshness of $y$, $\Gamma'' \inter y:\MN = \Gamma'';y:\MN$. Therefore, by the \ih\ on $\Phi''$ $\exists \Gamma'''$,
      $\exists \Delta''$, $\exists \MN'$, $\exists \Phi'''$,  $\exists \Phi''_u$
      s.t. $\Gamma'';y:\MN = \Gamma''' \inter \Delta''$, $\dem {\Phi'''}
      {\Gamma''';x:\MN'}{t'} \sigma$
      and $\dem {\Phi''_u} {\Delta''} u \MN'$. By freshness of $y$ and
      relevance,
      we have $y \notin \dom{\Delta''}$.
      Then $\Gamma''' = \Gamma^{iv};y:\MN$
      where $\Gamma'' = \Gamma^{iv} \inter \Delta''$. From $\Phi'''$ and
      \autoref{l:partial-substitution}   we have $\dem {\Phi'}
      {(\Gamma^{iv};x:\MN') \inter x:\MN} t \sigma$ while from $\Phi'_u$ and
      $\Phi''_u$ we obtain $\dem {\Phi_u} {\Delta' \inter \Delta''} u
      \MN\multunion\MN'$. Finally, for $\Gamma' := \Gamma^{iv}$, $\Delta :=
      \Delta' \inter \Delta''$, $\MM = \MN\multunion\MN'$ the result holds, since $(\Gamma^{iv};x:\MN') \inter x:\MN = \Gamma';x:\MM$ and $\Gamma' \inter \Delta = \Gamma^{iv} \inter \Delta'' \inter \Delta' = \Gamma'' \inter \Delta' = \Gamma$.
      \qedhere
  \end{itemize}
\end{proof}

To achieve completeness, we show that typing is preserved by anti-reduction.
\begin{lem}[Subject Expansion]
  \label{l:se-general}%
  Let $\dem {\Phi_{t_1}} \Gamma {t_1} \sigma$.
  If $t_0 \rew\redrule t_1$, where $\redrule \in \{\permlr, \nrsub, \whdblr, \ndlr\}$,
  then there exists $\dem {\Phi_{t_0}} \Gamma {t_0} \sigma$.
\end{lem}
\begin{proof}
  The proof is by induction on $\rew\redrule$ and uses
  \autoref{l:partial-anti-substitution} and \autoref{l:anti-substitution}.
  We detail some interesting cases of the proof.
  In all the cases shown, we suppose that the list context $\lc$ of the
  general rule is empty ($\lc = \ec$), since we can use subject expansion for
  $\rew\permlr$ to manipulate it.

  \begin{itemize}
    \item $t_0 = t \es {x} {us}
      \rrule\nrsub t \msub x {yz} \es y u \es z s = t_1$.
      Then $\Phi_{t_1}$ is of the form
      \[\inferrule*{ \inferrule*{ \dem {\Phi} {\Gamma';z:\MN_s;y:\MN_u}{t \msub {x} {yz}}{
            \sigma} \\ \dem {\Phi_u} {\Delta_u}{u}{\MN_u}}{\seq{(\Gamma' \inter \Delta_u);z:\MN_s}{t
        \msub {x} {yz} \es y u}{\sigma}} \\ \dem {\Phi_s} {\Delta_s}{s}{\MN_s}}{\seq{\Gamma' \inter
  \Delta_u \inter \Delta_s}{t \msub {x} {yz} \es y u \es z s}{ \sigma}}\]
  where $\Gamma = \Gamma' \inter \Delta_u \inter \Delta_s$. Also $(\Gamma' ;z:\MN_s) \inter \Delta_u =
  (\Gamma' \inter \Delta_u);z:\MN_s$ since $z \notin \dom{\Delta_u}$ by the Relevance \autoref{l:relevance}.
  By \autoref{l:anti-substitution} $\exists \Gamma''$,
  $\exists \Delta$, $\exists \MM$, $\exists \Phi'$,  $\exists \Phi_{yz}$ s.t.
  $\Gamma';z:\MN_s;y:\MN_u = \Gamma'' \inter \Delta$, $\dem {\Phi'} {\Gamma'';x:\MM}{t}{
  \sigma}$ and $\dem {\Phi_{yz}} \Delta {yz} \MM$. By freshness of $y,z$ and
  \autoref{l:relevance} we have that $y,z \notin \dom{\Gamma''} \cup \{x\}$.
  Then $\Gamma'' = \Gamma'$ and $\Delta = z:\MN_s;y:\MN_u$. From $\Phi_{yz}$,
  $\Phi_u$, $\Phi_s$ and \autoref{l:partial-substitution} we obtain $\dem
  {\Phi_{us}} {\Delta_u \inter \Delta_s}{us}{\MM}$ and construct $\Phi_{t_0}$ as:
  \[\Phi_{t_0} = \inferrule{\dem {\Phi'} {\Gamma';x:\MM}{t}{ \sigma} \\ \dem {\Phi_{us}}
  {\Delta_u \inter \Delta_s}{us}{\MM}}{\seq{\Gamma' \inter \Delta_u \inter \Delta_s}{t \es {x} {us}}{\sigma}}\]

\item If $t_0 = (\lam x t) u \rew{\db} t \es x u = t_1$.
  Then $\Phi_{t_1}$ is of the form
  \[ \inferrule*[right=\ruleCutR]{\dem {\Phi_t} {\Gamma'; x:\MM}{t}{ \sigma} \\
    \dem {\Phi_u} {\Gamma_u}{u }{ \MM}}
  {\seq{\Gamma' \inter  \Gamma_u}{ t\es x u}{ \sigma} } \]
  Therefore, we construct $\Phi_{t_0}$ as follows:
  \[\inferrule*[right=\ruleAppR]{
      \inferrule*[right=\ruleAbsR]{
      \dem {\Phi_t} {\Gamma'; x:\MM}{t}{ \sigma}}{
    \seq{\Gamma'}{\lam x t}{ \MM \ft \sigma}}
  \\ \dem {\Phi_u} {\Gamma_u}{u }{ \MM}}{
\seq{\Gamma' \inter  \Gamma_u}{ (\lam x t)u}{ \sigma}} \]

    \item $t_0 = \ctxnc\ndc x \es x {\lam y p} \rew\ndlrdist
      \ctx\llc{\ctxnc\ndc x \dist x {\lam y {p'}}} = t_1$,
      where $\lam y {z \es z p} \skelbs{}_\skelss \lam y {\ctx\llc{p'}}$.
      By subject expansion for $\rew\permlr$, there is $\dem {\Phi_{t'_1}}
      \Gamma {\ctxnc\ndc x \dist x {\lam y {\ctx\llc{p'}}}} \sigma$ and
      it is of the form
      \[\inferrule*[right=\ruleCutR]{
          \dem {\Phi} {\Gamma';x:\MN}{\ctxnc\ndc x}{ \sigma}
          \\ \inferrule*[Right=\many]{
          (\dem {\Phi_i} {\Delta_i}{\lam y {\ctx\llc{p'}}}{ \sigma_i})_{\iI}}{
      \seq \Delta {\lam y {\ctx\llc{p'}}} \MN}}{
      \seq {\Gamma' \inter \Delta}{\ctxnc\ndc x \dist x {\lam y
  {\ctx\llc{p'}}}}{\sigma}} \]
  where $\Del = \inter_{\iI} \Del_i$ and $\MN = \mult{\sig_i}_{\iI}$ where, by \autoref{l:typed-in-ndc}, $\MN \neq \emult$.
  Then, for each $\iI$ we have by subject expansion for $\rew\nrsub$
(of which $\rew\skelss$ is a subrelation) that $\dem {\Phi'_i} {\Delta_i}{\lam y {z \es z p} }{ \sigma_i}$ which has
  two different shapes, depending on $\sig_i$.
  \begin{enumerate}
    \item If $\sig_i = \MM_i \ft \tau_i$ then $\Phi'_i$ is of the form
      \[\inferrule*[right=\ruleAbsR]{
          \inferrule*[Right=\ruleCutR]{
            \inferrule*[right=\ruleAxR]{ }{\seq{z:\mult{\tau_i}}{z}{\tau_i}}
          \\ \dem {\Phi_p^i} {\Delta_i;y:\MM_i}{p}{\tau_i}}{
      \seq{\Delta_i;y:\MM_i}{z \es z p}{\tau_i}}}{
  \seq{\Delta_i}{\lam y {z \es z p}}{\MM_i \ft \tau_i}} \]
  Therefore we have $\Psi_i$ of the form \[\inferrule*[right=\ruleAbsR]{
    \dem {\Phi_p^i} {\Delta_i;y:\MM_i}{p}{\tau_i}}{
\seq{\Delta_i}{\lam y p }{\MM_i \ft \tau_i}} \]
\item If $\sigma_i = \ans$ then $\Del_i = \emptyset$ and we obtain $\Psi_i$ of the form
  $\inferrule*[right=\ruleAnsR,vcenter]{ }{\seq{}{\lam y p}{\ans}}$.
\end{enumerate}
We can then construct $\Phi_{t_0}$ as follows
\[\inferrule*[right=\ruleCutR]{
    \dem {\Phi} {\Gamma';x:\MN}{\ctxnc\ndc x}{ \sigma}
    \and \inferrule*[Right=\many]{
    (\dem {\Psi_i} {\Delta_i}{\lam y p}{ \sigma_i})_{\iI}}{
\seq \Delta {\lam y p}{ \MN}}}{
\seq {\Gamma' \inter \Delta}{\ctxnc\ndc x \es x {\lam y p}}{\sigma}} \]

\item $t_0 = \ctxnc\ndc x \dist x v \rew\ndlrabs \ctxnc\ndc v \dist x v = t_1$. Then $\Phi_{t_1}$ is of the form
  \[ \inferrule*[right=\ruleCutR]{
      \dem {\Phi} {\Gamma';  x:\MN'} {\ctxnc\ndc v}{\sigma} \\
    \dem {\Phi'_{v}} {\Del'}{v}{\MN'}}
  {\seq{\Gamma' \inter \Del'}{\ctxnc\ndc v \dist x v}{\sigma}} \]
  By \autoref{l:partial-anti-substitution} $\exists \Gam''$,
  $\exists\Del''$, $\exists\MN''$, $\exists\Phi'$,  $\exists \Phi''_v$
  s.t. $\Gamma';  x:\MN = \Gam'' \inter \Del'$, $\dem {\Phi'} {\Gam'' \inter
  x:\MN''}{\ctxnc\ndc x}{ \sigma}$ and $\dem {\Phi''_v} {\Delta''} {v}{
  \MN''}$. From $x \notin \fv v$ and the Relevance \autoref{l:relevance} we
  have that $x \notin \dom{\Del''}$. Thus $\Gam'' = \Gam'''; x:\MN'$ and then
  $\Gam'' \inter x:\MN' = \Gam''';x:\MN$ where $\MN = \MN' \sqcup \MN''$. From $\Phi'_v$
  and $\Phi''_v$ derivations we obtain $\dem {\Phi_v} {\Delta} {v}{
  \MN}$, where $\Delta = \Del' \inter \Del'' $. Then $\Phi_{t_0}$ is of the form
  \[\inferrule*[right=\ruleCutR]{
      \dem {\Phi'} {\Gamma'''; x: \MN}{\ctxnc\ndc x}{\sigma}
    \and \dem {\Phi_v} {\Delta}{v}{\MN}}{
\seq{\Gamma''' \inter \Delta }{\ctxnc\ndc x \dist x v}{\sigma}} \]
where $\Gamma''' \inter \Delta = \Gam' \inter \Del'$.
\qedhere
\end{itemize}
\end{proof}

\begin{thm}[$\whlr$-Normalization implies Typability]
  \label{t:completeness-name}%
  Let $t$ be a term. If $t$ is $\whlr$-normalizing,
  then $t$ is $\sysWR$-typable.
\end{thm}
\begin{proof}
  Let $t$ be $\whlr$-normalizing.
  Then $t \rew\whlr^n u$
  and $u$ is a $\whlr$-nf.  We reason by induction on $n$.
  If $n = 0$, then $t = u$ is typable by \autoref{l:whlr-norm-typ}.
  Otherwise, we have $t \rew{\whlr} t' \rew{\whlr}^{n-1} u$.
  By the \ih\ $t'$ is typable and thus by
  \autoref{l:se-general} (because $\rew\whslr$ is included in $\rew\nrsub$),
  $t$ turns out to be  also typable.
\end{proof}

\begin{thm}[$\ndlr$-Normalization implies Typability]
  \label{t:completeness-need}%
  Let $t$ be a term. If $t$ is $\ndlr$-normalizing, then
  $t$ is $\sysWR$-typable.
\end{thm}

\begin{proof}
  Similar to the previous proof but using \autoref{l:ndlr-norm-typ}
  instead of \autoref{l:whlr-norm-typ}.
\end{proof}

Summing up, Theorems~\ref{t:soundness-name}, \ref{t:completeness-name}, \ref{t:soundness-need} and \ref{t:completeness-need} give:

\begin{thm}
  \label{l:need_correct}
  $t \in \rterms$ is $\whlr$-normalizing
  \textit{iff} $t$ is $\ndlr$-normalizing
  \textit{iff} $t$ is $\sysWR$-typable.
\end{thm}

All the technical tools are now available to
conclude observational equivalence between our two
evaluation strategies based on node replication.
Let $\rewrel$ be any reduction
notion on $\rterms$. Then, two terms $t,u \in \rterms$  are said to
be \deft{$\rewrel$-observationally equivalent}, written $t \obseq_{\rewrel} u$, if for any context $\fc$, $\ctx\fc t$ is $\rewrel$-normalizing
\textit{iff} $\ctx\fc u$ is $\rewrel$-normalizing.

\begin{thm}
  For all terms $t, u \in \rterms$, $t$ and $u$
  are $\whlr$-observationally equivalent iff $t$ and $u$
  are $\ndlr$-observationally equivalent.
\end{thm}

\begin{proof} The proof uses \autoref{l:need_correct}.
  Indeed, we have $t \obseq_\whlr u$
   \textit{iff} ($\ctx\fc t$ is $\whlr$-normalizing
   \textit{iff} $\ctx\fc u$ is $\whlr$-normalizing for  any context $\fc$)
   \textit{iff} 
   ($\ctx\fc t$ is $\ndlr$-normalizing
   \textit{iff} $\ctx\fc u$ is $\ndlr$-normalizing for  any context $\fc$)
   \textit{iff} $t \obseq_\ndlr u$. 
\end{proof}

\section{Related Works and Conclusion}%
\label{sec:conclusion-atomic}

Several calculi with explicit substitutions (ES) bridge the gap between formal higher-order calculi and
concrete implementations of programming languages (see a survey
in~\cite{kesner07a}).  The first of such calculi,
\eg~\cite{abadi90,BlooRoselx}, were all based on \emph{structural}
substitution, in the sense that the ES operator is syntactically
propagated step-by-step through the term structure until a variable is reached, when the
substitution finally takes place.  The correspondence between ES and
Linear Logic Proof-Nets~\cite{DCKP03} led to the more recent notion of
calculi \emph{at a
distance}~\cite{accattoli10,accattoli14a,PNESAccattoli18},
enlightening a natural and new application of the Curry-Howard
interpretation.  These calculi implement linear/partial substitution
\emph{at a distance}, where the search of variable occurrences is abstracted
out with context-based rewriting rules, and thus no ES propagation
rules are necessary.  A third model was introduced by the seminal work
of Gundersen, Heijltjes, and Parigot~\cite{gundersen13a,gundersen13b},
introducing the atomic $\lambda$-calculus to implement node replication.

Inspired by the last approach we introduced the calculus $\rcalc$,
capturing the essence of node replication.  In contrast
to~\cite{gundersen13a}, we work with an implicit (structural)
mechanism of weakening and contraction, a design choice which aims at
focusing and highlighting the node replication model, which is the
core of our calculus, so that we obtain a rather simple and natural
formalism used in particular to specify evaluation strategies. Indeed,
besides the proof of the main operational meta-level properties of our
calculus (confluence, termination of the substitution calculus,
simulations), we use linear and non-linear versions of $\rcalc$ to
specify evaluation strategies based on node replication, namely
call-by-name and call-by-need evaluation strategies. In
particular, we provided simple tools to prove correctness of these
reduction strategies. This was achieved in the framework of our
concise calculus $\rcalc$, based not only on an implicit
treatment of weakening and contraction, but also on the notion of
commuting conversions by means of \emph{distance}. Indeed, the
treatment of weakening, contraction, and commuting conversions
result in a heavy machinery for the atomic $\lambda$-calculus, which
would make the correctness of these strategies much more involved. 

Moreover, characterisation of termination of different strategies based on $\rcalc$ were achieved with a rather standard type system and, surprisingly, no deep inference system was necessary at this point. This is of interest, since our type system is equipping full-laziness with a well-known denotational semantics. We think that this would be difficult to achieve in the framework of the atomic  $\lambda$-calculus.

The first description of call-by-need was given by
Wadsworth~\cite{wadsworth71}, where reduction is performed on
\emph{graphs} instead of terms. Weak call-by-need on \emph{terms} was
then introduced by Ariola and Felleisen~\cite{ariola97}, and by
Maraist, Odersky and Wadler~\cite{maraist98,maraist99}. Reformulations
were introduced by Accattoli, Barenbaum and Mazza
~\cite{accattoli14b} and by Chang and Felleisen~\cite{chang14}. Our
call-by-need strategy is inspired by the calculus
in~\cite{accattoli14b}, which uses the distance
paradigm~\cite{accattoli10} to gather together meaningful
and permutation rules, by clearly separating
\emph{multiplicative} from \emph{exponential} rules, in the sense of
Linear Logic~\cite{LL}.

Full laziness has been formalized in different ways.  Pointer
graphs~\cite{wadsworth71,shivers10} are DAGs allowing for an elegant
representation of sharing.  Labeled calculi~\cite{LevyPhD,blanc05} implement
pointer graphs by adding annotations to $\lambda$-terms, which makes
the syntax more difficult to handle.
Lambda-lifting~\cite{hughes84,peytonjones87} implements full laziness
by resorting to translations from $\lambda$-terms to supercombinators.
In contrast to all the previous formalisms, our calculus is defined on
standard $\lambda$-terms with explicit cuts, without the use of any
complementary syntactical tool.  So is Ariola and Felleisen's
call-by-need~\cite{ariola97}, however, their notion of full laziness
relies on external (ad-hoc) meta-level operations used to extract the
skeleton.  Our specification of call-by-need enables fully lazy
sharing, where the skeleton extraction operation is internally encoded
in the term calculus operational semantics.  Last but not least, our
calculus has strong links with proof-theory, notably deep inference.

Balabonski~\cite{balabonski12a,balabonski12b} relates many
formalisms of full laziness and shows that they are equivalent when
considering the number of $\beta$-steps to a normal form. It would
then be interesting to
understand if his unified approach,
(abstractly) stated by means of the theory of
residuals~\cite{LevyPhD,levy80}, applies to our own strategy.

We have also studied the calculus from a semantical point of view, by
means of intersection types. Indeed, the type system can be seen as a
model of our implementations of call-by-name and call-by-need, in the
sense that typability and normalization turn out to be equivalent.

Intersection types go back to~\cite{coppo78} and
have been
used to provide characterizations of qualitative~\cite{BDS13}
as well as quantitative~\cite{Carvalho07} models of the
$\lambda$-calculus, where typability and
normalization coincide. Quantitative models
specified by means of non-idempotent
types~\cite{gardner94,Kfoury00}  were first applied to
the $\lambda$-calculus (see a survey in~\cite{bucciarelli17}) and to several other
formalisms ever since, such as
call-by-value~\cite{ehrhard2012,CarraroG14},
call-by-need~\cite{kesner16,accattoli19},
call-by-push-value~\cite{GM18,BucciarelliKRV20} and classical
logic~\cite{kesner20a}. In the present work, we achieve for the first time a quantitative
characterization of fully lazy normalization, which provides upper
bounds for the length of reduction sequences to normal forms.

Characterizations provided by intersection type systems
sometimes lead to observational
equivalence results (\eg~\cite{kesner16}).
In this work we succeed to
prove  observational equivalence related to  a fully lazy implementation
of weak call-by-need, a result which would be extremely involved
to prove by means of syntactical tools of rewriting,
as done for weak call-by-need in~\cite{ariola97}.
Moreover, our result implies that our node replication implementation
of full laziness is observationally equivalent to standard call-by-name
and to weak call-by-need (see~\cite{kesner16}), as well as to
the more semantical notion of neededness (see~\cite{KesnerRV18}).

A Curry-Howard interpretation of the logical \emph{switch} rule of
deep inference is given in~\cite{sherratt19,sherratt20} as an
end-of-scope operator, thus introducing the \emph{spinal atomic
$\lambda$-calculus}. The calculus implements a refined optimization of
call-by-need, where only the \emph{spine} of the abstraction (tighter
than the skeleton) is duplicated.  It would be interesting to adapt
$\rcalc$ to spine duplication by means of an appropriate end-of-scope operator,
such as the one in~\cite{hendriks03}. Further optimizations might also be
considered.

Finally, this paper only considers weak evaluation strategies, \ie\ with
reductions forbidden
under abstractions, but it would be interesting to extend our notions
to full (strong) evaluations too~\cite{gregoire02,balabonski17}.
Extending full laziness to classical logic would be another
interesting research direction, possibly taking preliminary ideas
from~\cite{he18}. We
would also like to investigate (quantitative) \emph{tight} types for
our fully lazy strategy, as done for weak call-by-need
in~\cite{accattoli19}, which does not seem evident in our node
replication framework.

\section*{Acknowledgment}
\noindent Last author was partially supported by CNPq Universal 430667/2016-7 grant. 

\bibliographystyle{alphaurl}
\bibliography{bib}

\begin{thebibliography}{MOTW99}

\bibitem[ABKL14]{accattoli14a}
Beniamino Accattoli, Eduardo Bonelli, Delia Kesner, and Carlos Lombardi.
\newblock A nonstandard standardization theorem.
\newblock In {\em {POPL}}, pages 659--670. {ACM}, 2014.

\bibitem[ABM14]{accattoli14b}
Beniamino Accattoli, Pablo Barenbaum, and Damiano Mazza.
\newblock Distilling abstract machines.
\newblock In {\em {ICFP}}, pages 363--376. {ACM}, 2014.

\bibitem[Acc18a]{PNESAccattoli18}
Beniamino Accattoli.
\newblock Proof nets and the linear substitution calculus.
\newblock In {\em {ICTAC}}, volume 11187 of {\em Lecture Notes in Computer
  Science}, pages 37--61. Springer, 2018.

\bibitem[Acc18b]{accattoli18}
Beniamino Accattoli.
\newblock Proof nets and the linear substitution calculus.
\newblock In {\em Theoretical Aspects of Computing {\textendash} {ICTAC} 2018},
  pages 37--61. Springer International Publishing, 2018.
\newblock \href {https://doi.org/10.1007/978-3-030-02508-3_3}
  {\path{doi:10.1007/978-3-030-02508-3_3}}.

\bibitem[ACCL90]{abadi90}
Mart{\'{\i}}n Abadi, Luca Cardelli, Pierre{-}Louis Curien, and Jean{-}Jacques
  L{\'{e}}vy.
\newblock Explicit substitutions.
\newblock In {\em {POPL}}, pages 31--46. {ACM} Press, 1990.

\bibitem[AF97]{ariola97}
Zena~M. Ariola and Matthias Felleisen.
\newblock The call-by-need lambda calculus.
\newblock {\em J. Funct. Program.}, 7(3):265--301, 1997.

\bibitem[AGL19]{accattoli19}
Beniamino Accattoli, Giulio Guerrieri, and Maico Leberle.
\newblock Types by need.
\newblock In {\em {ESOP}}, volume 11423 of {\em Lecture Notes in Computer
  Science}, pages 410--439. Springer, 2019.

\bibitem[AK10]{accattoli10}
Beniamino Accattoli and Delia Kesner.
\newblock The structural \emph{lambda}-calculus.
\newblock In {\em {CSL}}, volume 6247 of {\em Lecture Notes in Computer
  Science}, pages 381--395. Springer, 2010.

\bibitem[Bal12a]{balabonski12b}
Thibaut Balabonski.
\newblock {\em La plein paresse, une certain optimalit{\'e} : partage de
  sous-termes et strat{\'e}gies de r{\'e}duction en r{\'e}{\'e}criture d'ordre
  sup{\'e}rieur}.
\newblock PhD thesis, Paris 7, 2012.
\newblock URL: \url{http://www.theses.fr/2012PA077198}.

\bibitem[Bal12b]{balabonski12a}
Thibaut Balabonski.
\newblock A unified approach to fully lazy sharing.
\newblock In {\em {POPL}}, pages 469--480. {ACM}, 2012.

\bibitem[Bal13]{balabonski13}
Thibaut Balabonski.
\newblock Weak optimality, and the meaning of sharing.
\newblock In {\em {ICFP}}, pages 263--274. {ACM}, 2013.

\bibitem[Bar85]{barendregt85}
Hendrik~Pieter Barendregt.
\newblock {\em The lambda calculus - its syntax and semantics}, volume 103 of
  {\em Studies in logic and the foundations of mathematics}.
\newblock North-Holland, 1985.

\bibitem[BBBK17]{balabonski17}
Thibaut Balabonski, Pablo Barenbaum, Eduardo Bonelli, and Delia Kesner.
\newblock Foundations of strong call by need.
\newblock {\em Proc. {ACM} Program. Lang.}, 1({ICFP}):20:1--20:29, 2017.

\bibitem[BCDC83]{BarendregtCoppoDezani83}
Henk Barendregt, Mario Coppo, and Mariangiola Dezani-Ciancaglini.
\newblock A filter lambda model and the completeness of type assignment.
\newblock {\em Bulletin of Symbolic Logic}, 48:931--940, 1983.

\bibitem[BDS13]{BDS13}
Hendrik~Pieter Barendregt, Wil Dekkers, and Richard Statman.
\newblock {\em Lambda Calculus with Types}.
\newblock Perspectives in logic. Cambridge University Press, 2013.

\bibitem[BE01]{bucciarelli01}
Antonio Bucciarelli and Thomas Ehrhard.
\newblock On phase semantics and denotational semantics: the exponentials.
\newblock {\em Ann. Pure Appl. Log.}, 109(3):205--241, 2001.

\bibitem[BKRV20]{BucciarelliKRV20}
Antonio Bucciarelli, Delia Kesner, Alejandro R{\'{\i}}os, and Andr{\'{e}}s
  Viso.
\newblock The bang calculus revisited.
\newblock In {\em {FLOPS}}, volume 12073 of {\em Lecture Notes in Computer
  Science}, pages 13--32. Springer, 2020.

\bibitem[BKV17]{bucciarelli17}
Antonio Bucciarelli, Delia Kesner, and Daniel Ventura.
\newblock Non-idempotent intersection types for the lambda-calculus.
\newblock {\em Log. J. {IGPL}}, 25(4):431--464, 2017.

\bibitem[BLM05]{blanc05}
Tomasz Blanc, Jean{-}Jacques L{\'{e}}vy, and Luc Maranget.
\newblock Sharing in the weak lambda-calculus.
\newblock In {\em Processes, Terms and Cycles}, volume 3838 of {\em Lecture
  Notes in Computer Science}, pages 70--87. Springer, 2005.

\bibitem[BN98]{Nipkow-Baader}
Franz Baader and Tobias Nipkow.
\newblock {\em Term rewriting and all that}.
\newblock Cambridge University Press, 1998.

\bibitem[BR95]{BlooRoselx}
Roel Bloo and Kristoffer Rose.
\newblock Preservation of strong normalization in named lambda calculi with
  explicit substitution and garbage collection.
\newblock In {\em Computing Science in the Netherlands}, pages 62--72.
  Netherlands Computer Science Research Foundation, 1995.

\bibitem[CD78]{coppo78}
Mario Coppo and Mariangiola Dezani{-}Ciancaglini.
\newblock A new type assignment for {$\lambda$}-terms.
\newblock {\em Arch. Math. Log.}, 19(1):139--156, 1978.

\bibitem[CDCV81]{CoppoDezaniVenneri81}
Mario Coppo, Mariangiola Dezani-Ciancaglini, and Betti Venneri.
\newblock Functional characters of solvable terms.
\newblock {\em Zeitschrift für mathematiche Logik und Grundlagen der
  Mathematik}, 27:45--58, 1981.

\bibitem[CF12]{chang14}
Stephen Chang and Matthias Felleisen.
\newblock The call-by-need lambda calculus, revisited.
\newblock In {\em {ESOP}}, volume 7211 of {\em Lecture Notes in Computer
  Science}, pages 128--147. Springer, 2012.

\bibitem[CG14]{CarraroG14}
Alberto Carraro and Giulio Guerrieri.
\newblock A semantical and operational account of call-by-value solvability.
\newblock In {\em FoSSaCS}, volume 8412 of {\em Lecture Notes in Computer
  Science}, pages 103--118. Springer, 2014.

\bibitem[dC07]{Carvalho07}
Daniel de~Carvalho.
\newblock {\em S{\'e}mantiques de la logique lin{\'e}aire et temps de calcul}.
\newblock PhD thesis, Universit{\'e} Aix-Marseille II, 2007.
\newblock URL: \url{https://www.theses.fr/2007AIX22066}.

\bibitem[DCKP03]{DCKP03}
Roberto Di~Cosmo, Delia Kesner, and Emmanuel Polonowski.
\newblock Proof nets and explicit substitutions.
\newblock {\em Math. Struct. Comput. Sci.}, 13(3):409--450, 2003.

\bibitem[Ehr12]{ehrhard2012}
Thomas Ehrhard.
\newblock Collapsing non-idempotent intersection types.
\newblock In {\em {CSL}}, volume~16 of {\em LIPIcs}, pages 259--273. Schloss
  Dagstuhl - Leibniz-Zentrum f{\"{u}}r Informatik, 2012.

\bibitem[Gar94]{gardner94}
Philippa Gardner.
\newblock Discovering needed reductions using type theory.
\newblock In {\em {TACS}}, volume 789 of {\em Lecture Notes in Computer
  Science}, pages 555--574. Springer, 1994.

\bibitem[GGP10]{guglielmi10}
Alessio Guglielmi, Tom Gundersen, and Michel Parigot.
\newblock A proof calculus which reduces syntactic bureaucracy.
\newblock In {\em {RTA}}, volume~6 of {\em LIPIcs}, pages 135--150. Schloss
  Dagstuhl - Leibniz-Zentrum f{\"{u}}r Informatik, 2010.

\bibitem[GHP13a]{gundersen13a}
Tom Gundersen, Willem Heijltjes, and Michel Parigot.
\newblock Atomic lambda calculus: {A} typed lambda-calculus with explicit
  sharing.
\newblock In {\em {LICS}}, pages 311--320. {IEEE} Computer Society, 2013.

\bibitem[GHP13b]{gundersen13b}
Tom Gundersen, Willem Heijltjes, and Michel Parigot.
\newblock A proof of strong normalisation of the typed atomic lambda-calculus.
\newblock In {\em {LPAR}}, volume 8312 of {\em Lecture Notes in Computer
  Science}, pages 340--354. Springer, 2013.

\bibitem[GHP21]{guerrieri21}
Giulio Guerrieri, Willem~B. Heijltjes, and Joseph~W.N. Paulus.
\newblock A deep quantitative type system.
\newblock In Christel Baier and Jean Goubault-Larrecq, editors, {\em 29th EACSL
  Annual Conference on Computer Science Logic (CSL 2021)}, volume 183 of {\em
  Leibniz International Proceedings in Informatics (LIPIcs)}, pages
  24:1--24:24, Dagstuhl, Germany, 2021. Schloss Dagstuhl--Leibniz-Zentrum
  f{\"u}r Informatik.
\newblock URL: \url{https://drops.dagstuhl.de/opus/volltexte/2021/13458}, \href
  {https://doi.org/10.4230/LIPIcs.CSL.2021.24}
  {\path{doi:10.4230/LIPIcs.CSL.2021.24}}.

\bibitem[Gir87]{LL}
Jean{-}Yves Girard.
\newblock Linear logic.
\newblock {\em Theor. Comput. Sci.}, 50:1--102, 1987.

\bibitem[Gir96]{girard96}
Jean-Yves Girard.
\newblock {\em Proof-nets: The parallel syntax for proof-theory}, volume 180 of
  {\em Lecture Notes in Pure Applied Mathematics}, page 97–124.
\newblock Marcel Dekker, 1996.

\bibitem[GL02]{gregoire02}
Benjamin Gr{\'{e}}goire and Xavier Leroy.
\newblock A compiled implementation of strong reduction.
\newblock In {\em {ICFP}}, pages 235--246. {ACM}, 2002.

\bibitem[GM18]{GM18}
Giulio Guerrieri and Giulio Manzonetto.
\newblock The bang calculus and the two girard's translations.
\newblock In {\em Linearity-TLLA@FLoC}, volume 292 of {\em {EPTCS}}, pages
  15--30, 2018.

\bibitem[He18]{he18}
Fanny He.
\newblock {\em The Atomic Lambda-Mu Calculus}.
\newblock PhD thesis, University of Bath, January 2018.

\bibitem[Hug83]{hughes84}
John Hughes.
\newblock The design and implementation of programming languages.
\newblock Technical Report {PRG}-40, {Oucl}, July 1983.

\bibitem[HvO03]{hendriks03}
Dimitri Hendriks and Vincent van Oostrom.
\newblock \reflectbox{$\lambda$}.
\newblock In {\em {CADE}}, volume 2741 of {\em Lecture Notes in Computer
  Science}, pages 136--150. Springer, 2003.

\bibitem[Jon87]{peytonjones87}
Simon L.~Peyton Jones.
\newblock {\em The Implementation of Functional Programming Languages}.
\newblock Prentice-Hall, 1987.

\bibitem[Kes07]{kesner07a}
Delia Kesner.
\newblock The theory of calculi with explicit substitutions revisited.
\newblock In {\em {CSL}}, volume 4646 of {\em Lecture Notes in Computer
  Science}, pages 238--252. Springer, 2007.

\bibitem[Kes09]{kesner09}
Delia Kesner.
\newblock A theory of explicit substitutions with safe and full composition.
\newblock {\em Log. Methods Comput. Sci.}, 5(3), 2009.

\bibitem[Kes16]{kesner16}
Delia Kesner.
\newblock Reasoning about call-by-need by means of types.
\newblock In {\em FoSSaCS}, volume 9634 of {\em Lecture Notes in Computer
  Science}, pages 424--441. Springer, 2016.

\bibitem[Kfo00]{Kfoury00}
A.~J. Kfoury.
\newblock A linearization of the lambda-calculus and consequences.
\newblock {\em J. Log. Comput.}, 10(3):411--436, 2000.

\bibitem[KL07]{kesner07}
Delia Kesner and St{\'{e}}phane Lengrand.
\newblock Resource operators for lambda-calculus.
\newblock {\em Inf. Comput.}, 205(4):419--473, 2007.

\bibitem[KR11]{kesner11}
Delia Kesner and Fabien Renaud.
\newblock A prismoid framework for languages with resources.
\newblock {\em Theor. Comput. Sci.}, 412(37):4867--4892, 2011.

\bibitem[KRV18]{KesnerRV18}
Delia Kesner, Alejandro R{\'{\i}}os, and Andr{\'{e}}s Viso.
\newblock Call-by-need, neededness and all that.
\newblock In Christel Baier and Ugo~Dal Lago, editors, {\em Foundations of
  Software Science and Computation Structures - 21st International Conference,
  {FOSSACS} 2018, Held as Part of the European Joint Conferences on Theory and
  Practice of Software, {ETAPS} 2018, Thessaloniki, Greece, April 14-20, 2018,
  Proceedings}, volume 10803 of {\em Lecture Notes in Computer Science}, pages
  241--257. Springer, 2018.
\newblock \href {https://doi.org/10.1007/978-3-319-89366-2\_13}
  {\path{doi:10.1007/978-3-319-89366-2\_13}}.

\bibitem[KV14]{kesner14}
Delia Kesner and Daniel Ventura.
\newblock Quantitative types for the linear substitution calculus.
\newblock In {\em {IFIP} {TCS}}, volume 8705 of {\em Lecture Notes in Computer
  Science}, pages 296--310. Springer, 2014.

\bibitem[KV20]{kesner20a}
Delia Kesner and Pierre Vial.
\newblock Non-idempotent types for classical calculi in natural deduction
  style.
\newblock {\em Logical Methods in Computer Science ; Volume 16},
  abs/1802.05494, 2020.
\newblock \href {https://arxiv.org/abs/1802.05494} {\path{arXiv:1802.05494}},
  \href {https://doi.org/10.23638/LMCS-16(1:3)2020}
  {\path{doi:10.23638/LMCS-16(1:3)2020}}.

\bibitem[Lam90]{lamping90}
John Lamping.
\newblock An algorithm for optimal lambda calculus reduction.
\newblock In {\em {POPL}}, pages 16--30. {ACM} Press, 1990.

\bibitem[L{\'e}v78]{LevyPhD}
Jean-Jacques L{\'e}vy.
\newblock {\em R{\'e}ductions correctes et optimales dans le lambda-calcul}.
\newblock PhD thesis, Universit{\'e} Paris VII, 1978.

\bibitem[L{\'e}v80]{levy80}
Jean-Jacques L{\'e}vy.
\newblock {\em Optimal Reductions in the Lambda-Calculus}, pages 159--191.
\newblock Academic Press Inc, 1980.

\bibitem[LM99]{levy99}
Jean-Jacques Lévy and Luc Maranget.
\newblock Explicit substitutions and programming languages.
\newblock In {\em Lecture Notes in Computer Science}, pages 181--200. Springer
  Berlin Heidelberg, 1999.
\newblock \href {https://doi.org/10.1007/3-540-46691-6_14}
  {\path{doi:10.1007/3-540-46691-6_14}}.

\bibitem[MOTW99]{maraist99}
John Maraist, Martin Odersky, David~N. Turner, and Philip Wadler.
\newblock Call-by-name, call-by-value, call-by-need and the linear lambda
  calculus.
\newblock {\em Theor. Comput. Sci.}, 228(1-2):175--210, 1999.

\bibitem[MOW98]{maraist98}
John Maraist, Martin Odersky, and Philip Wadler.
\newblock The call-by-need lambda calculus.
\newblock {\em J. Funct. Program.}, 8(3):275--317, 1998.

\bibitem[Plo75]{plotkin75}
Gordon~D. Plotkin.
\newblock Call-by-name, call-by-value and the lambda-calculus.
\newblock {\em Theor. Comput. Sci.}, 1(2):125--159, 1975.

\bibitem[She19]{sherratt19}
David~R. Sherratt.
\newblock {\em A lambda-calculus that achieves full laziness with spine
  duplication}.
\newblock PhD thesis, University of Bath, March 2019.

\bibitem[SHGP20]{sherratt20}
David Sherratt, Willem Heijltjes, Tom Gundersen, and Michel Parigot.
\newblock Spinal atomic lambda-calculus.
\newblock In {\em FoSSaCS}, volume 12077 of {\em Lecture Notes in Computer
  Science}, pages 582--601. Springer, 2020.

\bibitem[SU06]{soerensen06}
M.~H. Sørensen and Paweł Urzyczyn.
\newblock {\em Lectures on the Curry-Howard Isomorphism}.
\newblock Elsevier, 2006.
\newblock \href {https://doi.org/10.1016/s0049-237x(06)x8001-1}
  {\path{doi:10.1016/s0049-237x(06)x8001-1}}.

\bibitem[SW10]{shivers10}
Olin Shivers and Mitchell Wand.
\newblock Bottom-up beta-reduction: Uplinks and lambda-dags.
\newblock {\em Fundam. Informaticae}, 103(1-4):247--287, 2010.

\bibitem[Wad71]{wadsworth71}
Christopher~P. Wadsworth.
\newblock {\em Semantics and Pragmatics of the Lambda Calculus}.
\newblock PhD thesis, Oxford University, 1971.

\end{thebibliography}

\appendix
\section{Proofs}

\levelsubstitution*
\label{p:level-substitution}
\begin{proof}
  If $x \notin \fv t$, then $t \msub x p = t$ and the property holds in both
  cases since $\lv x t = 0$.
  Let $x \in \fv t$. We show the two cases.
  \begin{enumerate}
    \item $z \notin \fv p$. We reason by induction on $t$.
      \begin{itemize}
        \item If $t = x$, then
          $\lv z {x \msub x p} = \lv z p = 0 = \lv z x$.
        \item If $t = \lam y {t'}$, then
          $\lv z {(\lam y {t'}) \msub x p}
          = \lv z {\lam y {t' \msub x p}}
          = \lv z {t' \msub x p}
          =_{\ih} \lv z {t'}
          = \lv z {\lam y {t'}}$.
        \item If $t = t_1 t_2$, then
          $\lv z {(t_1 t_2) \msub x p}
          = \lv z {t_1 \msub x p t_2 \msub x p}$.
          Then, 
          $\lv z {(t_1 t_2) \msub x p}
          = \max (\lv z {t_1 \msub x p}, \lv z {t_2 \msub x p})$
          $=_{\ih} \max (\lv z {t_1}, \lv z {t_2})
          = \lv z {t_1 t_2}$.
        \item If $t = t' \cut y u$, then $t \msub x p = t' \msub x p \cut y {u \msub x p}$.
          There are two cases.
          \begin{itemize}
            \item If $z \notin \fv {u \msub x p}$, that is $z \notin \fv u$, then
              $\lv z {t' \msub x p \cut y {u \msub x p}}
              = \lv z {t' \msub x p}
              =_{\ih} \lv z {t'}
              = \lv z {t' \cut y u}$.
            \item Otherwise,
              $\lv z {t' \msub x p \cut y {u \msub x p}}
              = \max (\lv z {t' \msub x p}, \lv y {t' \msub x p} + \lv z {u
              \msub x p} + \isES{\cut y u})
              =_{\ih} \max (\lv z {t'}, \lv y {t'} + \lv z u + \isES{\cut y u})
              = \lv z {t' \cut y u}$.
          \end{itemize}
      \end{itemize}

    \item $z \in \fv p$. By induction on $t$.
      \begin{itemize}
        \item If $t = x$, then
          $t\msub {x}{p} = p$ and
          \[ \lv z {t\msub {x}{p}}
          = \lv z p
          = 0
          = \max (0, 0)
        = \max (\lv z x, \lv x x). \]

        \item If $t = \ly.t'$, then
          $t\msub {x}{p} = \ly.t'\msub {x}{p}$ and
          \[ \lv z {\ly.t'\msub {x}{p}} = \lv z {t'\msub {x}{p}}
          =_{\ih}  \max (\lv z {t'}, \lv x {t'})
        = \max (\lv z {\ly.t'}, \lv x {\ly.t'}). \]

        \item If $t = t_1t_2$ then $t\msub {x}{p} =  t_1\msub {x}{p}t_2\msub
          {x}{p}$. For $i \in
          \{1,2\}$, one has either $\lv z {t_i\msub
          {x}{p}} = \max(\lv z {t_i}, \lv x {t_i})$ by \ih\ if $x
          \in \fv {t_i}$, or $\lv z {t_i \msub x p} = \lv z {t_i}$ by
          Point \ref{l:level-substitution-void} otherwise.
          \begin{itemize}
            \item If $x \in \fv {t_1} \cap \fv {t_2}$ then,
              \begin{align*}
                \lv z {t\msub {x}{p}} & = \max(\lv z {t_1\msub {x}{p}},\lv z {t_2\msub {x}{p}}) \\
                                      & =_{\ih} \max(\max(\lv z {t_1}, \lv x
                                      {t_1}), \max(\lv z {t_2}, \lv x {t_2}) ) \\
                                      & = \max(\lv z {t_1}, \lv x
                                      {t_1}, \lv z {t_2}, \lv x {t_2}) \\
                                      & = \max(\max(\lv z {t_1}, \lv z {t_2}), \max(\lv x {t_1}, \lv x {t_2})) \\
                                      & = \max(\lv z {t_1 t_2}, \lv x {t_1 t_2})
              \end{align*}
            \item If $x \notin \fv {t_2}$, then
              \begin{align*}
                \lv z {t\msub {x}{p}}
                & = \max(\lv z {t_1\msub {x}{p}},\lv z {t_2\msub {x}{p}}) \\
                & =_{\ih + P.\ref{l:level-substitution-void}} \max(\max(\lv z {t_1}, \lv x {t_1}), \lv z {t_2}) \\
                & = \max(\lv z {t_1}, \lv x {t_1}, \lv z {t_2}) \\
                & = \max(\max(\lv z {t_1}, \lv z {t_2}), \max(\lv x {t_1}, 0)) \\
                & = \max(\lv z {t_1 t_2}, \lv x {t_1 t_2})
              \end{align*}
            \item If $x \notin \fv{t_1}$ the case is as above.
          \end{itemize}

        \item If $t = t_1 \cut y {t_2}$ then
          $t \msub x p = t_1 \msub x p \cut y {t_2 \msub x p}$.
          By $\alpha$-conversion we can assume
          $y \notin \fv{p}$. By \ih\, one has $\lv z {t_i\msub {x}{p}}
          = \max(\lv z {t_i}, \lv x {t_i})$ for $i \in \{1,2\}$,
          if $x \in \fv {t_i}$, $\lv z {t_i \msub x p} = \lv z {t_i}$ otherwise.
          There are two cases.

          \begin{enumerate}
            \item If $z \notin \fv{t_2\msub {x}{p}}$ then $z \notin
              \fv{t_2}$ and
              necessarily $x \notin \fv{t_2}$ since $z \in \fv{p}$.
              Therefore,
              \begin{align*}
                \lv z {t\msub {x}{p}}
                & = \lv z {t_1\msub {x}{p}} \\
                & =_{\ih} \max(\lv z {t_1}, \lv x {t_1}) \\
                & = \max (\lv z {t_1 \cut y {t_2}}, \lv x {t_1 \cut y {t_2}})
              \end{align*}

            \item  If $z \in \fv{t_2\msub {x}{p}}$ then,
              \begin{enumerate}
                \item If $x \notin \fv {t_1}$, then $x \in \fv {t_2}$.
                  \begin{align*}
                    \lv z {t \msub x p} &=
                    \max (\lv z { t_1\msub{x}{p}}, \lv y {t_1\msub{x}{p}} + \lv
                    z {t_2 \msub x p} + \isES{\cut y {t_2}}) \\
                                        &=_{P.\ref{l:level-substitution-void}}
                                        \max (\lv z {t_1}, \lv y {t_1} + \lv z
                                        {t_2 \msub x p} + \isES{\cut y {t_2}}) \\
                                        &=_{\ih} \max (\lv z {t_1}, \lv y {t_1}
                                        + \max (\lv z {t_2}, \lv x {t_2}) +
                                        \isES{\cut y {t_2}}) \\
                                        &= \max (\lv z {t_1}, \lv y {t_1} + \lv
                                        z {t_2} + \isES{\cut y {t_2}}, \lv y {t_1} + \lv x
                                        {t_2} + \isES{\cut y {t_2}}) \\
                                        &= \begin{cases}
                                          \max (\max (\lv z {t_1}, \lv y {t_1} +
                                          \lv z {t_2} + \isES{\cut y {t_2}}),\\
                                          \phantom{(\max(\max} \lv y {t_1} + \lv x {t_2} + \isES{\cut
                                          y {t_2}} ) & z \in \fv {t_2} \\
                                          \max (\lv z {t_1}, \lv y {t_1} + \lv x
                                          {t_2} + \isES{\cut y {t_2}} ) & z \notin \fv {t_2}
                                        \end{cases} \\
                                        &= \max (\lv z {t_1 \cut y {t_2}}, \lv x {t_1 \cut y {t_2}} )
                  \end{align*}
                \item If $x \notin \fv {t_2}$, then $x \in \fv {t_1}$ and $z \in \fv {t_2}$:
                  \begin{align*}
                    \lv z {t \msub x p}
              &= \max (\lv z {t_1 \msub x p}, \lv y {t_1 \msub x p} + \lv z {t_2
              \msub x p} + \isES{\cut y {t_2}}) \\
              &=_{\ih + P.\ref{l:level-substitution-void}} \max (\lv z {t_1}, \lv x
              {t_1} , \lv y {t_1}  + \lv z {t_2} + \isES{\cut y {t_2}}) \\
              &= \max (\max (\lv z {t_1}, \lv y {t_1} + \lv z {t_2} + \isES{\cut y {t_2}}), \lv x {t_1 \cut y {t_2}} ) \\
              &= \max (\lv z {t_1 \cut y {t_2}}, \lv x {t_1 \cut y {t_2}} )
                  \end{align*}

                \item If $x \in \fv {t_1} \cap \fv {t_2}$:
                  \begin{align*}
                    \lv z {t\msub {x}{p}}
              & = \max(\lv z {t_1\msub {x}{p}}, \lv y {t_1\msub {x}{p}} + \lv z {t_2\msub
              {x}{p}} + \isES{\cut y {t_2}}) \\
              &=_{\ih} \max(\max(\lv z {t_1}, \lv x {t_1} ), \lv y {t_1} +
              \max (\lv z {t_2}, \lv x {t_2} ) + \isES{\cut y {t_2}})  \\
              &= \max(\lv z {t_1}, \lv x {t_1} , \lv y {t_1} +
              \lv z {t_2}  + \isES{\cut y {t_2}},\\
              &\phantom{=\max(}\lv y {t_1} +
              \lv x {t_2} + \isES{\cut y {t_2}}) \\
              &= \begin{cases}
                \max (\lv z {t_1}, \lv y {t_1} + \lv z {t_2} + \isES{\cut y {t_2}},\\
                \phantom{\max(}\lv x {t_1}, \lv y {t_1} + \lv x {t_2} + \isES{\cut y {t_2}} ) & z \in \fv {t_2} \\
                \max (\lv z {t_1}, \lv x {t_1} , \lv y {t_1} + \lv x {t_2} +
                \isES{\cut y {t_2}}) & z \notin \fv {t_2}
              \end{cases} \\
              &= \max (\lv z {t_1 \cut y {t_2}}, \lv x {t_1 \cut y {t_2}} )
                  \qedhere
                  \end{align*}
              \end{enumerate}
          \end{enumerate}
          \end{itemize}
      \end{enumerate}
\end{proof}

\stabilitylevels*
\label{p:stability-levels}
\begin{proof} \hfill
  \begin{enumerate}
    \item Let $t_0 =  \ctx{\fc}{o}$ and $t_1 = \ctx{\fc}{o'}$, where  $o \rew{\permlr} o'$ is a root step.
      We reason by induction on $\fc$.
      First we consider the base cases, where $\fc = \ec$.
      \begin{itemize}
        \item $t_0 = \lam y {t \cut x u} \rew\permlr
          (\lam y t) \cut x u = t_1$, where $y \notin \fv{u}$.
          We have two cases:
          \begin{enumerate}
            \item If $w \notin \fv u$.
              \[ \lv w {\lam y {t \cut x u}} =
                \lv w {t \cut x u} =
                \lv w {t} =
                \lv w {\lam y t} =
              \lv w {(\lam y t) \cut x u} \]
            \item If $w \in \fv u$.
              \begin{align*}
                \lv w {\lam y {t \cut x u}}
                &= \lv w {t \cut x u}\\
                &= \max(\lv w t, \lv x t + \lv w u + \isES{\cut x u})\\
                &= \max(\lv w {\lam y t}, \lv x {\lam y t} + \lv w u +\isES{\cut x
                u})\\
                &= \lv w {(\lam y t) \cut x u}
              \end{align*}
        \end{enumerate}

      \item $t_0 = t \cut x u s \rew\permlr (ts) \cut x u = t_1$, where $x \notin \fv{s}$.
        We have two cases:
        \begin{enumerate}
          \item If $w \notin \fv u$.
            \begin{align*}
              \lv w {{t \cut x u} s}
              &= \max(\lv w {{t \cut x u}}, \lv w {s} )\\
              &= \max(\lv w {t}, \lv w {s} )\\
              &= \lv w {ts}\\
              &= \lv w {{(ts) \cut x u}}
            \end{align*}
          \item If $w \in \fv u$.
            \begin{align*}
              &\lv w {{t \cut x u} s}\\
              &= \max(\lv w {{t \cut x u}}, \lv w {s} )\\
              &= \max(\lv w {t}, \lv x t + \lv w u +\isES{\cut x u}, \lv w {s} )\\
              &= \max(\lv w {t}, \lv w {s}, \lv x {t} + \lv w u +\isES{\cut x u})\\
              &= \max(\lv w {t}, \lv w {s}, \max(\lv x {t}, 0) + \lv w u +\isES{\cut x u}) & (x \notin \fv{s})\\
              &= \max(\lv w {t}, \lv w {s}, \max(\lv x {t}, \lv x s) + \lv w u +\isES{\cut x u})\\
              &= \max(\lv w {ts}, \lv x {ts} + \lv w u +\isES{\cut x u})\\
              &= \lv w {{(ts) \cut x u}}
            \end{align*}
        \end{enumerate}

      \item $t_0 = t s \cut x u \rew\permlr (ts) \cut x u = t_1$, where $x \notin \fv{t}$.
        We have two cases:
        \begin{enumerate}
          \item If $w \notin \fv u$.
            \begin{align*}
              \lv w {ts \cut x u}
              &= \max(\lv w t, \lv w {s \cut x u})\\
              &= \max(\lv w t, \lv w s)\\
              &= \lv w {ts}\\
              &= \lv w {{(ts) \cut x u}}
            \end{align*}
          \item If $w \in \fv u$.
            \begin{align*}
              \lv w {t s \cut x u}
              &= \max(\lv w t, \lv w {s \cut x u})\\
              &= \max(\lv w t, \lv w s, \lv x s + \lv w u + \isES{\cut x u})\\
              &= \max(\lv w t, \lv w s, \max(0, \lv x s) + \lv w u + \isES{\cut x u})\\
              &= \max(\lv w t, \lv w s, \max(\lv x t, \lv x s) + \lv w u +
              \isES{\cut x u})\\
              &= \max(\lv w {ts}, \lv x {ts} + \lv w u + \isES{\cut x u})\\
              &= \lv w {(ts) \cut x u}
            \end{align*}
        \end{enumerate}

      \item Let $t_0 = t \cut  y {s \cut x
        u} \rew\permlr t \cut y s \cut x u = t_1$, where $x \notin \fv{t}$.
        We have four cases:
        \begin{enumerate}
          \item If $w \notin \fv s \cup \fv u$:
            \[ \lv w {t \cut y {s \cut x u}}
              = \lv w t
              = \lv w {t \cut y s}
            = \lv w {t \cut y s \cut x u} \]
          \item If $w \in \fv s, w \notin \fv u$:
            \begin{align*}
              \lv w {t \cut y {s \cut x u}}
              &= \max (\lv w t, \lv y t + \lv w {s \cut x u} + \isES{\cut y s})\\
              &= \max (\lv w t, \lv y t + \lv w s + \isES{\cut y s})\\
              &= \lv w {t \cut  y s}\\
              &= \lv w {t \cut  y s \cut x u}
            \end{align*}
          \item If $w \notin \fv s, w \in \fv u$.
            \begin{align*}
              &\lv w {t \cut  y {s \cut x u}}\\
              &= \max (\lv w t, \lv y t + \lv w {s \cut x u} + \isES{\cut y s})\\
              &= \max (\lv w t, \lv y t + \lv x s + \lv w u + \isES{\cut x u} +
              \isES{\cut y s})\\
              &\geq \max (\lv w {t \cut  y s}, \lv x {t \cut  y s} +
              \lv w u + \isES{\cut x u})\\
              &= \lv w {t \cut y s \cut x u}
            \end{align*}
          \item If $w \in \fv s \cap \fv u$:
            \[ \begin{array}{lll}
              &\lv w {t \cut y {s \cut x u}}\\
              &= \max (\lv w t, \lv y t + \lv w {s \cut x u} + \isES{\cut y
              s})\\
              &= \max (\lv w t, \lv y t + \max (\lv w s, \lv x s + \lv w u +
              \isES{\cut x u}) + \isES{\cut y s})\\
              &= \max (\lv w t, \lv y t + \lv w s + \isES{\cut y s},\\
              &\qquad \lv y t + \lv x s + \lv w u + \isES{\cut x u} + \isES{\cut
              y s})\\
              &= \max (\max (\lv w t, \lv y t + \lv w s + \isES{\cut y s}),\\
              &\qquad\max (\lv x t, \lv y t + \lv x s + \isES{\cut y s}) + \lv w
              u + \isES{\cut x u})\\
              &\geq \max (\lv w {t \cut y s}, \lv x {t \cut  y s} + \lv w u +
              \isES{\cut x u})\\
              &= \lv w {t \cut  y s \cut x u}
          \end{array} \]
      \end{enumerate}
  \end{itemize}

  The inductive cases are the following:
  \begin{itemize}
    \item If $\fc  = \lam x {\fc'}$, where $x \neq w$, then
      $\lv w {\ctx{\lam x {\fc'}}{o}} =
      \lv w {\ctx{{\fc'}}{o}} \geq_{\ih}
      \lv w {\ctx{{\fc'}}{o'}} =
      \lv w {\ctx{\fc}{o'}}$. 
    \item If $\fc  = \fc' u$, then
      $\lv w {\ctx{\fc'}{o} u} =
      \max(\lv w {\ctx{\fc'}{o}}, \lv w u) \geq_{\ih}
      \max(\lv w {\ctx{\fc'}{o'}}, \lv w u) =
      \lv w {\ctx{\fc}{o'}}$.
    \item If $\fc = u \fc'$, then
      $\lv w {u \ctx{\fc'}{o}} =
      \max(\lv w u, \lv w {\ctx{\fc'}{o}}) \geq_{\ih}
      \max(\lv w u, \lv w {\ctx{\fc'}{o'}}) =
      \lv w {\ctx{\fc}{o'}}$.
    \item If $\fc  = \fc'  \cut x u$, then
      \begin{enumerate}
        \item If $w \notin \fv u$, then
          $\lv w {\ctx{\fc'}{o}  \cut x u}=
          \lv w {\ctx{\fc'}{o}} \geq_{\ih}
          \lv w {\ctx{\fc'}{o'}} =
          \lv w {\ctx{\fc}{o'}}$. 
        \item If $w \in \fv u$, then
          $\lv w {\ctx{\fc'}{o}\cut x u } =
          \max (\lv w {\ctx{\fc'}{o}}, \lv x {\ctx{\fc'}{o}} + \lv w u + \isES{\cut x u})
          \geq_{\ih} \max (\lv w {\ctx{\fc'}{o'}}, \lv x {\ctx{\fc'}{o'}} + \lv w u + \isES{\cut x u})
          = \lv w {\ctx{\fc'}{o'} \cut x u}
          = \lv w {\ctx{\fc}{o'}}$.
      \end{enumerate}
    \item If $\fc  = u \cut x {\fc'}$, then
      \begin{enumerate}
        \item If $w \notin \fv {\ctx{\fc'}{o}}$, then
          $\lv w {u \cut x {\ctx{\fc'}{o}}} =
          \lv w u =
          \lv w {u \cut x {\ctx{\fc'}{o'}}} =
          \lv w {\ctx{\fc}{o'}}$.
        \item If $w \in \fv {\ctx{\fc'}{o}}$, then
          $\lv w {u \cut x {\ctx{\fc'}{o}}} =
          \max (\lv w u, \lv x u + \lv w {\ctx{\fc'}{o}} + \isES{\cut x
          {\ctx{\fc'} o}}) \geq_{\ih}
          \max (\lv w u, \lv x u + \lv w {\ctx{\fc'}{o'}} + \isES{\cut x
          {\ctx{\fc'} o}}) =
          \lv w {u \cut x {\ctx{\fc'}{o'}}} =
          \lv w {\ctx{\fc}{o'}}$.
      \end{enumerate}
  \end{itemize}

\item We reason by induction on the reduction
  relation, \ie\ by  induction on the context $\fc$ where
  the root reduction takes place. 
  We detail the base case which is $\fc=\ec$.
  In all such cases we use Point~\ref{l:stability-levels-permutation}
  to push $\lc$ outside, \ie\ we can write $t_0
  \rew{\nrsub} t_1$ as
  $t_0 \rew\permlr \ctx{\lc}{t'_0} \rew{\nrsub'} \ctx{\lc}{t'_1} = t_1$,
  where $t'_0 \rew{\nrsub'}  t'_1$ does not push any list context outside.
  We then show the property for steps $t'_0 \rew{\nrsub'}  t'_1$
  not pushing any substitution outside
  and we conclude by
  $\lv w {t_0} \geq_{P.\ref{l:stability-levels-permutation}} \lv w {\ctx{\lc}{t'_0}} \geq \lv w {\ctx{\lc}{t'_1}} = \lv w {t_1}$.
  The inductive cases for $\fc$ are treated as in Point~\ref{l:stability-levels-permutation}.
  \begin{itemize}
    \item $t'_0 = {t \es {x} {us}} \rew{\apprr} {t \msub {x} {yz} \es y u \es z s}=t'_1 $, where $y$ and $z$ are fresh variables.

      \begin{enumerate}
        \item If $w \notin \fv{us}$ (\ie\ $w \notin \fv{u}$ and $w \notin \fv{s}$):
          \begin{align*}
            \lv w {t \es {x} {us}}
            &= \lv w t\\
            & =_{L.\ref{l:level-substitution}:\ref{l:level-substitution-void}}
            \lv w {t \msub {x} {yz}}\\
            &= \lv w {t \msub {x} {yz} \es y u}\\
            &= \lv w {t \msub {x} {yz} \es y u \es z s}
          \end{align*}
        \item   If $w \in \fv{u}$ and $ w \notin \fv{s}$. There are two  cases.
          \begin{enumerate}
            \item   If $x \notin \fv{t}$:
              \begin{align*}
                \lv w {t \es {x} {us}}
                &= \max( \lv w t, \lv x t + \lv w {us} +1)\\
                &= \max( \lv w t, \lv x t + \lv w u +1)\\
                &= \max( \lv w t, 0 + \lv w u +1)\\
                &= \max( \lv w t, \lv y t  + \lv w u + 1)\\
                & =_{L.\ref{l:level-substitution}:\ref{l:level-substitution-void}}
                \max( \lv w t, \lv y {t \msub {x} {yz} } + \lv w u + 1)\\
                & =_{L.\ref{l:level-substitution}:\ref{l:level-substitution-void}}
                \max( \lv w {t \msub {x} {yz} }, \lv y {t \msub {x} {yz} } + \lv
                w u + 1)\\
                &= \lv w {t \msub {x} {yz} \es y u}\\
                &= \lv w {t \msub {x} {yz} \es y u \es z s}
              \end{align*}

            \item   If $x \in \fv{t}$:
              \begin{align*}
                \lv w {t \es {x} {us}}
                &= \max( \lv w t, \lv x t + \lv w {us} +1)\\
                &= \max( \lv w t, \lv x t   + \lv w u + 1)\\
                &= \max( \lv w t, \lv x t + 0  + \lv w u + 1)\\
                &= \max( \lv w t, \max (0, \lv x t + 0)  + \lv w u + 1)\\
                & = \max( \lv w t, \max (\lv y t, \lv x t + \lv y {yz})  + \lv w u + 1)\\
                &=_{L.\ref{l:level-substitution}:\ref{l:level-substitution-non-void}}
                \max( \lv w t, \lv y {t \msub {x} {yz} } + \lv w u + 1)\\
                &=_{L.\ref{l:level-substitution}:\ref{l:level-substitution-void}}
                \max( \lv w {t \msub {x} {yz} }, \lv y {t \msub {x} {yz} } + \lv
                w u + 1)\\
                &= \lv w {t \msub {x} {yz} \es y u}\\
                &= \lv w {t \msub {x} {yz} \es y u \es z s}
              \end{align*}
          \end{enumerate}

        \item If $w \notin \fv{u}$ and $w \in \fv{s}$. There are two cases.
          \begin{enumerate}
            \item  If $x \notin \fv{t}$:
              \begin{align*}
                \lv w {t \es {x} {us}}
                &= \max( \lv w t, \lv x t + \lv w {us} +1)\\
                &= \max( \lv w t, \lv x t + \lv w s +1)\\
                &= \max( \lv w t, 0 + \lv w s +1)\\
                &= \max( \lv w t, \lv z t  + \lv w s + 1)\\
                &=_{L.\ref{l:level-substitution}:\ref{l:level-substitution-void}}
                \max( \lv w t, \lv z {t \msub {x} {yz} } + \lv w s + 1)\\
                &=_{L.\ref{l:level-substitution}:\ref{l:level-substitution-void}}
                \max( \lv w {t \msub {x} {yz} }, \lv z {t \msub {x} {yz} } + \lv
                w s + 1)\\
                &= \max( \lv w {t \msub x {yz} \es y u}, \lv z {t \msub x {yz}
                \es y u} + \lv w s + 1)\\
                &= \lv w {t \msub x {yz} \es y u \es z s}
              \end{align*}

            \item   If $x \in \fv{t}$:
              \begin{align*}
                \lv w {t \es {x} {us}}
                &= \max( \lv w t, \lv x t + \lv w {us} +1)\\
                &= \max( \lv w t, \lv x t   + \lv w s + 1)\\
                &= \max( \lv w t, \lv x t + 0    + \lv w s + 1)\\
                &= \max( \lv w t, \max (0, \lv x t + 0)  + \lv w s + 1)\\
                &= \max( \lv w t, \max (\lv z t, \lv x t + \lv z {yz}  )  + \lv
                w s + 1)\\
                &=_{L.\ref{l:level-substitution}:\ref{l:level-substitution-non-void}}
                \max( \lv w t, \lv z {t \msub {x} {yz} } + \lv w s + 1)\\
                &=_{L.\ref{l:level-substitution}:\ref{l:level-substitution-void}}
                \max( \lv w {t \msub {x} {yz} }, \lv z {t \msub {x} {yz} } + \lv
                w s + 1)\\
                &= \max( \lv w {t \msub x {yz} \es y u}, \lv z {t \msub x {yz}
                \es y u} + \lv w s + 1)\\
                &= \lv w {t \msub {x} {yz} \es y u \es z s}
              \end{align*}
          \end{enumerate}

        \item If $w \in \fv{u}$ and $w \in \fv{s}$. There are two  cases.
          \begin{enumerate}
            \item  If $x \notin \fv{t}$:
              \begin{align*}
                &\lv w {t \es {x} {us}}\\
                &= \max( \lv w t, \lv x t + \lv w {us} +1)\\
                &= \max( \lv w t, \lv x t + \max( \lv w u, \lv w s) +1)\\
                &= \max( \lv w t, \max( \lv w u, \lv w s) +1))\\
                &= \max (\lv w {t}, \lv w u +1, \lv w s +1 )\\
                &= \max (\lv w {t}, \lv y {t} + \lv w u +1, \lv z {t } + \lv w s +1)\\
                &=_{L.\ref{l:level-substitution}:\ref{l:level-substitution-void}}
                \max (\lv w {t \msub {x} {yz}}, \lv y {t \msub {x} {yz} } + \lv
                w u +1, \lv z {t \msub {x} {yz}} + \lv w s +1 )\\
                &= \max (\lv w {t \msub {x} {yz}}, \lv y {t \msub {x} {yz} } +
                \lv w u +1, \lv z {t \msub {x} {yz} \es y u} + \lv w s +1 )\\
                &= \max (\lv w {t \msub {x} {yz} \es y u}, \lv z {t \msub {x}
                {yz} \es y u} + \lv w s +1 )\\
                &= \lv w {t \msub {x} {yz} \es y u \es z s}
              \end{align*}

            \item   If $x \in \fv{t}$:
              \begin{align*}
                &\lv w {t \es {x} {us}}\\
                &= \max( \lv w t, \lv x t + \lv w {us} +1)\\
                &= \max( \lv w t, \lv x t + \lv w u +1, \lv x t + \lv w {s} +1)\\
                &= \max( \lv w t, \lv x t + \lv w u +1, \max(0, \lv x t + 0) +
                \lv w {s} +1)\\
                &= \max( \lv w t, \lv x t + \lv w u +1, \max(\lv z t, \lv x
                t + \lv z {yz}) + \lv w {s} +1)\\
                &=_{L.\ref{l:level-substitution}:\ref{l:level-substitution-non-void}}
                \max(\lv w {t}, \lv x t   + \lv w u +1, \lv z {t \msub {x} {yz}}
                + \lv w s + 1)\\
                &= \max(\lv w {t}, \max(0, \lv x t + 0)  + \lv w u +1, \lv z {t
                \msub {x} {yz}} + \lv w s + 1)\\
                & = \max(\lv w {t}, \max( \lv y t, \lv x t + \lv y {yz})  + \lv
                w u +1, \lv z {t \msub {x} {yz}} + \lv w s + 1)\\
                &=_{L.\ref{l:level-substitution}:\ref{l:level-substitution-non-void}}
                \max(\lv w {t}, \lv y {t \msub {x} {yz}} + \lv w u +1, \lv z {t
                \msub {x} {yz}} + \lv w s + 1)\\
                &= \max(\lv w {t \msub {x} {yz}}, \lv y {t \msub {x} {yz}} + \lv
                w u +1, \lv z {t \msub {x} {yz}} + \lv w s + 1)\\
                &= \max(\lv w {t \msub {x} {yz} \es y u}, \lv z {t \msub {x}
                {yz}} + \lv w s + 1)\\
                &= \max(\lv w {t \msub {x} {yz} \es y u}, \lv z {t \msub {x} {yz} \es y u} + \lv w s + 1)\\
                &= \lv w {t \msub {x} {yz} \es y u \es z s}
              \end{align*}
          \end{enumerate}
      \end{enumerate}

    \item $t'_0 =  {t \es {x} {\lam y u}} \rew{\distrr} {t \dist {x} {\lam y z \es {z} u}}=t'_1 $.
      There are two cases.
      \begin{enumerate}
        \item $w \notin \fv{\lam y u}$:
          \[ \lv w {t \es x {\lam y u}}
            = \lv w t
          = \lv w {t \dist {x} {\lam y z \es {z} u}} \]
        \item $w \in \fv{\lam y u}$ (\ie\ $w \in \fv{u}$ and $w \neq y$)
          \begin{align*}
            \lv w {t \es {x} {\lam y u}}
            &= \max( \lv w {t}, \lv x t + \lv w {\lam y u} +1 )\\
            &= \max( \lv w {t}, \lv x t + \lv w {u} +1 )\\
            &= \max( \lv w {t}, \lv x t + \max(0, 0 + \lv w u +1))\\
            &= \max( \lv w {t}, \lv x t + \max(\lv w z, \lv z z + \lv w u +1))\\
            &= \max( \lv w {t}, \lv x t + \lv w {z \es {z} u}  )\\
            &= \max( \lv w {t}, \lv x t + \lv w {\lam y z \es {z} u}  )\\
            &= \lv w {t \dist {x} {\lam y z \es {z} u}}
          \end{align*}
      \end{enumerate}

    \item $t'_0 = {t \dist {x} {\lam y u}} \rew{\absrr} t \msub
      {x} {\lam y u} =t'_1 $,
      where $u$ is pure.  There are two cases:
      \begin{enumerate}
        \item If  $w \notin \fv{\lam y u}$ or $x \notin \fv t$:
          \[
            \lv w {t \dist {x} {\lam y u}} =  \lv w t =_{L.\ref{l:level-substitution}.\ref{l:level-substitution-void}}
            \lv w { {t \msub {x} {\lam y {u'}}}} =
            \lv w {\ctx \lc {t \msub {x} {\lam y {u'}}}}
          \]
        \item If  $w \in \fv{\lam y u}$ and $x \in \fv t$: 
          \[
            \lv w {t \dist {x} {\lam y u}} =
            \max(\lv w t, \lv x t) =_{L.\ref{l:level-substitution}.\ref{l:level-substitution-non-void}}
            \lv w {t \msub {x} {\lam y u}}
          \]
      \end{enumerate}

    \item  $t'_0 =  {t \es {x} {y}} \rew{\varrr} {t \msub {x} y}=t'_1 $.
      \begin{enumerate}
        \item If  $w \neq y$:
          \[ \lv w  {t \es x y}
            = \lv w t
            =_{L.\ref{l:level-substitution}.\ref{l:level-substitution-void}}
          \lv w  {t \msub x y} \]
        \item If $w =y$ and $x  \notin \fv{t}$:
          \begin{align*}
            \lv w  {t \es {x} y}
            &= \max( \lv w  {t}, \lv x t + \lv w y + 1)\\
            &= \max( \lv w  {t}, 1)\\
            &\geq \lv w t\\
            &=_{L.\ref{l:level-substitution}.\ref{l:level-substitution-void}}
            \lv w {t \msub {x} y}
          \end{align*}
        \item If $w = y$ and $x \in \fv t$:
          \begin{align*}
            \lv w {t \es x y}
            &= \max (\lv w t, \lv x t + \lv w y + 1)\\
            &\geq \max (\lv w t, \lv x t)\\
            &=_{L.\ref{l:level-substitution}.\ref{l:level-substitution-non-void}}
            \lv w {t \msub x y}
            \qedhere
          \end{align*}
      \end{enumerate}
  \end{itemize}
\end{enumerate}
\end{proof}

\mssubstitution*
\label{p:ms-substitution}
\begin{proof}
  By induction on $t$.
  In this proof, $\level(o)$ denotes the first element of an object
  $o \in \subobj$: $\level(\subes k n) = k$ and $\level(\subdist k) = k$.
  If  $\subes k n \in \subobj$ we also define $\snd(\subes k
    n) = n$.
  \begin{itemize}
    \item If $t = y$, then $\submeas{y} = \submeas{y \msub x p } =  \emult$ so
      the property is straightforward for any $n \in \mathbb{N}$.
    \item If $t = \lam y u$, then
      $\submeas{t \msub x p} = \submeas{u \msub x p}$
      and $\lv x t = \lv x u$.
      The  property trivially  holds by the \ih\ 
    \item If $t = u_1 u_2$, then
      we have $\submeas{t \msub x p} =
      \submeas {u_1 \msub x p} \sqcup \submeas {u_2 \msub x p}$
      and $\lv x {u_1 u_2} = \max(\lv x {u_1} , \lv x {u_2} )$.
      Let $o \in\submeas{t \msub x p}$ thus
      $o\in\submeas{u_1 \msub x p} \sqcup \submeas{u_2 \msub x p}$. 
      Suppose w.l.o.g. that $o\in\submeas{u_1 \msub x p}$.
      Let $K_1 = \lv x {u_1} \leq K$. By the \ih\ we have either
      (1) $o \in \submeas {u_1}_{\subdistname}$, 
      (2) $o \in \submeas {u_1}^{ > K_1}_{\subesname}$, or
      (3) $o = \subes k n$ where $k \leq K_1$ and $n \leq N_1$
      for some $N_1 \in \mathbb{N}$. 
      If (1) holds, then $o \in \submeas t_{\subdistname}$ and we are done. Otherwise,
      $o=  \subes  {k} n$, and  we consider two cases.
      \begin{itemize}
        \item If $k  > K$, then (2)  implies
          $o \in \submeas {u_1}^{ > K}_{\subesname}$ which implies 
          $o \in \submeas {t}^{ > K}_{\subesname}$ while
          (3)   implies $k \leq K$ which leads to a contradiction. 
        \item If $k \leq K$,  we are done.
        \end{itemize}
        Considering that $o\in\submeas{u_2 \msub x p}$ we have
        a similar result for some $N_2 \in \mathbb{N}$. We thus have the
        result for $N = \max(N_1,N_2)$.
    \item If $t = u_1 \es y {u_2}$, then
      we can assume by $\alpha$-conversion that $y \notin \fv{p}$.
      Therefore,
      \begin{align*}
        \submeas t
        &= \submeas {u_1} \sqcup (\lv y {u_1} + 1) \cdot
        \submeas {u_2} \sqcup [\subes{\lv y {u_1} + 1}{\tmsz {u_2}}] \text{ and}\\
        \submeas{t \msub x p}
        &= \submeas{u_1 \msub x p}
        \sqcup (\lv y {u_1 \msub x p} + 1) \cdot \submeas{u_2 \msub x
        p}\\
        &\sqcup \mult{\subes {\lv y {u_1 \msub x p} + 1}{\tmsz{u_2 \msub x p}}}   \\
        &=_{L.\ref{l:level-substitution}:\ref{l:level-substitution-void}} \submeas{u_1 \msub x p}
        \sqcup (\lv y {u_1} + 1) \cdot \submeas{u_2 \msub x p}
        \sqcup \mult{\subes {\lv y {u_1} + 1}{\tmsz{u_2 \msub x p}}}
      \end{align*}
      There are two cases:
      \begin{enumerate}
        \item If $x \notin \fv{u_2}$, then
          $\lv x t = \lv x {u_1}$.
          Moreover, $\submeas{u_2 \msub x p} = \submeas{u_2}$ and  $\tmsz{u_2 \msub
          x p} = \tmsz{u_2}$.  Let  $o \in \submeas{t \msub x p}$.
          \begin{itemize}
            \item If $o \in \submeas{u_1 \msub x p}$, then let $K_1 = \lv x
              {u_1} = \lv x t = K$, so that the \ih\ gives either (1) $
              o \in \submeas {u_1}_{\subdistname}$,  (2) $o \in \submeas{u_1}^{>
              K_1}_{\subesname}$, or (3) $o = \subes {k} n$ where $k \leq
              K_1$ and $n \leq N_1$ for some $N_1 \in \mathbb{N}$. If (1) holds, then $o \in \submeas t_{\subdistname}$ and we are
              done. If (2) holds, then $o \in \submeas{u_1}^{> K}_{\subesname}$ since
              $K_1 = K$, which implies $o \in \submeas{t}^{> K}_{\subesname}$ and we
              are done. Otherwise, (3) holds and $k \leq K_1 =K$ as
              required.
            \item If $o \in (\lv y {u_1} + 1) \cdot \submeas{u_2 \msub x
              p} = (\lv y {u_1} + 1) \cdot \submeas{u_2}$, then $o \in
              \submeas t = \submeas t_{\subdistname} \sqcup  \submeas
              t^{>K}_{\subesname} \sqcup
              \submeas t^{\leq K}_{\subesname}$, which particularly implies in the
              last case that $o = \subes k n$ and $k \leq
              K$. Note that, since $\submeas t^{\leq
                  K}_{\subesname}$ is finite, we can take $N_2 = \max\{\snd(o)
                \;|\; o \in \submeas t^{\leq
                  K}_{\subesname}\}$.
            \item If $o = \subes {\lv y {u_1} + 1}{\tmsz{u_2 \msub x
              p}} = \subes {\lv y {u_1} + 1}{\tmsz{u_2}}$, then $o
              \in \submeas{t}$ and thus either $o \in \submeas t^{> K}$ or $o
              \in \submeas t^{\leq K}$, which particularly implies in the
              last case that $\level(o) \leq K$ and
                $\snd(o)\leq N_2$, with $N_2$ defined as above.
            \end{itemize}
            The result then holds for $N = \max(N_1,N_2)$.
        \item If $x \in \fv{u_2}$, then $\lv x t = \max(\lv x {u_1},
          \lv y {u_1} + \lv x {u_2} +1)$. Let  $o \in \submeas{t \msub x p}$.
          \begin{itemize}
            \item If $o \in \submeas{u_1 \msub x p}$, then let $K_1 = \lv x
              {u_1} = \lv x t \leq K$, so that the \ih\ gives either (1)
              $o \in \submeas {u_1}_{\subdistname}$, (2) $o \in \submeas{u_1}^{>
              K_1}_{\subesname}$, or (3) $o = \subes k n$ where $k
              \leq K_1$ and $n \leq N_1$ for some $N_1 \in \mathbb{N}$.  If (1) holds, then $o \in \submeas t_{\subdistname}$ and we
              are done.  Otherwise $o = \subes k n$ and we consider
              two cases.
              \begin{itemize}
                \item If $k  > K$, then (2)  implies
                  $o \in \submeas{u_1}^{> K}_{\subesname}$, and thus $o \in
                  \submeas{t}^{> K}_{\subesname}$,
                  while (3)  implies $k  \leq K$
                  which leads to a contradiction. 
                \item If $k \leq  K$, then we are done. 
              \end{itemize}
            \item If $o \in (\lv y {u_1} + 1) \cdot \submeas{u_2 \msub x
              p}$, then there is $o' \in \submeas{u_2 \msub x p}$ such that
              $\level(o) = \level(o') + (\lv y {u_1} + 1)$. Let $K_2 =
              \lv x {u_2}$, so that the \ih\ gives either (1) $o' \in
              \submeas {u_2}_{\subdistname}$,  (2) $o' \in \submeas{u_2}^{>
              K_2}_{\subesname}$,
              or (3) $o' = \subes k n$ where $k \leq K_2$ and
                $n \leq N_2$ for some $N_2 \in \mathbb{N}$.  If (1)
              holds, then $o \in (\lv y {u_1} + 1) \cdot \submeas
              {u_2}_{\subdistname}$, thus $o \in \submeas t_{\subdistname}$ and we are done.
              If (2) holds, then $o \in (\lv y {u_1} + 1) \cdot \submeas
              {u_2}^{> K_2}_{\subesname}$ and thus $\level(o) > K_2 + (\lv y
              {u_1} + 1)$.  We consider two cases.
              \begin{itemize}
                \item If $\level(o)  > K \geq K_2 + \lv y {u_1} + 1$,
                  then (2) implies $o \in \submeas{t}^{> K}_{\subesname}$ while (3) leads
                  to a contradiction.
                \item If $\level(o) \leq  K$, then we are done. 
              \end{itemize}
            \item If $o = \subes {\lv y {u_1} + 1}{\tmsz{u_2 \msub x p}}$, then
              $\level(o) = \lv y {u_1} + 1 \leq K$ and $\snd(o) = \tmsz{u_2 \msub x p}$.
          \end{itemize}
          The result then holds for $N = \max(N_1,N_2,\tmsz{u_2 \msub x p})$.
      \end{enumerate}
    \item If $t = u_1 \dist y {u_2}$, the analysis is similar.
      \qedhere
  \end{itemize}
\end{proof}

\mspi*
\label{p:ms-stable-by-pi}
\begin{proof}
  Let $t = \ctx \fc {t_0} \rew\permlr \ctx \fc {t_1} = t'$,
  where $t_0 \rew\permlr t_1$ is a reduction step at the root
  position. We proceed by induction on $\fc$.
  We detail the base case where the context $\fc$ is $\ec$,
  by inspecting the cases where the explicit cuts are explicit substitutions,
  as the remaining cases for explicit distributors are
  similar.
  \begin{enumerate}
    \item If $t= t_0 = \lam y {t \es x u} \rew\permlr (\lam y t) \es x u = t_1  =t'$,
      where $y \notin \fv u$:
      \begin{align*}
        \submeas{t}
        &= \submeas{t \es x u} \\
        &=  \submeas t \sqcup (\lv x t + 1) \cdot \submeas u
        \sqcup \mult{\subes{\lv x t + 1}{\tmsz u}}  \\
        &= \submeas{\lam{y}{t}} \sqcup (\lv x {\lam{y}{t}} + 1) \cdot \submeas u
        \sqcup \mult{\subes {\lv x {\lam{y}{t}} + 1}{\tmsz u}} \\
        &= \submeas{t'}
      \end{align*}
    \item If $t=t_0 = t \es x u s \rew\permlr (ts) \es x u = t_1=t'$,
      where $x \notin \fv s$:
      \begin{align*}
        \submeas{t}
        &= \submeas{t \es x u} \sqcup \submeas s \\
        &= \submeas t \sqcup (\lv x t + 1) \cdot \submeas u
        \sqcup \submeas s \\
        &= \submeas{ts} \sqcup (\lv x {ts} + 1) \cdot \submeas u
        \sqcup \mult{\subes{\lv x {ts} + 1}{\tmsz u}} \\
        &= \submeas{t'}
      \end{align*}
    \item If $t=t_0 = t s \es x u \rew\permlr (ts) \es x u = t_1=t'$,
      where $x \notin \fv t$:
      \begin{align*}
        \submeas{t}
        &= \submeas t \sqcup \submeas{s \es x u} \\
        &= \submeas t \sqcup
        \submeas s \sqcup (\lv x s + 1) \cdot \submeas u
        \sqcup \mult{\subes{\lv x s + 1}{\tmsz u}} \\
        &= \submeas{ts} \sqcup (\lv x {ts} + 1) \cdot \submeas u
        \sqcup \mult{\subes{\lv x {ts} + 1}{\tmsz u}} \\
        &= \submeas{t'}
      \end{align*}
    \item If $t= t_0 = t \es y {s \es x u} \rew\permlr t \es y s \es x u = t_1=t'$,
      where $x \notin \fv t$:
      \begin{align*}
        \submeas{t}
        &= \submeas t \sqcup (\lv y t + 1) \cdot \submeas{s \es x u}
        \sqcup \mult{\subes{\lv y t + 1}{\tmsz {s \es x u}}} \\
        &= \submeas t \sqcup (\lv y t + 1) \cdot
        \left(\submeas s \sqcup (\lv x s + 1) \cdot \submeas u
        \sqcup \mult{\subes{\lv x s + 1}{\tmsz u}} \right) \\
        &\quad \sqcup \mult{\subes{\lv y t + 1}{\tmsz {s \es x u}}} \\
        &= \submeas t \sqcup (\lv y t + 1) \cdot \submeas s 
        \sqcup (\lv y t + \lv x s + 2) \cdot \submeas u \\
        &\quad \sqcup \mult{
          \subes{\lv y t + \lv x s + 2}{\tmsz u},
        \subes{\lv y t + 1}{\tmsz{s \es x u}}} \\
        &= \left( \submeas t \sqcup (\lv y t + 1) \cdot \submeas s
        \sqcup \mult{\subes{\lv y t + 1}{\tmsz{s \es x u}}} \right) \\
        &\quad \sqcup (\lv y t + \lv x s + 2) \cdot \submeas u
        \sqcup \mult{\subes{\lv y t + \lv x s + 2}{\tmsz u}} \\
        & \subgemult{} \left( \submeas t \sqcup (\lv y t + 1) \cdot \submeas s
        \sqcup \mult{\subes{\lv y t + 1}{\tmsz s}} \right) \\
        &\quad \sqcup (\lv x {t \es y s} + 1) \cdot \submeas u
        \sqcup \mult{\subes{\lv x {t \es y s} + 1}{\tmsz u}} \\
        &= \submeas{t'}
        \end{align*}
        The $\subgemult{}$ inequation is justified by the following facts:
        \begin{itemize}
          \item $\tmsz {s\es x u} > \tmsz s$.
          \item $\lv y t + \lv x s +2  =  \max(0,\lv y t + \lv x s +1)+1 = \lv x {t \es y s} +1$.
        \end{itemize}
    \end{enumerate}
    The inductive cases are straightforward.
  \end{proof}


\diamond*
\label{p:diamond}
\begin{proof}
  We split the statement above in three different properties, each one proved by induction
  on the involved relation relations.
  \begin{enumerate}
    \item If $t \rew\whdblr u$ and $t \rew\whdblr s$, then  there exists $t'$ such that $u \rew\whdblr t'$ and $s \rew\whdblr t'$.
      We consider the following cases:
      \begin{itemize}
        \item  (\whlrdbapp, \whlrdbapp)
          We then have $t = t_0 t_1$ such that $t \rew\whdblr u_0 t_1 = u$ and
          $t \rew\whdblr s_0 t_1 = s$, where $t_0 \rew \whdblr u_0$ and $t_0 \rew \whdblr s_0$.
          By the \ih\  there is $t'_0$ such that $s_0 \rew\whdblr t'_0$ and $u_0 \rew\whdblr t'_0$.
          Therefore  $s \rew\whdblr t_0't_1=t'$ and $u \rew\whdblr t'$.

        \item (\whlrdbsub, \whlrdbsub)
          We then have $t = t_0 \cut x {t_1}$ such that
          $t \rew\whdblr u_0 \cut x {t_1} = u$ and
          $t \rew\whdblr s_0 \cut x {t_1} = s$, where $t_0 \rew\whdblr u_0$ and
          $t_0 \rew\whdblr s_0$.
          By the \ih\ there is $t'_0$ such that $s_0 \rew\whdblr t'_0$ and $u_0 \rew\whdblr t_0'$.
          Therefore $s \rew\whdblr t_0' \cut x {t_1} = t'$
          and $u \rew\whdblr t'$.

        \item (\whlrdbroot, \whlrdbroot), (\whlrdbroot, \whlrdbapp), (\whlrdbroot, \whlrdbsub) and (\whlrdbapp, \whlrdbsub) are impossible cases.
      \end{itemize}

    \item If $t \rew\whslr u$ and $t \rew\whslr s$, then there exists $t'$
      such that $u \rew\whslr t'$ and $s \rew\whslr t'$.
      We consider the following cases:
      \begin{itemize}
        \item (\whlrsubroot, \whlrsubroot) Impossible since
          $u$ and $s$ are assumed to be different.
        \item (\whlrsubroot, \whlrsubsub) 
          We have $t \in \purelterms$
          then $t = t_0 \dist x {\lam y {\ctx{\llc} p \cut z
          {t_1}}}$, where $y \notin \fv{\llc}
          \cup \fv {t_1}$ and such that
          $t \rrule\nrsub \ctx{\llc}{t_0 \msub x {\lam y p}} \cut z {t_1} = u$.
          There are three cases for $t$.
          \begin{itemize}
          \item If $\cut z {t_1} = \es z {\ctx\lc w}$ or $\cut z {t_1} = \es z {\ctx\lc {p_1 p_2}}$,
            then the only possibility in each case
            is $\whlrsubroot$ on term $\ctx\llc p \es z {t_1}$. We
            then have $t \rew\whslr t_0 \dist x {\lam y
                {\ctx{\lc'}{\ctx{\llc} p \msub z {q}}}} = s$
              for some $\lc'$ and some pure term $q$.
              So that $u \rew\whslr \ctx{\lc'}{\ctx\llc{t_0 \msub x {\lam y p}} \msub z {q}} = u'$
              and $s \rew\whslr \ctx{\lc'}{\ctx\llc{t_0 \msub x {\lam y {p \msub z {q}}}}} = s'$. The equality
              $u' = s'$ holds because we can assume $z \neq y$ by $\alpha$-equivalence,
              and $z \notin \fv \llc$ by definition.
          \item If $\cut z {t_1} = \es z {\lam w {t'_1}}$, then the only possible case
            is $\whlrsubroot$ on term $\ctx\llc p \es z {\lam w {t'_1}}$. We then have
             $t = t_0 \dist x {\lam y {\ctx\llc p \es z {\lam w {t'_1}}}}
              \rew\whslr t_0 \dist x {\lam y {\ctx\llc p \dist z {\lam w {w' \es {w'} {t'_1}}}}} = s$. And we close the diagram with 
              $u \rew\whslr \ctx\llc{t_0 \msub x {\lam y p}} \dist z {\lam w {w' \es {w'} {t'_1}}} = t'$
              and $s \rew\whslr t'$.
            \item If $\cut z {t_1} = \dist z {\lam w {t_1'}}$, then we have two
              different cases:

          \begin{itemize}
          \item  If the reduction happens inside $t_1'$, then $t = t_0 \dist x {\lam y {\ctx\llc p \dist z {\lam w {t_1'}}}}
              \rew\whslr t_0 \dist x {\lam y {\ctx\llc p \dist z {\lam w {s_1}}}} = s$,
              where $t_1' \rew\whslr s_1$.
              Then we close by 
              $u \rew\whslr \ctx\llc{t_0 \msub x {\lam y p}} \dist z {\lam w {s_1}} = t'$
              and $s \rew\whslr t'$.
          \item Otherwise, the $\whlrsubroot$ case for $\ctx\llc p \dist z {\lam w {t'_1}}$ gives
              $t \rew\whslr t_0 \dist x {\lam y {\ctx{\lc}{\ctx{\llc} p \msub z {v}}}} = s$ for some $\lc$ and some value $v$.
              So that $u \rew\whslr \ctx\lc{\ctx\llc{t_0 \msub x {\lam y p}} \msub z {v}} = u'$
              and $s \rew\whslr \ctx\lc{\ctx\llc{t_0 \msub x {\lam y {p \msub z {v}}}}} = s'$. The equality
              $u' = s'$ holds because we can assume 
              $y \notin \fv{v} \cup \{z\}$ by $\alpha$-equivalence,
              and $z \notin \fv \llc$ by definition.
          \end{itemize}
          
          \end{itemize}
        \item  (\whlrsubapp, \whlrsubapp)
          We then have $t = t_0 t_1$ such that $t \rew\whslr u_0 t_1 = u$ and
          $t \rew\whslr s_0 t_1 = s$, where $t_0 \rew \whslr u_0$ and $t_0 \rew \whslr s_0$.
          By the \ih\  $s_0 \rew\whslr t'_0$ and $u_0 \rew\whslr t'_0$.
          Therefore  $u \rew\whslr t_0't_1 = t'$ and $s \rew\whslr t'$.
        \item (\whlrsubsub, \whlrsubsub)
          We have $t = t_0 \dist x {\lam y {t_1}}$ such that
          $t \rew\whslr t_0 \dist x {\lam y {u_1}}=u$ and
          $t \rew\whslr t_0 \dist x {\lam y {s_1}}=s$,
          where $t_1 \rew\whslr u_1$ and
          $t_1 \rew\whslr s_1$.
          By the \ih\ $s_1 \rew\whslr t_1'$ and $u_1 \rew\whslr t_1'$. Therefore
          $u \rew\whslr t_0 \dist x {\lam y {t_1'}} = t'$
          and $s \rew\whslr t'$.
        \item (\whlrsubroot, \whlrsubapp) and (\whlrsubsub, \whlrsubapp) are impossible cases.
      \end{itemize}

    \item If $t \rew\whdblr u$ and $t \rew\whslr s$, then there exists $t'$ such that $u \rew\whslr t'$ and $s \rew\whdblr t'$.
      We consider the following cases:
      \begin{itemize}
        \item (\whlrdbroot, \whlrsubapp) 
          We have $t = \ctx\lc{\lam x {t_0}} \cut y {t_2} t_1$ such that
          $t \rew\whdblr \ctx\lc{t_0 \es x {t_1}} \cut y {t_2} = u$.
          There are three cases for $t \rew\whslr s$.
          \begin{itemize}
            \item If $t = \ctx\lc{\lam x {t_0}} \dist y {\lam z {t'_2}} t_1
              \rew\whslr \ctx\lc{\lam x {t_0}} \dist y {\lam z {t'_3}} t_1 = s$, where $t_2=\lam z {t'_2}$ and $t'_2 \rew\whslr t'_3$, 
              then $u \rew\whslr \ctx\lc{t_0 \es x {t_1}} \dist y {\lam z {t'_3}} = t'$ and $s \rew\whdblr t'$.
            \item If $t = \ctx\lc{\lam x {t_0}} \es y {\lam z {t'_2}} t_1
              \rew\whslr \ctx\lc{\lam x {t_0}} \dist y {\lam z {w \es w {t'_2}}} t_1 = s$,
              where $t_2=\lam z {t'_2}$,
              then $u \rew\whslr \ctx\lc{t_0 \es x {t_1}} \dist y {\lam z {w \es w {t'_2}}} = t'$ and $s \rew\whdblr t'$.
            \item Otherwise, we have
              $t \rew\whslr \ctx{\lc'}{\ctx\lc{\lam x {t_0}} \msub y p} t_1 = s$,
              for some $\lc'$ and some pure term $p$.
              Therefore, $u \rew\whslr \ctx{\lc'}{\ctx\lc{t_0 \es x {t_1}} \msub y p} = t'$
              and $s \rew\whdblr t'$ because $y \notin \fv {t_1}$.
              Note that  $y$ may be free in $\lc$.
          \end{itemize}
        \item (\whlrdbapp, \whlrsubapp)
          We have $t = t_0t_1$ such that $t \rew\whdblr u_0t_1 = u$ and $t \rew\whslr s_0t_1 = s$,
          where $t_0 \rew\whdblr u_0$ and $t_0 \rew\whslr s_0$.
          By \ih\ there exists $t_0'$ such that  $s_0 \rew\whdblr t_0'$
          and $u_0 \rew\whslr t_0'$.
          Therefore, $u \rew\whslr t_0't_1 = t'$ and $s \rew\whdblr t'$.
        \item (\whlrdbsub, \whlrsubroot)
          We have $t = t_0 \cut x {t_1}$ such that $t \rew\whdblr u_0 \cut x {t_1} = u$, where $t_0 \rew\whdblr u_0$.
          If $t = t_0 \es x {\ctx\lc{\lam y {t_2}}} \rew\whslr
          \ctx\lc {t_0 \dist x {\lam y {z \es z {t_2}}}} = s$, where $t_1 = \lam y {t_2}$, then
          $s \rew\whdblr \ctx\lc {u_0 \dist x {\lam y {z \es z {t_2}}}} = t'$
          and $u \rew\whslr t'$.
          Otherwise, $t \rew\whslr \ctx\lc{t_0 \msub x p} = s$ for some $\lc$ and some pure term $p$.
          We show that $t_0 \msub x p \rew\whdblr u_0 \msub x p$
          by induction on $t_0 \rew\whdblr u_0$.
          From this, we can deduce $s \rew\whdblr \ctx\lc{u_0 \msub x p} = t'$
          and conclude because $u \rew\whslr t'$.
          \begin{enumerate}
            \item If $t_0 = \ctx{\lc'}{\lam y q} t_2 \rew\db \ctx{\lc'}{q \es y {t_2}}=u_0$,
              then w.l.o.g. we can assume by $\alpha$-conversion that
                  $y \notin \fv{p} \cup \{x\}$, then $t_0 \msub x p = \ctx{\lc'\msub x p}{\lam y q\msub x p} t_2\msub x p \rew\whdblr \ctx{\lc'\msub x p}{q \msub x p \es y {t_2 \msub x p}} = u_0 \msub x p$.
                \item If $t_0 = t_0' t_2 \rew\whdblr u_0' t_2 = u_0$
                  from $t_0' \rew\whdblr u_0' $,
                  then by the \ih\ and $\whlrdbapp$ rule we can conclude $t_0 \msub x p = t_0' \msub x p t_2 \msub x p \rew\whdblr u_0' \msub x p t_2 \msub x p = u_0 \msub x p$.
                \item If $t_0 = t_0' \cut y {t_2} \rew\whdblr u_0' \cut y
                  {t_2} = u_0$ from $t_0' \rew\whdblr u_0' $,
                  then w.l.o.g. we can assume by $\alpha$-conversion that
                  $x \neq y$, then by \ih
                  and the $\whlrdbsub$ rule we conclude  $t_0 \msub x p = t_0' \msub x p \cut y {t_2 \msub x p} \rew\whdblr u_0' \msub x p \cut y {t_2 \msub x p} = u_0 \msub x p$.
              \end{enumerate}
            \item (\whlrdbsub, \whlrsubsub)
              We have $t = t_0 \dist x {\lam y {t_1}}$ such that
              $t \rew\whdblr u_0 \dist x {\lam y {t_1}} = u$ and
              $t \rew\whslr t_0 \dist x {\lam y {s_1}} = s$,
              where $t_0 \rew\whdblr u_0$ and
              $t_1 \rew\whslr s_1$. 
              Therefore $u \rew\whslr u_0 \dist x {\lam y {s_1}} = t'$
              and $s \rew\whdblr t'$.
            \item  (\whlrdbroot, \whlrsubroot), (\whlrdbroot, \whlrsubsub), (\whlrdbapp, \whlrsubroot), (\whlrdbapp, \whlrsubsub) and (\whlrdbsub, \whlrsubapp) are impossible cases.
              \qedhere
          \end{itemize}
      \end{enumerate}
\end{proof}

\ndvprop*
\label{p:ndv_prop}
\begin{proof}
  \begin{enumerate}
    \item $x \in \ndv t$. By induction on $t$.
      \begin{itemize}
        \item $t = x$. Then we take $\ndc = \ec$.
        \item $t = t'u$. By the \ih\ there exists $\ndc'$ such that $t' = \ctxnc{\ndc'} x$.
          We then take $\ndc = \ndc' u$.
        \item $t = t' \es y u$.
         By $\alpha$-conversion we can assume $x \neq y$. Either $x \in \ndv{t'}$ or ($x \in \ndv u$ and $y \in \ndv {t'}$).
          In the first case, there exists by the \ih\ on $t'$ a context $\ndc'$ such that $t' = \ctxnc{\ndc'} x$.
          We then take $\ndc = \ndc' \es y u$.
          In the second case, there exists by the \ih\ on $t'$ a context $\ndc_1$ such that $t' = \ctxnc{\ndc_1} y$.
          By the \ih\ on $u$ we have $u = \ctxnc{\ndc_2} x$.
          We then take $\ndc = \ctxnc{\ndc_1} y \es y {\ndc_2}$.
        \item $t = t' \dist x u$.
          By the \ih\ there exists $\ndc'$ such that $t' = \ctxnc{\ndc'} x$.
          We then take $\ndc = \ndc' \dist x u$.
      \end{itemize}
    \item $t = \ctxnc\ndc x$. By induction on $\ndc$.
      \begin{itemize}
        \item $\ndc = \ec$. Then $t = x$ and $\ndv t = \{x\}$.
        \item $\ndc = \ndc' u$. Then 
          $t = t'u$ and by the \ih\ $x \in \ndv{t'}$, so $x \in \ndv t$ by definition.
        \item $\ndc = \ndc' \cut x u$. Then 
          $t = t' \cut x u$ and by the \ih\ $x \in \ndv{t'}$, so $x \in \ndv t$ by definition.
        \item $\ndc = \ctxnc{\ndc_1} y \es y {\ndc_2}$.
          Then  $t = t' \es y u$, where $y \in \fv{t'}$.
          By the \ih\ $x \in \ndv u$, so $x \in \ndv t$.
          \qedhere
      \end{itemize}
  \end{enumerate}
\end{proof}

\neednf*
\label{p:need_nf}
\begin{proof}
  We first show that $t \in \neutg$ iff $t$ is in $\ndlr$-nf and $t$ is not an answer.
  \begin{description}
    \item[$\Rightarrow$] We reason by induction on $t \in \neutg$.
      \begin{itemize}
        \item $t = x$. This case is straightforward.
        \item $t = t'u$ where $t' \in \neutg$.
          By the \ih $t'$ is
          in $\ndlr$-nf and is not an answer,
          so it is not  possible to  apply any $\db$-reduction at the root. Then
          $t$ is in $\ndlr$-nf, and since it is an application it is not
          an answer.
        \item $t = t' \cut x u$ where $t' \in \neutg$ and $x \notin
          \ndv{t'}$.  By the \ih $t'$ is in
          $\ndlr$-nf and it is not an answer.  Moreover, we
          cannot apply rules $\rew\ndlrdist$ nor $\rew\ndlrabs$ because by
          \autoref{l:ndv_prop} there is no context $\ndc$
          surrounding $x$. Then $t$ is in $\ndlr$-nf and is not an answer.
        \item $t = t' \es x u$ where $t', u \in \neutg$ and $x \in \ndv{t'}$.
          By the \ih $t'$ and $u$ are in $\ndlr$-nf and are  not answers.
          We cannot  apply rule $\rew\ndlrdist$ because
          $u$ is not an answer. Then $t$ is in $\ndlr$-nf and is not an answer.
      \end{itemize}
    \item[$\Leftarrow$] We reason by induction on $t$.
      \begin{itemize}
        \item $t = x$ is immediate.
        \item $t = t'u$. Then $t'$ is in $\ndlr$-nf and is not an answer (otherwise $\db$ would be applicable).
          By the \ih $t' \in \neutg$ and thus $t \in \neutg$.
        \item $t = t' \es x u$. Then $t'$ is in $\ndlr$-nf and is not an answer.
          By the \ih $t' \in \neutg$. There are two cases.
          If $x \notin \ndv{t'}$, then $t \in \neutg$ by
          definition and we are done.
          Otherwise $x \in \ndv{t'}$, and we get $t' = \ctxnc\ndc x$ by
          \autoref{l:ndv_prop}. Thus 
            $u$ cannot be an answer because $\rew\ndlrdist$ would apply.
            Moreveor, $u$ is in $\ndlr$-nf because
            otherwise $t$ would not be in $\ndlr$-nf.
            Thus,  $u \in \neutg$ by the \ih and
            we get $t \in \neutg$ by definition.
        \item $t = t' \dist x u$.
          We have $x \notin \ndv{t'}$, because $\rew\ndlrabs$ does not apply. By
          the \ih $t' \in \neutg$, so that $t \in \neutg$.
      \end{itemize}
  \end{description}
  Neutral terms are also normal.
  Answers are normal because the calculus is weak and they belong to the grammar $\normg$.
\end{proof}

\vskip 0cm plus 1cm 
\moreweight*
\label{p:more_weight}
\begin{proof}
  By induction on $\Phi$.
  \begin{itemize}
    \item $\Phi = \inferrule{ }{\seq {x: \mult \sigma}{x}{ \sigma}}$,
      then, $\dermeasaux \Phi m = (0,0,1) = (0,0,1) + (0, 0 + (m-n) * 0,0)$.
    \item If $\Phi = \inferrule{
        \dem {\Phi_t} {\Gamma; y:\MM}{t}{\tau}}{
        \seq{\Gamma}{\lam y t}{ \MM \ft \tau}}$, then 
      \begin{align*}
        \dermeasaux \Phi m
        &= \dermeasaux {\Phi_t} m + (1,m,0) \\
        &=_{\ih} \dermeasaux {\Phi_t} n + (0,(m-n)*\dersz{\Phi_t},0) + (1,n,0) + (0,m-n,0) \\
        &= \dermeasaux \Phi n + (0,(m-n)*\dersz{\Phi_t},0) + (0,m-n,0) \\
        &= \dermeasaux \Phi n + (0,(m-n)*(\dersz{\Phi_t}+1),0)\\
        &= \dermeasaux \Phi n + (0,(m-n)*\dersz{\Phi},0)
      \end{align*}
    \item If $\Phi = \inferrule{ }{\seq {}{\lam x t}{ \ans}}$,
      then, $\dermeasaux \Phi m = (1,m,0) = \dermeasaux \Phi n + (0,(m-n)*\dersz\Phi,0)$.
    \item If $\Phi = \inferrule{
        \dem {\Phi_t} {\Gamma}{t}{\MM \ft \tau} \and
        \dem {\Phi_u} {\Delta}{u}{\MM}}{
        \seq{\Gamma \inter \Delta}{t u}{\tau}}$, then 
      \begin{align*}
        \dermeasaux \Phi m
        &= \dermeasaux {\Phi_t} m + \dermeasaux {\Phi_u} m + (1,m,0) \\
        &=_{\ih} \dermeasaux {\Phi_t} n + (0,(m-n)*\dersz{\Phi_t},0)\\
        &\quad+ \dermeasaux {\Phi_u} n + (0,(m-n)*\dersz{\Phi_u},0) + (1,n,0) + (0,m-n,0) \\
        &= \dermeasaux \Phi n + (0,(m-n)*(\dersz{\Phi_t}+\dersz{\Phi_u}+1),0) \\
       &= \dermeasaux \Phi n + (0,(m-n)*\dersz{\Phi},0)
      \end{align*}
    \item If $\Phi = \inferrule{
        \dem {\Phi_t} {\Gamma; x:\MM}{t}{\sigma} \and
        \dem {\Phi_u} {\Delta}{u}{\MM}}{
      \seq{\Gamma \inter   \Delta}{t \cut x u}{\tau}}$, then
      \begin{align*}
        \dermeasaux \Phi m
        &= \dermeasaux {\Phi_t} m + \dermeasaux {\Phi_u}{m + \lv x t +
          \isES{\cut x u}} \\
        &=_{\ih} \dermeasaux {\Phi_t} n + (0,(m-n)*\dersz{\Phi_t},0) + \dermeasaux
                    {\Phi_u}{n + \lv x t + \isES{\cut x u}} \\
                    & \sep \sep + (0,(m-n)*\dersz{\Phi_u},0) \\
        &= \dermeasaux \Phi n + (0,(m-n)*(\dersz{\Phi_t}+\dersz{\Phi_u}),0) \\
        &= \dermeasaux \Phi n + (0,(m-n)*\dersz{\Phi},0)
      \qedhere
      \end{align*}
  \end{itemize}
\end{proof}

\partialsub*
\label{p:partial-substitution}
\begin{proof}
  By induction on $\Phi$.
  \begin{itemize}
    \item If $\Phi = \inferrule{ }{\seq {x: \mult\sigma} x \sigma}$,
      then $\MN = \mult\sigma$ and $\Psi = \dem {\Phi_u} \Delta u {\mult\sigma}$.
      So, $\dermeasaux \Psi m
      = \dermeasaux {\Phi_u} m
      = (0, 0, 1) + \dermeasaux{\Phi_u}{m + 0} - (0, 0, 1)
      = \dermeasaux \Phi m + \dermeasaux{\Phi_u}{m + \lv \ec C} - (0,0,\msetsz\MN)$.

    \item If $\Phi = \inferrule{
        \dem {\Phi'} {\Gamma; x:\MM; y:\MM_y}{\ctxnc{\fc'} x} \tau}{
        \seq{\Gamma; x:\MM}{\lam y {\ctxnc{\fc'} x}}{\MM_y \ft \tau}}$,
      then we can assume by $\alpha$-conversion
      that $y \notin \dom\Delta$ so that $(\Gam; y:\MM_y)\inter\Delta =
      \Gam \inter \Delta;y:\MM_y$. By using the \ih\ we can then construct
      $ \Psi = \inferrule{
        \dem {\Psi'} {\Gamma \inter \Delta; x:\MM \setminus \MN;
      y:\MM_y}{\ctxnc{\fc'} u}\tau}{
        \seq{\Gamma \inter \Delta; x:\MM \setminus \MN}{\lam y \ctxnc{\fc'} u}{\MM_y \ft \tau}}$.
      Then,
      \begin{align*}
        \dermeasaux \Psi m &= \dermeasaux {\Psi'} m + (1,m,0) \\
        &=_{\ih} \dermeasaux {\Phi'} m + \dermeasaux {\Phi_u}{m + \lv \ec {\fc'}} - (0,0,\msetsz\MN) + (1,m,0) \\
        &= \dermeasaux \Phi m + \dermeasaux {\Phi_u}{m + \lv \ec {\lam y {\fc'}}} - (0,0,\msetsz\MN)
      \end{align*}

    \item If $\Phi = \inferrule{ }{\seq \emptyset {\lam y \ctxnc{\fc'} x} \ans}$,
      we can build $\Psi = \inferrule{ }{\seq \emptyset {\lam y \ctxnc{\fc'} u}
      \ans}$.
      In particular, we have $\MM = \MN = \emult$, and thus
      $\Phi_u$ comes from the application of the $\many$ rule to $0$ premises, so that $\dermeasaux{\Phi_u}{m +\lv\ec\fc} = (0,0,0)$.
      We have $\dermeasaux \Phi m = \dermeasaux \Psi m = \dermeasaux \Psi m + (0,0,0) - (0,0,0)$.

    \item If $\Phi = \inferrule{
        \dem {\Phi_1} {\Gamma_1; x:\MM_1}{\ctxnc{\fc'} x}{\MM' \ft \sigma} \and
        \dem {\Phi_2} {\Gamma_2; x:\MM_2} t {\MM'}}{
        \seq{\Gamma_1 \inter \Gamma_2; x:\MM}{\ctxnc{\fc'} x t}\sigma}$,
      by \ih\ there is $\MN \sqsubseteq \MM_1$ such that we can construct
      \[ \Psi = \inferrule{
        \dem {\Psi_1} {\Gamma_1 \inter \Delta; x:\MM_1 \setminus \MN}{\ctxnc{\fc'}
        u}{\MM' \ft \sigma} \and
        \dem {\Phi_2} {\Gamma_2;x:\MM_2} t {\MM'}}{
        \seq{\Gamma_1 \inter \Gamma_2 \inter \Delta;x:\MM \setminus\MN}{\ctxnc{\fc'} u t}\sigma} \]
      because $\MM \setminus \MN = \MM_1 \setminus \MN \sqcup \MM_2$.
      We have
      \begin{align*}
        \dermeasaux \Psi m &= \dermeasaux {\Psi_1} m + \dermeasaux {\Phi_2} m + (1,m,0) \\
        &=_\ih \dermeasaux{\Phi_1} m + \dermeasaux{\Phi_u}{m + \lv \ec {\fc'}} - (0,0,\msetsz\MN) + \dermeasaux{\Phi_2} m + (1,m,0) \\
        &= \dermeasaux \Phi m + \dermeasaux{\Phi_u}{m + \lv \ec {\fc' t}} - (0,0,\msetsz\MN)
      \end{align*}

    \item If \[
        \Phi = \inferrule{
          \dem {\Phi_1} {\Gamma_1;x:\MM_1} t {\mult{\tau_i}_{\iI} \ft \sigma} \and
          \inferrule{
            \left(\dem{\Phi_2^i} {\Gamma_2^i; x:\MM_2^i}{\ctxnc{\fc'} x}{\tau_i}\right)_\iI}{
            \seq{\Gamma_2; x:\MM_2}{\ctxnc{\fc'} x}{\mult{\tau_i}_{\iI}}}}{
          \seq{\Gamma_1 \inter \Gamma_2; x:\MM}{t \ctxnc{\fc'} x} \sigma} \]
      where $\MM_2 = \sqcup_\iI \MM_2^i$ and $\Gam_2 = \inter_{\iI} \Gam_2^i$.
      \autoref{l:split} gives $\dem{\Phi_u^i} {\Delta^i} u {\MN^i}$
      such that $\MN^i \sqsubseteq \MM_2^i$ for all
      $\iI$ and $\MN = \sqcup_\iI \MN^i$. Moreover, $\dermeasaux {\Phi_u} m  = \sum_{\iI} \dermeasaux {\Phi^i_u} m$.
      By using the \ih\ we can construct
      \[ \Psi = \inferrule{
        \dem {\Phi_1} {\Gamma_1;x:\MM_1} t {\mult{\tau_i}_{\iI} \ft \sigma} \and
        \inferrule{
          \left( \dem{\Psi_2^i} {\Gamma_2^i + \Delta^i;x:\MM_2^i \setminus
        \MN^i}{\ctxnc{\fc'} u}{\tau_i}\right)_\iI}{
          \seq{\Gamma_2 \inter \Delta; x : \MM_2 \setminus \MN}{\ctxnc{\fc'} u}{\mult{\tau_i}_{\iI}}}}{
        \seq{\Gamma_1 \inter \Gamma_2 \inter \Delta; x:\MM \setminus \MN}{t \ctxnc{\fc'}
    u}\sigma} \]
      where $\MM \setminus \MN = \MM_1 \sqcup \MM_2 \setminus \MN$.
      We have
      \begin{align*}
        \dermeasaux \Psi m &= \dermeasaux {\Phi_1} m + \sum_\iI \dermeasaux {\Psi_2^i} m + (1,m,0) \\
        &=_\ih \dermeasaux{\Phi_1} m + \sum_\iI \left(\dermeasaux{\Phi_2^i} m + \dermeasaux{\Phi_u^i}{m + \lv \ec {\fc'}} - (0,0,\msetsz{\MN^i})\right) + (1,m,0) \\
        &= \dermeasaux \Phi m + \dermeasaux{\Phi_u}{m + \lv \ec {t \fc'}} - (0,0,\msetsz\MN)
      \end{align*}

    \item If $\Phi = \inferrule{
        \dem{\Phi_1} {\Gamma_1; x:\MM_1; y:\MM_y}{\ctxnc{\fc'} x} \sigma \and
        \dem {\Phi_2} {\Gamma_2;x:\MM_2} t {\MM_y}}{
        \seq{\Gamma_1 \inter \Gamma_2; x:\MM}{\ctxnc{\fc'} x \cut y t}\sigma}$,
      then we can assume by $\alpha$-conversion
      that $x \notin \fv{u}$ and $y \notin
      \fv{u}$ thus, by the Relevance
      \autoref{l:relevance}, $y \notin \dom{\Del}$ so that
      in particular $(\Gam_1; y:\MM_y)\inter \Del =
      \Gam_1\inter \Del;y:\MM_y$.
      By using the \ih\ we can then construct
      \[ \Psi = \inferrule{
        \dem {\Psi_1} {\Gamma_1 \inter \Delta; x:\MM_1 \setminus \MN;
        y:\MM_y}{\ctxnc{\fc'} u} \sigma \and
        \dem {\Phi_2} {\Gamma_2;x:\MM_2} t {\MM_y}}{
        \seq{\Gamma_1 \inter \Gamma_2 \inter \Delta; x:\MM \setminus \MN}{\ctxnc{\fc'} u \cut y t}\sigma} \]
      because $\MM \setminus \MN = \MM_1 \setminus \MN \sqcup \MM_2$.
      We have:
      \begin{align*}
        \dermeasaux \Psi m &= \dermeasaux{\Psi_1} m + \dermeasaux{\Phi_2}{m +
        \lv y {\ctxnc{\fc'} u} + \isES{\cut y t}} \\
        &=_\ih \dermeasaux{\Phi_1} m + \dermeasaux{\Phi_u}{m + \lv \ec {\fc'}}
        - (0,0,\msetsz\MN)\\
        &\qquad + \dermeasaux{\Phi_2}{m + \lv y {\ctx{\fc'} x}} +
        \isES{\cut y t} \\
        &= \dermeasaux \Phi m + \dermeasaux{\Phi_u}{m + \lv \ec {\fc' \cut y t}} - (0,0,\msetsz\MN)
      \end{align*}

    \item If \[
        \Phi = \inferrule{
          \dem {\Phi_1} {\Gamma_1;x:\MM_1;y:\mult{\tau_i}_\iI} t \sigma \and
          \inferrule{
            \left(\dem{\Phi_2^i} {\Gamma_2^i; x:\MM_2^i}{\ctxnc{\fc'} x}{\tau_i}\right)_\iI}{
            \seq{\Gamma_2; x:\MM_2}{\ctxnc{\fc'} x}{\mult{\tau_i}_\iI}}}{
          \seq{\Gamma_1 \inter \Gamma_2; x:\MM}{t \cut y {\ctxnc{\fc'} x}}\sigma} \]
      where $\MM = \MM_1 \sqcup \MM_2$, $\MM_2 = \sqcup_\iI \MM_2^i$ and $\Gam_2 = \inter_{\iI} \Gam_2^i$.
      Lem.~\ref{l:split} gives $\dem {\Phi_u^i} {\Delta^i} u {\MN^i}$
      for all $\iI$.  Moreover, $\dermeasaux {\Phi_u} m  = \sum_{\iI} \dermeasaux {\Phi^i_u} m$.
      By using the \ih\ we can construct
      \[ \Psi = \inferrule{
        \dem {\Phi_1} {\Gamma_1;x:\MM_1;y:\mult{\tau_i}_\iI} t \sigma \and
        \inferrule{
          \left(\dem{\Psi_2^i} {\Gamma_2^i \inter \Delta^i;x:\MM_2^i \setminus
        \MN^i}{\ctxnc{\fc'} u}{\tau_i}\right)_\iI}{
          \seq{\Gamma_2 \inter \Delta; x:\MM_2 \setminus \MN}{\ctxnc{\fc'} u}{ \mult{\tau_i}_\iI}}}{
        \seq{\Gamma_1 \inter \Gamma_2 \inter   \Delta; x:\MM \setminus \MN}{t \cut y
    {\ctxnc{\fc'} u}}\sigma} \]
      because $\MM \setminus \MN = \MM_1 \sqcup \MM_2 \setminus \MN$,
      where $\MN = \sqcup_\iI \MN^i$.
      We have
      \begin{align*}
        \dermeasaux \Psi m &= \dermeasaux {\Phi_1} m + \sum_\iI \dermeasaux
        {\Psi_2^i}{m + \lv y t + \isES{\cut y {\ctxnc{\fc'} u}}} \\
        &=_\ih \dermeasaux{\Phi_1} m + \sum_\iI (
        \dermeasaux{\Phi_2^i}{m + \lv y t + \isES{\cut y {\ctxnc{\fc'}
        u}}}\\
        &\qquad + \dermeasaux{\Phi_u^i}{m
      + \lv y t + \isES{\cut y {\ctxnc{\fc'} u}} + \lv \ec {\fc'}} - (0,0,\msetsz{\MN^i})) \\
        &= \dermeasaux \Phi m + \dermeasaux{\Phi_u}{m + \lv \ec {t \cut y {\fc'}}} - (0,0,\msetsz\MN)
      \end{align*}
      Notice that  a special case is when $y \notin \fv t$.
      Then, $I = \emptyset$, $\Gamma = \Gamma_1$, $\MN = \emult$ and $\dem
      {\Phi_u} \emptyset u \emult$ is made only of a nullary $\many$ rule. Hence, $\Phi = \Phi_1 = \Psi$.
      \qedhere
  \end{itemize}
\end{proof}

\srpermlr*
\label{p:sr_permlr}
\begin{proof}
  Let $t_0 = \ctx\fc{t_0'}$ and $t_1 = \ctx\fc{t_1'}$,
  where $t_0' \rew\permlr t_1'$ is a root step.
  We reason by induction on $\fc$.
  We first consider the base cases where $\fc = \ec$.
  \begin{enumerate}
    \item $t_0' = \lam y {t \cut x u} \rrule\permlr (\lam y t) \cut x u =t_1'$,
      where $y \notin \fv u$.
      There are two possible typing derivations.
      \begin{enumerate}
        \item The typing derivation is of the form
          \begin{mathpar}
            \Phi =
            \inferrule*[Right=\ruleCutR]{
              \dem {\Phi_t} {\Gamma';y:\MN;x:\MM} t \tau
            \and \dem{\Phi_u} {\Delta_u} u \MM}{
            \inferrule*[right=\ruleAbsR]{
            \seq{\Gamma' \inter \Delta_u;y:\MN}{t \cut x u}{\tau}}{
        \seq{\Gamma' \inter \Delta_u}{\lam y {t \cut x u}}{\MN \ft \tau}}}
      \end{mathpar}
      We construct the following derivation.
      \begin{mathpar}
        \Psi = \inferrule*[right=\ruleCutR,vcenter]{
          \inferrule*[right=\ruleAbsR]{
          \dem {\Phi_t} {\Gamma';y:\MN;x:\MM} t \tau}{
        \seq{\Gamma';x:\MM}{\lam y t}{\MN \ft \tau}} \and
      \dem {\Phi_u} {\Delta_u} u \MM}{
    \seq{\Gamma' \inter \Delta_u}{(\lam y t) \cut x u}{\MN \ft \tau}}
  \end{mathpar}
  Moreover,
  \[\begin{array}{lll}
    \dermeasaux  \Phi m & = &
    \dermeasaux{\Phi_t} m  + \dermeasaux{\Phi_u}{m+\lv x {t}+\isES{\cut x u}}+ (1,m,0) \\
                        & = &
                        \dermeasaux{\Phi_t} m + (1,m,0) +
                        \dermeasaux{\Phi_u}{m+\lv x {\lam y
                        t}+\isES{\cut x u}} \\
                        & = & \dermeasaux\Psi m.
  \end{array} \]
\item The typing derivation is of the form
  \begin{mathpar}
    \Phi = \inferrule*[right=\ruleAnsR,vcenter]{ }{
    \seq{}{\lam y {t \cut x u}}\ans}
  \end{mathpar}
  We construct the following derivation that has the same measure.
  \begin{mathpar}
    \Psi = \inferrule*[right=\ruleCutR,vcenter]{
      \inferrule*[right=\ruleAnsR]{ }{\seq{}{\lam y t}{ \ans}}
    \and \inferrule*[Right=\many]{ }{\seq{}{u}{ \emult}}}{
  \seq{}{(\lam y t) \cut x u}{ \ans}}
\end{mathpar}
\end{enumerate}

\item $t_0' = t \cut x u s \rrule\permlr (ts) \cut x u = t_1'$,
  where $x \notin \fv{s}$.
  The typing derivation is of the form
  \begin{mathpar}
    \Phi = \inferrule*[right=\ruleAppR,vcenter]{
      \inferrule*[right=\ruleCutR]{
        \dem {\Phi_t} {\Gamma';x:\MM} t {\MN \ft \sigma} \and
      \dem {\Phi_u} {\Delta_u} u \MM}{
    \seq{\Gamma' \inter \Delta_u}{t \cut x u}{\MN \ft \sigma}}
  \and \dem {\Phi_s} {\Delta_s} s \MN}{
\seq{\Gamma' \inter\Delta_u \inter\Delta_s}{t \cut x u s}{ \sigma}}
\end{mathpar}
We construct the following derivation.
\begin{mathpar}
  \Psi = \inferrule*[right=\ruleCutR,vcenter]{
    \inferrule*[right=\ruleAppR]{
      \dem {\Phi_t} {\Gamma';x:\MM} t {\MN \ft \sigma}
    \and \dem {\Phi_s} {\Delta_s} s \MN}{
  \seq{\Gamma' \inter\Delta_s;x:\MM}{ts}{\sigma}}
\and \dem {\Phi_u} {\Delta_u} u \MM}{
{\Gamma' \inter \Delta_u \inter\Delta_s}{(ts)\cut x u: \sigma}}
\end{mathpar}
Moreover, since $\lv x {t} = \lv x {ts}$,
\[ \dermeasaux \Phi m =
  \dermeasaux{\Phi_t} m + \dermeasaux{\Phi_s} m + (1,m,0) +
\dermeasaux{\Phi_u}{m+\lv x {ts} + \isES{\cut x u}} = \dermeasaux \Psi m. \]

\item $t_0' = t s \cut x u  \rrule\permlr (ts) \cut x u = t_1'$,
  where $x \notin \fv{t}$.
  Let \[ \Phi_{s \cut x u} =
    \inferrule*[Right=\many,vcenter]{
      \left( \inferrule*[right=\ruleCutR,vcenter]{
          \dem {\Phi^i_s} {\Delta_s^i;x:\MM_i} s {\rho_i}
          \and
          \inferrule*[Right=\many]{
            \left(\dem {\Phi_u^{i,j}} {\Delta_u^{i,j}} u {\delta_j}
          \right)_{\jJi}}{
      \seq{\Delta_u^i}{u}{\MM_i}}}{
    \seq{\Delta_u^i \inter \Delta_s^i}{s \cut x u}{\rho_i}}
\right)_{\iI}}{
\seq{\Delta_u \inter \Delta_s}{s \cut x u}{\MN}} \]
  The typing derivation $\Phi$ is of the form
  \begin{mathpar}
    \inferrule*[right=\ruleAppR,vcenter]{
      \dem {\Phi_t} {\Gamma'} t {\MN \ft \sigma}
      \and
    \dem{\Phi_{s \cut x u}}{\Delta_u \inter \Delta_s}{s \cut x u } \MN}{
\seq{\Gamma' \inter \Delta_u \inter \Delta_s}{t s \cut x u}{\sigma}}
\end{mathpar}
where $\MM_i=\mult{\delta_j}_{\jJi}$,
$\MN=\mult{\rho_i}_{\iI}$,
$\Delta_u^i=\inter_{\jJi} \Delta_u^{i,j}$,
$\Delta_u=\inter_{\iI} \Delta_u^i $,
and $\Delta_s=\inter_{\iI} \Delta_s^i$.

Now, let
\begin{mathpar}
  \Phi_s = \inferrule*[right=\ruleAppR,vcenter]{
    \inferrule*[Right=\many]{
    \left( \dem{\Phi_s^i} {\Delta_s^i; x:\MM_i} s {\rho_i} \right)_{\iI}}{
\seq{\Delta_s; x:\MM}{s}{\MN}}}{
\seq{\Gamma' \inter \Delta_s; x:\MM}{ts}{\sigma}}
\and \Phi_u =
\inferrule*[Right=\many,vcenter]{
\left( \dem{\Phi_u^{i,j}} {\Delta_u^{i,j}} u {\delta_j} \right)_{\jJi, \iI}}{
\seq{\Delta_u}{u}{\MM}}
\end{mathpar}
We construct the following derivation $\Psi$.
\[ \inferrule*[right=\ruleCutR]{
    \dem {\Phi_t} {\Gamma'} t {\MN \ft \sigma}
    \and \dem{\Phi_s}{\Gamma' \inter \Delta_s; x:\MM}{ts}{\sigma}
  \and \dem{\Phi_u}{\Delta_u} u \MM
}{
\seq{\Gamma' \inter \Delta_u \inter \Delta_s}{(ts) \cut x u }{ \sigma}} \]
where $\MM=\multunion_{\iI}\MM_i$, so that $\MM=\mult{\delta_j}_{\jJi, \iI}$.
Moreover, because $\lv x s = \lv x {ts}$,
\begin{align*}
  \dermeasaux \Phi m
  &= \dermeasaux{\Phi_t} m + (1,m,0)\\
  &+ \sum_{\iI} \left(\dermeasaux{\Phi_s^i} m
    + \sum_{\jJi}
  \dermeasaux{\Phi_u^{i,j}}{m + \lv x s + \isES{\cut x u}} \right)\\
  &= \dermeasaux \Psi m
\end{align*}

\item $t_0' = t \cut x {u \cut y s} \rrule\permlr t \cut x u \cut y s =t_1'$,
  where $y \notin \fv{t}$.
  Let \[ \Phi_{u \cut y s} =
      \inferrule*[Right=\many,vcenter]{
        \left( \inferrule*[right=\ruleCutR]{
            \dem {\Phi_u^i} {\Delta_u^i;y:\MN_i} u {\rho_i}
            \and \inferrule*[Right=\many]{
              \left(\dem{\Phi_s^{i,j}} {\Delta_s^{i,j}} s {\delta_j}
            \right)_{\jJi}}{
        \seq{\Delta_s^i}{s}{\MN_i}}}{
      \seq{\Delta_u^i \inter \Delta_s^i}{u \cut y s}{\rho_i}}
    \right)_{\iI}
}{\seq{\Delta_u \inter \Delta_s}{u \cut y s}{\MM}} \]
  The typing derivation $\Phi$ is of the form
  \begin{mathpar}
    \inferrule*[right=\ruleCutR,vcenter]{
      \dem {\Phi_t} {\Gamma'; x:\MM} t \sigma
      \and \dem{\Phi_{u \cut y s}}{\Delta_u \inter \Delta_s}{u \cut y s}{\MM}
}{\seq{\Gamma'\inter\Delta_u\inter\Delta_s}{t \cut x {u \cut y s}}{\sigma}}
\end{mathpar}
where $\MM = \mult{\rho_i}_{\iI}$, $\MN_i= \mult{\delta_j}_{\jJi}$,
$\Delta_u = \inter_{\iI} \Delta_u^i$,  $\Delta_s^i =  \inter_{\jJi} \Delta_s^{i,j}$, and
$\Delta_s = \inter_{\iI} \Delta_s^i$.

Now, let \[ \Phi_{t \cut x u} =
    \inferrule*[right=\ruleCutR,vcenter]{
      \dem {\Phi_t} {\Gamma'; x:\MM} t \sigma
      \and \inferrule*[Right=\many]{
        \left(
          \dem {\Phi_u^i} {\Delta_u^i;y:\MN_i} u {\rho_i}
      \right)_{\iI}}{
  \seq{\Delta_u;y:\MN}{u}{\MM}}}{
\seq{\Gamma' \inter \Delta_u; y:\MN}{t \cut x u}{\sigma}} \]
We then construct the following derivation $\Psi$.
\begin{mathpar}
  \inferrule*[right=\ruleCutR]{
    \dem{\Phi_{t \cut x u}}{\Gamma' \inter \Delta_u; y:\MN}{t \cut x u}{\sigma}
\and \inferrule*[Right=\many]{
  \left(
    \dem {\Phi_s^{i,j}} {\Delta_s^{i,j}} s {\delta_j}
\right)_{\jJi, \iI}}{
\seq{\Delta_s}{s}{\MN}}}{
\seq{\Gamma' \inter \Delta_u \inter \Delta_s}{t \cut x u \cut y s}{\sigma}}
\end{mathpar}
where $\MN = \multunion_{\iI} \MN_i$, so that
$\MN = \mult{\delta_j}_{\jJi, \iI}$. Moreover, because $y \notin \fv t$, we have that
$\lv y {t \cut x u} = \lv x t + \lv y u + \isES{\cut x u}$ if $y \in
\fv u$, and 
$\lv y {t \cut x u} = 0$ otherwise.
Now, we show that $\dermeasaux{\Phi_s^{i,j}}{m+\lv x t + \isES{\cut x u} + \lv y
u+ \isES{\cut y s}} =
\dermeasaux{\Phi_s^{i,j}}{m + \lv y {t \cut x u} + \isES{\cut y s}}$.
If $y \in \fv u$, this is immediate.
Otherwise, by the Relevance Lemma~\ref{l:relevance} we
have $J_i = \emult$ for any $i$ thus $s$ is not typed,
so that both measures are equal to (0,0,0).
Then, \begin{align*}
  \dermeasaux \Phi m
        &= \dermeasaux{\Phi_t} m
        + \sum_{\iI} \dermeasaux{\Phi_u^i}{m+\lv
          x t+\isES{\cut x u}}\\
        &\qquad + \sum_{\iI} \sum_{\jJi} \dermeasaux{\Phi_s^{i,j}}{m+\lv x t + \isES{\cut x
      u} + \lv y u + \isES{\cut y s}} \\
        & = \dermeasaux{\Phi_t} m + \sum_{\iI}
        \dermeasaux{\Phi_u^i}{m+\lv x t+\isES{\cut x u}}\\
        &\qquad + \sum_{\iI} \sum_{\jJi} \dermeasaux{\Phi_s^{i,j}}{m+\lv y {t \cut x
        u}+\isES{\cut y s}} \\
        &= \dermeasaux \Psi m
\end{align*}
\end{enumerate}

Now, we analyse all the
inductive cases:
\begin{enumerate}
  \item If $\fc = \lam x {\fc'}$, then
    we have $\sigma = \MM \ft \tau$ and $\dem {\Phi'}
    {\Gamma;x:\MM}{\ctx{\fc'}{o}} \tau$.
    By the \ih\ there is $\dem {\Psi'} {\Gamma;x:\MM}{\ctx{\fc'}{o'}} \tau$
    and therefore $\dem {\Psi} \Gamma {\lam x {\ctx{\fc'}{o'}}} \tau$.
    Moreover, $\dermeasaux \Phi m = \dermeasaux{\Phi'} m + (1,m,0) =_{\ih} \dermeasaux{\Psi'} m + (1,m,0) = \dermeasaux \Psi m$.
  \item If $\fc = \fc' u$, then
    we have $\dem {\Phi'} {\Gamma'}{\ctx{\fc'}o}{\MN \ft \sigma}$
    and $\dem {\Phi_u} \Delta u \MN$.
    By the \ih\  there is $\dem {\Psi'} {\Gamma'}{\ctx{\fc'}{o'}}{\MN \ft \sigma}$,
    so $\dem \Psi {\Gamma'\inter\Delta}{\ctx{\fc'}{o'}u} \sigma$.
    Moreover, $\dermeasaux \Phi m = \dermeasaux{\Phi'} m + \dermeasaux{\Phi_u} m + (1,m,0) =_{\ih} \dermeasaux{\Psi'} m + \dermeasaux{\Phi_u} m + (1,m,0) = \dermeasaux \Psi m$.
  \item If $\fc = u \fc'$, the case is similar.
  \item If $\fc = \fc'  \cut x u$, then
    we have $\dem {\Phi'} {\Gamma';x:\MM}{\ctx{\fc'} o} \sigma$
    and $\dem {\Phi_u} \Delta u \MM$.
    By the \ih\   there is
    $\dem {\Psi'} {\Gamma';x:\MM}{\ctx{\fc'}{o'}} \sigma$,
    so $\dem \Psi {\Gamma'\inter\Delta}{\ctx{\fc'}{o'} \cut x u} \sigma$.
    Moreover, $\dermeasaux \Phi m
    = \dermeasaux{\Phi'} m + \dermeasaux{\Phi_u}{m + \lv x t + \isES{\cut x
    u}}
    =_{\ih} \dermeasaux{\Psi'} m + \dermeasaux{\Phi_u}{m + \lv x t +
    \isES{\cut x u}}
    = \dermeasaux \Psi m$.
  \item If $\fc = u \cut x {\fc'}$, then
    we have $\dem {\Phi_u} {\Delta;x:\MM} u \sigma$
    and $\dem {\Phi'} {\Gamma'}{\ctx{\fc'}o} \MM$.
    By the \ih\  there is
    $\dem {\Psi'} {\Gamma'}{\ctx{\fc'}{o'}}{\MM}$,
    so $\dem \Psi {\Gamma'\inter \Delta}{u \cut x {\ctx{\fc'}{o'}}} \sigma$.
    Moreover, $\dermeasaux \Phi m
    = \dermeasaux{\Phi_u} m + \dermeasaux{\Phi'}{m + \lv x u + \isES{\cut x
    u}}
    =_{\ih} \dermeasaux{\Phi_u} m + \dermeasaux{\Psi'}{m + \lv x u +
    \isES{\cut x u}}
    = \dermeasaux \Psi m$.
    \qedhere
\end{enumerate}
\end{proof}

\ndlrnormtyp*
\label{p:ndlr-norm-typ}
\begin{proof}
  First, we show that if $t$ is an answer $\ctx\lc{\lam x p}$, we can type it with type $\ans$ and $\Gamma = \emptyset$.
  We reason by induction on $\lc$. If $\lc = \ec$, this is immediate.
  Otherwise, using the induction hypothesis, we build:
  \[ \inferrule*[right=\ruleCutR]{
    \seq \emptyset {\ctx\lc{\lam x p}}\ans
    \and \inferrule*[Right=\many]{ }{\seq\emptyset u \emult}}{
    \seq \emptyset {\ctx\lc{\lam x p} \cut y u} \ans} \]
  The statement is then trivial since $\Gamma = \emptyset$.
  For neutral terms, we use induction on $\neutg$ with a stronger hypothesis:
  there exists a derivation for any given type $\tau$.
  \begin{itemize}
    \item $t = x$. We can build $\dem \Phi {x:\mult\tau} x \tau$. Note that $x \in \ndv t$.
    \item $t = t'u$, where $t' \in \neutg$.
      By the \ih\ there is a derivation $\dem {\Phi'} \Gamma {t'} {\emult \ft \tau}$ verifying the statement.
      We then build:
      \[ \inferrule*[right=\ruleAppR]{
        \dem {\Phi'} \Gamma {t'} {\emult \ft \tau}
        \and \inferrule*[Right=\many]{ }{\seq \emptyset u \emult}}{
        \seq{\Gamma}{t'u} \tau} \]
      The statement holds by the \ih\ because $\ndv t = \ndv{t'}$.
    \item $t = t' \cut x u$, where $t' \in \neutg$.
      By the \ih\ there is a derivation $\dem {\Phi'} {\Gamma_{t'}} {t'} \tau$
      verifying the statement.
      Let $\Gamma_{t'} = \Gamma'; x:\mult{\sigma_i}_\iI$.
      There are two cases.
      \begin{itemize}
        \item If $x \notin \ndv{t'}$, then by the \ih\ $I = \emult$.
          We can then build the following derivation.
          \[ \inferrule*[right=\ruleCutR]{
            \dem {\Phi'} {\Gamma'}{t': \tau}
            \and \inferrule*[Right=\many]{ }{
              \seq \emptyset u \emult}}{
            \seq{\Gamma'}{t' \cut x u} \tau} \]
          The  property holds for $\Gamma = \Gamma'$ because $\ndv t = \ndv{t'}$.
        \item Otherwise, $t = t' \es x u$, and $u \in \neutg$. We apply
          the \ih\ on u.
          There are derivations $\dem {\Phi_u^i} {\Delta_i} u {\sigma_i}$.
          We take $\Gamma = \Gamma_{t'} \inter_\iI \Delta_i$ and we  build:
          \[ \inferrule*[right=\ruleCutR]{
            \dem {\Phi'} {\Gamma_{t'}}{t'} \tau
            \and \inferrule*[Right=\many]{
              (\dem {\Phi_u^i} {\Delta_i} u {\sigma_i})}{
              \seq{\inter_\iI \Delta_i} u {\mult{\sigma_i}_\iI}}}{
            \seq \Gamma {t' \es x u}\tau} \]
          where $\Gamma = \Gamma' \inter   _\iI \Delta_i$.
          Moreover, $\ndv t = (\ndv{t'} \setminus x) \cup \ndv u$ so the second property holds on $\Gamma$ by the two induction hypothesis.
          \qedhere
      \end{itemize}
  \end{itemize}
\end{proof}

\partialasub*
\label{p:partial-anti-substitution}
\begin{proof}
  By induction on the structure of $\fc$.
  \begin{itemize}
    \item If $\fc = \ec$ then the property trivially holds taking 
      $\Gamma' = \emptyset$, $\Delta = \Gamma$, $\MM = \mult{\sig}$,
      $\dem {\Phi'} {x:\mult{\sig}} x \sigma$ and $\Phi_u = \Phi$.

    \item  If $\fc = \lam y \fc'$ then $y \notin \fv u$ and by $\alpha$-conversion we can assume that $x \neq y$. There are two cases:
      \begin{enumerate}
        \item If $\Phi = \infer{
          \dem {\Phi_0} {\Gamma; y:\MM_y}{\ctxnc{\fc'} u}\tau}{
        \seq \Gamma {\lam y {\ctxnc{\fc'} u}}{\MM_y \ft \tau}}$ then by
        \ih\ there are
        $\Gamma', \Delta, \MM, \Phi'_0$ and $\Phi_u$ such that
        $\Gamma; y:\MM_y = \Gamma'_0 \inter \Del$, $\dem {\Phi'_0} {\Gamma'_0 \inter
        x:\MM}{\ctxnc{\fc'} x} \tau$ and $\dem {\Phi_u} \Delta u \MM$.
        By the Relevance \autoref{l:relevance} $y \notin \dom\Del$ thus $\Gamma'_0 = \Gamma'; y:\MM_y$.  Therefore,
        $\Gamma'_0 \inter
        x:\MM = (\Gamma' \inter x:\MM); y:\MM_y$ and
        \[
          \Phi' = \infer{\dem {\Phi'_0}
          {(\Gamma' \inter x:\MM); y:\MM_y}{\ctxnc{\fc'} x} \tau}{\seq{\Gamma' \inter
        x:\MM}{\lam y {\ctxnc{\fc'} x}} {\MM_y \ft \tau}}
      \]

    \item If $\Phi = \infer{ }{\seq\emptyset{\lam {y} \ctxnc{\fc'} u}\ans}$ then
      taking $\Gamma', \Delta = \emptyset$, $\MM = \emult$ and $\dem {\Phi_u}
      \emptyset u \emult$ we have
      \[
        \Phi' = \infer{ }{\seq\emptyset{\lam {y}
        \ctxnc{\fc'} x}{\ans}}
      \]
  \end{enumerate} 

\item  If $\fc = \fc' t$ then $\Phi = \infer{
    \dem {\Phi_1} {\Gamma_1}{\ctxnc{\fc'} u}{\MM' \ft \sigma} \and
  \dem {\Phi_2} {\Gamma_2} t {\MM'}}{
\seq{\Gamma_1 \inter \Gamma_2}{\ctxnc{\fc'} u t}\sigma}$, where $\Gamma = \Gamma_1
\inter \Gamma_2$. By \ih\ there are
$\Gamma'_1, \Delta, \MM, \Phi'_1$ and $\Phi_u$ such that
$\Gamma_1 = \Gamma'_1 \inter \Del$,
$\dem {\Phi'_1} {\Gamma'_1 \inter
x:\MM}{\ctxnc{\fc'} x}{\MM' \ft \sigma}$ and $\dem {\Phi_u} \Delta u \MM$.
Therefore, taking $\Gamma' = \Gamma'_1 \inter \Gamma_2$ we have
\[
  \Phi' = \infer{
    \dem {\Phi'_1} {\Gamma'_1 + x:\MM}{\ctxnc{\fc'} x}{\MM' \ft \sigma} \and
  \dem {\Phi_2} {\Gamma_2} t {\MM'}}{
\seq{(\Gamma'_1 \inter x:\MM) \inter \Gamma_2}{\ctxnc{\fc'} x t}\sigma}
\]
where $(\Gamma'_1 \inter x:\MM) \inter \Gamma_2 = \Gamma' \inter x:\MM$. 

\item  If $\fc =  t \fc' $ then $\Phi$ is of the form
  \[
    \infer{
      \dem {\Phi_1} {\Gamma_1} t {\mult{\tau_i}_{\iI} \ft \sigma} \and
      \infer{
      \left(\dem {\Phi_i} {\Gamma_i}{\ctxnc{\fc'} u}{\tau_i}\right)_\iI}{
  \seq{\Gamma_2}{\ctxnc{\fc'} u:\mult{\tau_i}_{\iI}}}}{
\seq{\Gamma_1 \inter \Gamma_2}{t \ctxnc{\fc'} u}\sigma} \] where $\Gamma_2 = \inter_{\iI} \Gamma_i$ and $\Gamma = \Gamma_1 \inter \Gamma_2$. There are two cases:
\begin{enumerate}
  \item If $I \neq \emptyset$ then by \ih\ $\exists \Gamma'_i$,
    $\exists \Del_i$, $\exists \MM_i$, $\exists \Phi'_i$,  $\exists \Phi_u^i$
    s.t. $\Gamma_i = \Gamma'_i \inter \Del_i$, $\dem {\Phi'_i} {\Gamma'_i \inter
    x:\MM_i}{\ctxnc{\fc'} x}{\tau_i}$ and $\dem {\Phi_u^i} {\Delta_i} u {\MM_i}$,
    for all $\iI$. Let $\Del = \inter_{\iI} \Del_i$ and $\MM = \sqcup_{\iI}\MM_i$
    then from Split \autoref{l:split} we have $\Phi_u = \infer{\left (\dem
    {\Phi_u^i} {\Delta_i} u {\MM_i} \right)_{\iI}}{\seq \Delta u \MM}$. Let $\Gamma'_2 = \inter_{\iI} \Gamma'_i$ then $\Gamma'_2 \inter \Del = \Gamma_2$ and $\Phi'$ is defined by 
    \[
      \infer{
        \dem {\Phi_1} {\Gamma_1} t {\mult{\tau_i}_{\iI} \ft \sigma} \and
        \infer{
        \left(\dem {\Phi'_i} {\Gamma'_i \inter x:\MM_i}{\ctxnc{\fc'} x}{\tau_i}\right)_\iI}{
    \seq{\Gamma'_2 \inter x:\MM}{\ctxnc{\fc'} x}{\mult{\tau_i}_{\iI}}}}{
\seq{(\Gamma_1 \inter \Gamma'_2) \inter x:\MM}{t \ctxnc{\fc'} x} \sigma} \]
where $\Gam' = \Gamma_1 \inter \Gamma'_2$.

\item If $I = \emptyset$, then $\mult{\tau_i}_{\iI} = \emult$, $\Gamma_2 =
  \emptyset$ and $\Gamma = \Gamma_1$. Therefore, taking $\Gamma' = \Gamma_1$,
  $\Delta = \emptyset$, $\MM = \emult$, $\Phi_u = \infer{}{\seq\emptyset u
  \emult}$,
  we have $\Gamma_1 = \Gamma_1 \inter x:\emult = \Gamma' \inter x:\emult$ and $\Gamma' \inter \Delta = \Gamma_1 \inter \emptyset = \Gamma$.
  We take
  \[
    \Phi' = \inferrule{
      \dem{\Phi_1} {\Gamma_1} t {\emult \ft \sigma} \and 
    \seq{\emptyset}{\ctxnc{\fc'} x}\emult}{
  \seq{\Gamma_1}{t \ctxnc{\fc'} x}\sigma}.
  \]
\end{enumerate}

\item  If $\fc = \fc'  \cut y t$ then $\Phi = \infer{
    \dem {\Phi_1} {\Gamma_1; y:\MM_y}{\ctxnc{\fc'} u}\sigma \and
  \dem {\Phi_2} {\Gamma_2} t {\MM_y}}{
\seq{\Gamma_1 \inter \Gamma_2}{\ctxnc{\fc'} u \cut y t}\sigma}$ where $\Gamma =
\Gamma_1 \inter \Gamma_2$. Moreover, $y \notin \fv u$ and by $\alpha$-conversion we
can assume that $x \neq y$.
By \ih\ there are $\Gamma'_1, \Delta, \MM, \Phi'_1$ and $\Phi_u$ such that $\Gamma_1; y:\MM_y =
\Gamma'_1 \inter \Del$, $\dem {\Phi'_1} {\Gamma'_1 \inter
x:\MM}{\ctxnc{\fc'} x} \sigma$ and $\dem {\Phi_u} \Delta u \MM$.
By the Relevance \autoref{l:relevance} $y \notin \dom\Del$ thus $\Gamma'_1 = \Gamma''; y:\MM_y$, $\Gamma'_1 \inter x:\MM = (\Gamma'' \inter x:\MM);y:\MM_y$ and $\Gamma'' \inter \Del = \Gamma_1$. Therefore, taking $\Gamma' = \Gamma'' \inter \Gamma_2$ we have
\[
  \Phi' = \infer{
    \dem {\Phi'_1} {(\Gamma'' \inter x:\MM);y:\MM_y}{\ctxnc{\fc'} x} \sigma \and
  \dem {\Phi_2} {\Gamma_2} t {\MM_y}}{
\seq{(\Gamma'' \inter x:\MM) \inter \Gamma_2}{\ctxnc{\fc'} x \cut y t}\sigma}
\]
where $(\Gamma'' \inter x:\MM) \inter \Gamma_2 = \Gamma' \inter x:\MM$. 

\item  If $\fc = t \cut y {\fc'}$ then $\Phi$ is of the form
  \[ \infer{
      \dem {\Phi_1} {\Gamma_1;y:\mult{\tau_i}_\iI} t \sigma \and
      \infer{
      \left(\dem{\Phi_i} {\Gamma_i}{\ctxnc{\fc'} u}{\tau_i}\right)_\iI}{
  \seq{\Gamma_2}{\ctxnc{\fc'} u}{\mult{\tau_i}_\iI}}}{
\seq{\Gamma_1 \inter \Gamma_2}{t \cut y {\ctxnc{\fc'} u}}\sigma} \]
where $\Gamma_2 = \inter_{\iI} \Gamma_i$ and $\Gamma = \Gamma_1 \inter \Gamma_2$. There are two cases:
\begin{enumerate}
  \item If $I \neq \emptyset$ then by \ih\ there are
    $\Gamma'_i, \Delta_i, \MM_i, \Phi'_i$ and $\Phi_u^i$
    such that $\Gamma_i = \Gamma'_i \inter \Del_i$, $\dem {\Phi'_i} {\Gamma'_i \inter
    x:\MM_i}{\ctxnc{\fc'} x}{\tau_i}$ and $\dem {\Phi_u^i} {\Delta_i} u {\MM_i}$,
    for all $\iI$. Let $\Del = \inter_{\iI} \Del_i$ and $\MM = \sqcup_{\iI}\MM_i$
    then from Split \autoref{l:split} we have $\Phi_u = \infer{\left (\dem
    {\Phi_u^i} {\Delta_i} u \MM_i \right)_{\iI}}{\seq \Delta u \MM}$. Let $\Gamma'_2 = \inter_{\iI} \Gamma'_i$ then $\Gamma'_2 \inter \Del = \Gamma_2$ and $\Phi'$ is defined by
    \[
      \infer{
        \dem {\Phi_1} {\Gamma_1; y:\mult{\tau_i}_{\iI}} t \sigma \and
        \infer{
        \left(\dem {\Phi'_i} {\Gamma'_i \inter x:\MM_i}{\ctxnc{\fc'} x}{\tau_i}\right)_\iI}{
    \seq{\Gamma'_2 \inter x:\MM}{\ctxnc{\fc'} x}{\mult{\tau_i}_{\iI}}}}{
\seq{(\Gamma_1 \inter \Gamma'_2) \inter x:\MM}{t \cut y {\ctxnc{\fc'} x}}\sigma} \]
where $\Gamma' = \Gamma_1 \inter \Gamma'_2$.

\item If $I = \emptyset$ then $\mult{\tau_i}_{\iI} = \emult$, $\Gamma_2 =
  \emptyset$ and $\Gamma = \Gamma_1$. Moreover, $y \notin \dom {\Gamma_1}$.
  Therefore, taking $\Gamma' = \Gamma_1$, $\Delta = \emptyset$, $\MM =
  \emult$, $\Phi_u = \infer{}{\seq\emptyset u \emult}$,
  we have $\Gamma_1 = \Gamma_1 \inter x:\emult = \Gamma' \inter x:\emult$ and
  $\Gamma' \inter \Delta = \Gamma_1 \inter    \emptyset = \Gamma$.
  We take
  \[
    \Phi' = \infer{\dem
      {\Phi_1} {\Gamma_1} t \sigma \and \seq\emptyset{\ctxnc{\fc'}
    x}\emult}{\seq{\Gamma_1}{t \cut y {\ctxnc{\fc'} x}}\sigma}
    \qedhere
  \]
\end{enumerate}
\end{itemize}
\end{proof}

\end{document}